\documentclass[reqno,11pt]{amsart}
\usepackage{setspace,tikz,xcolor,mathrsfs,listings,multicol,amssymb,mathdots}
\usepackage{rotating}
\usepackage[vcentermath]{youngtab}
\usepackage[margin=0.8in,includefoot,footskip=30pt]{geometry}
\usepackage{enumerate,enumitem}
\usepackage{booktabs}
\usepackage[centertableaux]{ytableau}
\usepackage[all,cmtip]{xy}
\usepackage[colorlinks,allcolors=black]{hyperref}
\usepackage[thinc]{esdiff}
\usetikzlibrary{arrows,matrix}
\tikzset{tab/.style={matrix of math nodes,column sep=-.35, row sep=-.35,text height=7pt,text width=7pt,align=center,inner sep=2,font=\footnotesize}}
\usepackage{tikz}\usetikzlibrary{graphs,quotes,fit,positioning,matrix,calc,decorations.markings,angles,decorations.pathmorphing,decorations.pathreplacing}

\tikzset{math mode/.style = {execute at begin node=$, execute at end node=$}}
\tikzset{arrow/.style={postaction={decorate,thick,decoration={markings,mark = at position #1 with {\arrow{>}}}}},arrow/.default=0.5}
\tikzset{invarrow/.style={postaction={decorate,thick,decoration={markings,mark = at position #1 with {\arrow{<}}}}},invarrow/.default=0.5}
\tikzset{bosonic/.style={ultra thick}}

\definecolor{aj}{rgb}{0.8, 0.8, 1.0}
\allowdisplaybreaks[1]
\tikzset{yline/.style={draw=yellow,ultra thick,line join=round}}
\tikzset{gline/.style={draw=black!50!green,line join=round}}
\tikzset{cline/.style={draw=cyan,line join=round}}
\tikzset{vline/.style={draw=black,ultra thick,line join=round}}
\tikzset{rline/.style={draw=red,ultra thick,line join=round}}
\tikzset{gpart/.style={circle,fill=black!35!green,inner sep=1.8pt}}
\tikzset{rpart/.style={circle,fill=black!35!red,inner sep=1.8pt}}
\tikzset{scol/.style={black!50!red}}
\tikzset{gcol/.style={black!35!green}}
\usepackage{cleveref}
\usepackage{tikz}
\usetikzlibrary{3d}\usetikzlibrary{graphs,quotes,fit,positioning,matrix,calc,decorations.markings,angles,decorations.pathmorphing,decorations.pathreplacing}
\tikzset{math mode/.style = {execute at begin node=$, execute at end node=$}}
\tikzset{arrow/.style={postaction={decorate,thick,decoration={markings,mark = at position #1 with {\arrow{>}}}}},arrow/.default=0.5}
\tikzset{invarrow/.style={postaction={decorate,thick,decoration={markings,mark = at position #1 with {\arrow{<}}}}},invarrow/.default=0.5}
\renewcommand\ss{\scriptstyle}
\tikzset{part/.style={circle,fill=black,inner sep=1.8pt,outer sep=3pt}}
\tikzset{hole/.style={thick,circle,draw=black,fill=white,inner sep=1.8pt,outer sep=3pt}}
\tikzset{partline/.style={draw=black!35!green,ultra thick}}
\tikzset{holeline/.style={draw=black!20!red,ultra thick}}
\tikzset{bgline/.style={dotted}}
\tikzset{bgplaq/.style={fill=lightgray!30!white}}
\tikzset{arr/.style={postaction={decorate,thick,decoration={markings,mark = at position #1 with {\arrow{>}}}}}}
\tikzset{invarr/.style={postaction={decorate,thick,decoration={markings,mark = at position #1 with {\arrow{<}}}}}}
\tikzset{6varr/.style={postaction={decorate,thick,decoration={markings,mark = at position #1 with {\arrow{triangle 60}}}}}}
\tikzset{inv6varr/.style={postaction={decorate,thick,decoration={markings,mark = at position #1 with {\arrow{triangle 60 reversed}}}}}}

\usepackage{ytableau}
\newcommand{\bra}[1]{\left\langle #1\right|}
\newcommand{\ket}[1]{\left|#1\right\rangle}
\DeclareMathOperator{\id}{id} 
 
\newcommand{\ev}{\text{ev}}

\newcommand{\Ac}{\mathcal{A}^\circ}

\tikzset{scol/.style={cyan}}
\tikzset{fcol/.style={red}}

\definecolor{darkred}{rgb}{0.7,0,0} 
\definecolor{UQgold}{RGB}{196, 158, 54} 
\definecolor{UQpurple}{RGB}{73, 7, 94} 
\usepackage{listings}
\lstdefinelanguage{Sage}[]{Python}
{morekeywords={False,sage,True},sensitive=true}
\definecolor{dblackcolor}{rgb}{0.0,0.0,0.0}
\definecolor{dbluecolor}{rgb}{0.01,0.02,0.7}
\definecolor{dgreencolor}{rgb}{0.2,0.4,0.0}
\definecolor{dgraycolor}{rgb}{0.30,0.3,0.30}
\definecolor{pinkish}{rgb}{1.0, 0.44, 0.37}

\theoremstyle{plain}
\newtheorem{thm}{Theorem}[section]
\newtheorem{lemma}[thm]{Lemma}

\newtheorem{prop}[thm]{Proposition}
\newtheorem{cor}[thm]{Corollary}
\newtheorem{dfn}[thm]{Definition}

\newtheorem{remark}[thm]{Remark}

\numberwithin{equation}{section}

\title{Shuffle algebras, lattice paths and Macdonald functions}
\author{Alexandr Garbali and Ajeeth Gunna}
\address{Alexandr Garbali, Ajeeth Gunna, School of Mathematics and Statistics, University of Melbourne, Parkville, Victoria 3010, Australia.}
\email{alexandr.garbali@unimelb.edu.au, agunna@student.unimelb.edu.au}
\begin{document}
\begin{abstract}
We consider partition functions on the $N\times N$ square lattice with the local Boltzmann weights given by the $R$-matrix of the $U_{t}(\widehat{sl}_{n+1|m})$ quantum algebra. We identify boundary states such that the square lattice can be viewed on a conic surface. The partition function $Z_N$ on this lattice computes the weighted sum over all possible closed coloured lattice paths with $n+m$ different colours: $n$ ``bosonic'' colours and $m$ ``fermionic'' colours. Each bosonic (fermionic) path of colour $i$ contributes a factor of $z_i$ ($w_i$) to the weight of the configuration. We show the following:
\begin{enumerate}[label=\roman*)]
    \item $Z_N$ is a symmetric function in the spectral parameters $x_1\dots x_N$ and generates basis elements of the commutative trigonometric Feigin--Odesskii shuffle algebra. The generating function of $Z_N$ admits a shuffle-exponential formula analogous to the Macdonald Cauchy kernel.
    \item $Z_N$ is a symmetric function in two alphabets $(z_1\dots z_n)$ and $(w_1\dots w_m)$. When $x_1\dots x_N$ are set to be equal to the box content of a skew Young diagram $\mu/\nu$ with $N$ boxes the partition function $Z_N$ reproduces the skew Macdonald function $P_{\mu/\nu}\left[w-z\right]$.
\end{enumerate}
\end{abstract}
\maketitle

\section{Introduction}
Lattice partition functions with Yang--Baxter (YB) integrable Boltzmann weights have been widely used in a variety of problems in integrability. These include diagonalization and correlation functions in XXZ type Hamiltonians \cite{Baxter,Kor82,Slavnov,ReshSU3,KBI,WheelerSU3,Hutsalyuk,LV}, applications to symmetric functions \cite{CdGW,GdGW,WheelZJ,GunZJ,BWnonsymmetric,BS20,Mucciconi,GW,BBF11,BBBGIwahori}, integrable probability \cite{CorwinPetrov,Borodin_sym,KMMO,BWcoloured,ABW} and other areas of mathematical physics. More recently it was shown that the Feigin--Odesskii shuffle algebra $\Ac$ can be realized using lattice partition functions associated to the vertex model of $U_{t}(\widehat{sl}_n)$ \cite{Ng_tale,GZJ}. In the present work we find further connections between this shuffle algebra and lattice partition functions. 

Following \cite{GZJ} we consider domain-wall boundary partition functions of vertex models with a special relation between the spectral parameters. These partition functions are symmetric rational functions in the spectral parameters which satisfy the wheel conditions of the shuffle algebra $\Ac$, an algebra of symmetric functions with a special product called the shuffle product. Therefore the domain-wall partition functions represent certain elements of $\Ac$. We use them as building blocks to design lattice partition functions which realize shuffle products of elements of $\Ac$. As an application of this result we compute vertex model partition functions on the cone with the Boltzmann weights determined by the $R$-matrix of the $U_{t}(\widehat{sl}_{n+1|m})$ algebra. The generating function of these partition functions is given by the mixed Cauchy kernel \cite{FHSSY}, an object similar to the Cauchy kernel in the Macdonald theory. We then specialize the spectral parameters to particular values and obtain a lattice path formula for the skew Macdonald functions. 

In our construction the appearance of the Macdonald symmetric functions as lattice partition functions is quite different in spirit from the existing lattice constructions of symmetric functions which we mentioned above. It is analogous to the computation of the off-shell Bethe vectors (the mixed Cauchy kernel) in integrable models by means of the algebraic Bethe Ansatz and computation of the eigenvectors (Macdonald functions) by demanding that the spectral parameters satisfy Bethe equations. For example, in \cite{FJMM_BA} the mixed Cauchy kernel was used as the off-shell Bethe vector in the diagonalization problem of the integrable model associated to the quantum toroidal $gl_1$ algebra.

\subsection{Conic partition function}
Consider a cone pointing upwards and the square lattice drawn on it. More specifically, we draw $N$ directed parallel lines which wrap around the tip of the cone and then self-intersect:
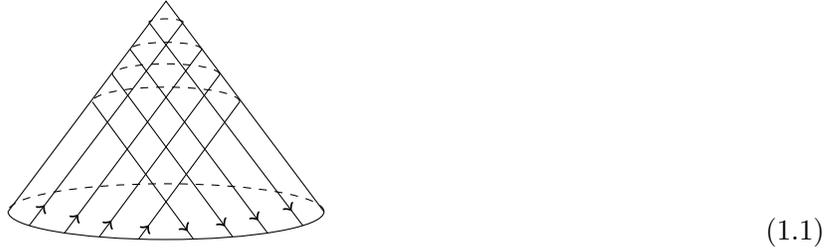
\begin{align}\label{eq:cone}
\begin{tikzpicture}[scale=0.7,transform shape]
\newcommand{\hht}{4}  
\newcommand{\wt}{3}  
\newcommand{\hwt}{0.5\wt}  
\pgfmathsetmacro{\hratio}{\hht / \wt}  
\pgfmathsetmacro{\ihratio}{\hwt / \hht}  
  \draw[dashed] (2*\wt,0) arc [start angle=0, end angle=180, x radius=\wt cm, y radius=\hwt cm];
    \draw (0,0) arc [start angle=-180, end angle=0, x radius=\wt cm, y radius=\hwt cm];
     \draw (0,0) -- (\wt,\hht) -- (2*\wt,0);
\foreach \angle in {30,50,65,80}{
  \pgfmathsetmacro{\xa}{\wt+\wt*cos(-\angle)}  
  \pgfmathsetmacro{\ya}{\hwt*sin(-\angle)}  
  \pgfmathsetmacro{\xb}{\wt+\wt*cos(\angle-180)}  
  \pgfmathsetmacro{\yb}{\hwt*sin(\angle-180)}  
  \pgfmathsetmacro{\xxa}{0.5*\xa+\ihratio*\ya-0.09}
  \pgfmathsetmacro{\yya}{0.5*\hratio*\xa+0.5*\ya}
  \pgfmathsetmacro{\xxb}{0.5*\xb+\wt-\ihratio*\yb+0.09}
  \pgfmathsetmacro{\yyb}{0.5*\hratio*\xa + 0.5*\ya}
\draw[invarrow=0.1] (\xa,\ya) -- (\xxa,\yya) ; 
  \draw[arrow=0.1] (\xb,\yb) -- (\xxb,\yyb) ; 
  \pgfmathsetmacro{\xxr}{0.5*\xxa-0.5*\xxb}
  \pgfmathsetmacro{\yyr}{-0.2*\xxr}
  \draw[dashed] (\xxb,\yyb) arc [start angle=180, end angle=0, x radius=\xxr cm, y radius=\yyr cm];
  }
\end{tikzpicture}
\end{align}

On this square lattice, we draw closed coloured paths where the colours are labelled by $1\dots n+m$ and the label $0$ is associated to the absence of a path. A path is drawn as follows: pick a colour $i$ and an intersection point of any two lines, and start drawing a path of colour $i$ from this intersection point in any of the two directions indicated by the arrows. By repeating the same process, we draw a continuous path. For such a path to be closed, it has to end where it started and for that it has to wrap around the cone. After a closed path is drawn one can continue drawing other coloured paths, ensuring that no two paths share the same edge, although they can occupy the same vertex.
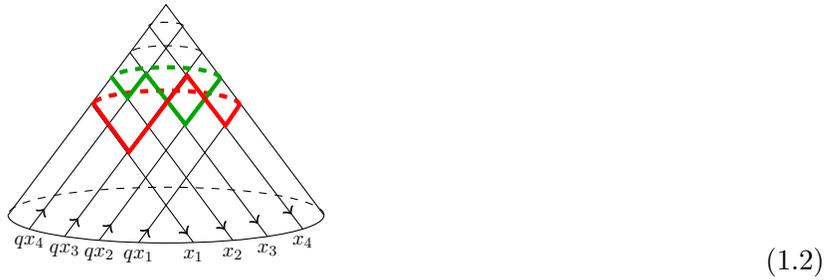
\begin{align}\label{eq:cone_paths}
\begin{tikzpicture}[scale=0.7,transform shape]
\newcommand{\hht}{4}  
\newcommand{\wt}{3}  
\newcommand{\hwt}{0.5\wt}  
\pgfmathsetmacro{\hratio}{\hht / \wt}  
\pgfmathsetmacro{\ihratio}{\hwt / \hht}  
  \draw[dashed] (2*\wt,0) arc [start angle=0, end angle=180, x radius=\wt cm, y radius=\hwt cm];
    \draw (0,0) arc [start angle=-180, end angle=0, x radius=\wt cm, y radius=\hwt cm];
     \draw (0,0) -- (\wt,\hht) -- (2*\wt,0);
\foreach \angle/\lab in {30/4,50/3,65/2,80/1}{
  \pgfmathsetmacro{\xa}{\wt+\wt*cos(-\angle)}  
  \pgfmathsetmacro{\ya}{\hwt*sin(-\angle)}  
  \pgfmathsetmacro{\xb}{\wt+\wt*cos(\angle-180)}  
  \pgfmathsetmacro{\yb}{\hwt*sin(\angle-180)}  
  \node[below] at (\xa,\ya) {$x_\lab$}; 
  \node[below] at (\xb,\yb) {$q x_\lab$}; 
  \pgfmathsetmacro{\xxa}{0.5*\xa+\ihratio*\ya-0.09}
  \pgfmathsetmacro{\yya}{0.5*\hratio*\xa+0.5*\ya}
  \pgfmathsetmacro{\xxb}{0.5*\xb+\wt-\ihratio*\yb+0.09}
  \pgfmathsetmacro{\yyb}{0.5*\hratio*\xa + 0.5*\ya}
\draw[invarrow=0.1] (\xa,\ya) -- (\xxa,\yya) ; 
  \draw[arrow=0.1] (\xb,\yb) -- (\xxb,\yyb) ; 
  \pgfmathsetmacro{\xxr}{0.5*\xxa-0.5*\xxb}
  \pgfmathsetmacro{\yyr}{-0.2*\xxr}
  \draw[dashed] (\xxb,\yyb) arc [start angle=180, end angle=0, x radius=\xxr cm, y radius=\yyr cm];
  }
\pgfmathsetmacro{\xa}{\wt+\wt*cos(-80)}
\pgfmathsetmacro{\ya}{\hwt*sin(-80)}
\pgfmathsetmacro{\xb}{\wt+\wt*cos(80-180)}  
  \pgfmathsetmacro{\yb}{\hwt*sin(80-180)}  
  \pgfmathsetmacro{\xxa}{0.5*\xa+\ihratio*\ya-0.09}
  \pgfmathsetmacro{\yya}{0.5*\hratio*\xa+0.5*\ya}
\pgfmathsetmacro{\xxb}{0.5*\xb+\wt-\ihratio*\yb+0.09}
  \pgfmathsetmacro{\yyb}{0.5*\hratio*\xa + 0.5*\ya}
  \pgfmathsetmacro{\xxr}{0.5*\xxa-0.5*\xxb}  
  \pgfmathsetmacro{\yyr}{-0.2*\xxr}
\draw[dashed, red, ultra thick] (\xxb,\yyb) arc [start angle=180, end angle=0, x radius=\xxr cm, y radius=\yyr cm];  
\pgfmathsetmacro{\xa}{\wt+\wt*cos(-50)}
\pgfmathsetmacro{\ya}{\hwt*sin(-50)}
\pgfmathsetmacro{\xxa}{0.5*\xa+\ihratio*\ya-0.09}
  \pgfmathsetmacro{\yya}{0.5*\hratio*\xa+0.5*\ya}
\draw[gcol, ultra thick] (\xa-0.9,\ya+3) -- (\xa-1.56,\ya+2.13) --(\xxa+0.3,\yya-0.4)--(\xxa-0.05,\yya-0.85) --(\xxa-0.36,\yya-0.45);
\pgfmathsetmacro{\xa}{\wt+\wt*cos(-65)}
\pgfmathsetmacro{\ya}{\hwt*sin(-65)}
\pgfmathsetmacro{\xb}{\wt+\wt*cos(65-180)}  
  \pgfmathsetmacro{\yb}{\hwt*sin(65-180)}  
  \pgfmathsetmacro{\xxa}{0.5*\xa+\ihratio*\ya-0.09}
  \pgfmathsetmacro{\yya}{0.5*\hratio*\xa+0.5*\ya}
\pgfmathsetmacro{\xxb}{0.5*\xb+\wt-\ihratio*\yb+0.09}
  \pgfmathsetmacro{\yyb}{0.5*\hratio*\xa + 0.5*\ya}
  \pgfmathsetmacro{\xxr}{0.5*\xxa-0.5*\xxb}
  \pgfmathsetmacro{\yyr}{-0.2*\xxr}
\draw[dashed, gcol, ultra thick] (\xxb,\yyb) arc [start angle=180, end angle=0, x radius=\xxr cm, y radius=\yyr cm];
\draw[red, ultra thick] (1.6,2.15) -- (2.3,1.2) -- (3.4,2.65)-- (4.13,1.7) -- (4.4,2.1) ;
\end{tikzpicture}
\end{align}
Such configurations $C$ carry a weight $W_C$ which is computed by multiplying the values of all local Boltzmann weights\footnote{These Boltzmann weights correspond to the super-symmetric $R$-matrix of the algebra $U_{t}(\widehat{sl}_{n+1|m})$ which was computed in \cite{BazS}. Our  convention for the weights correspond to \cite{WheelZJ,ABW}.}:
\begin{equation}\label{tikz:coloredvertices_intro}
\begin{tabular}{c@{\hskip 0.32cm}c@{\hskip 0.32cm}c@{\hskip 0.32cm}c@{\hskip 0.32cm}c@{\hskip 0.32cm}c@{\hskip 0.32cm}c}
\begin{tikzpicture}[scale=0.6,baseline=-2pt]
\draw[invarrow=0.75] (-1,0) --(1,0);
\draw[invarrow=0.75] (0,-1) --(0,1);
\node[right] at (1,0) {$\scriptstyle x$};
\node[above] at (0,1) {$\scriptstyle y$};
\end{tikzpicture}:
&
\begin{tikzpicture}[scale=0.6,baseline=-2pt]
\draw[red, ultra thick] (-1,0) node[left,black]{$\ss i$}  -- (0,0)--(0,1) node[above,black]{$\ss i$};
\draw (0,-1) node[below,black]{$\ss 0$} --(0,0)-- (1,0)node[right,black]{$\ss 0$};
\end{tikzpicture}
&
\begin{tikzpicture}[scale=0.6,baseline=-2pt]
\draw (-1,0) node[left,black]{$\ss 0$}  -- (0,0)--(0,1) node[above,black]{$\ss 0$};
\draw[red, ultra thick] (0,-1) node[below,black] {$\ss i$} --(0,0)-- (1,0) node[right,black]{$\ss i$};
\end{tikzpicture}
&
\begin{tikzpicture}[scale=0.6,baseline=-2pt]
\draw[red, ultra thick] (-1,0)  node[left,black]{$\ss i$}-- (1,0) node[right,black]{$\ss i$};
\draw (0,-1)  node[below,black]{$\ss 0$}--(0,1) node[above,black]{$\ss 0$};
\end{tikzpicture}
&
\begin{tikzpicture}[scale=0.6,baseline=-2pt]
\draw (-1,0) node[left,black]{$\ss 0$}  -- (1,0) node[right,black]{$\ss 0$};
\draw[red, ultra thick] (0,-1)node[below,black]{$\ss i$}  --(0,1) node[above,black]{$\ss i$};
\end{tikzpicture}
&
\begin{tikzpicture}[scale=0.6,baseline=-2pt]
\draw (-1,0) node[left,black]{$\ss 0$}  -- (1,0) node[right,black]{$\ss 0$};
\draw (0,-1)node[below,black]{$\ss 0$}  --(0,1)node[above,black]{$\ss 0$};
\end{tikzpicture}
\\[3em]
&
\begin{tikzpicture}[scale=0.6,baseline=-2pt]
\draw[red, ultra thick] (-1,0) node[left,black]{$\ss i$}  -- (0,0)--(0,1) node[above,black]{$\ss i$};
\draw[gcol, ultra thick] (0,-1) node[below,black]{$\ss j$} --(0,0)-- (1,0)node[right,black]{$\ss j$};
\end{tikzpicture}
&
\begin{tikzpicture}[scale=0.6,baseline=-2pt]
\draw[gcol, ultra thick] (-1,0) node[left,black]{$\ss j$}  -- (0,0)--(0,1) node[above,black]{$\ss j$};
\draw[red, ultra thick] (0,-1) node[below,black] {$\ss i$} --(0,0)-- (1,0) node[right,black]{$\ss i$};
\end{tikzpicture}
&
\begin{tikzpicture}[scale=0.6,baseline=-2pt]
\draw[red, ultra thick] (-1,0)  node[left,black]{$\ss i$}-- (1,0) node[right,black]{$\ss i$};
\draw[gcol, ultra thick] (0,-1)  node[below,black]{$\ss j$}--(0,1) node[above,black]{$\ss j$};
\end{tikzpicture}
&
\begin{tikzpicture}[scale=0.6,baseline=-2pt]
\draw[gcol, ultra thick] (-1,0) node[left,black]{$\ss j$}  -- (1,0) node[right,black]{$\ss j$};
\draw[red, ultra thick] (0,-1)node[below,black]{$\ss i$}  --(0,1) node[above,black]{$\ss i$};
\end{tikzpicture}
&
\begin{tikzpicture}[scale=0.6,baseline=-2pt]
\draw[ red,ultra thick] (-1,0) node[left,black]{$\ss i$}  -- (1,0) node[right,black]{$\ss i$};
\draw[red, ultra thick] (0,-1)node[below,black]{$\ss i$}  --(0,1)node[above,black]{$\ss i$};
\end{tikzpicture}
\\[3em]
& $\dfrac{1-t}{1-tx/y}$ & $\dfrac{(1-t)x/y}{1-t x/y}$& $\dfrac{t(1-x/y)}{1-t x/y}$&$\dfrac{1-x/y}{1-t x/y}$&$
 \begin{cases}
  \qquad 1  & i \leq n \\
  \dfrac{x/y-t}{1-tx/y}  & i> n
\end{cases}$
\end{tabular}
\end{equation}
where red ``$i$'' and green ``$j$'' can be replaced by any pair of colours such that the condition $i<j$ is preserved. In our conventions ``0'' is considered as the greatest colour. The vertices in the second row and last column are interpreted as two paths of colour $i\in\{0\dots n+m\}$ touching each other but not intersecting. The weight of this vertex depends on the colour label $i$ and for this reason we distinguish ``bosonic'' ($i\leq n$) and ``fermionic'' paths ($i> n$). The collection of all global configurations on the $N$ by $N$ lattice \eqref{eq:cone} with such bosonic and fermionic paths is denoted by $\Omega_N^{(n,m)}$. In addition to the local weights \eqref{tikz:coloredvertices_intro} each global path configuration is multiplied by a factor which accounts for the loop content of the configuration. The partition function is defined by:
\begin{align}\label{eq:Z}
  Z_N = \sum_{C\in \Omega_N^{(n,m)}}  
  z_0^{N-\Lambda(C)}
  \prod_{i=1}^n   z_i^{\Lambda_i(C)}
    \prod_{i=1}^m (-w_i)^{\Lambda_{n+i}(C)}
    \times W_C
\end{align}
where $\Lambda(C)$ denotes the total number of loops in $C$ and $\Lambda_j(C)$ denotes the total number of loops of colour $j$ in $C$.
Therefore the new variables $z_i$ count bosonic loops of colours $i$ and the variables $w_i$ count fermionic loops of colours $n+i$. The factor $z_0$ can be viewed as counting the ``empty'' loops cycling around the cone. The loop content can be conveniently represented by two non-negative integer compositions $\kappa=(\kappa_1\dots \kappa_n)$ and $\lambda=(\lambda_1\dots\lambda_m)$ in which $\kappa_i$ counts the number of loops of colour $i=1\dots n$ and $\lambda_i$ counts the number of loops of colour $i=n+1\dots n+m$. Therefore we can write:
\begin{align}\label{eq:Zkl}
  Z_N = 
\sum_{\kappa,\lambda:~|\kappa|+|\lambda|\leq N} 
  z_0^{N-|\kappa|-|\lambda|}
  z^\kappa w^\lambda 
    Z_{N,\kappa,\lambda}
\end{align}
where $z^\kappa=z_1^{\kappa_1}\cdots z_n^{\kappa_n}$, $|\kappa|$ denotes the total number of loops in $\kappa$ and similarly for $\lambda$. Therefore $Z_{N,\kappa,\lambda}$ is the conic partition function with the loop content given by $\kappa$ and $\lambda$. Let us consider an example. Set $N=4, n=1,m=1$ and suppose ``red colour $=1$'' and ``green colour $=2$'', then \eqref{eq:cone_paths} represents a valid configuration and we can compute its contribution to the partition function $Z_4$:
$$
z_0^2 z_1 w_1 \times 
\frac{\left(1-\frac{1}{q}\right) (1-t)^6 t^2 \left(1-\frac{x_2}{q x_3}\right) \left(1-\frac{x_1}{q x_4}\right) \left(1-\frac{x_4}{q x_2}\right)}{q^4 \left(1-\frac{t x_3}{q x_2}\right) \left(1-\frac{t x_4}{q
   x_1}\right) \left(1-\frac{t x_4}{q x_2}\right) \left(1-\frac{t x_4}{q x_3}\right)\prod_{i\in{3,4}}\prod_{j\in{1,2,3}}\left(1-\frac{t x_j}{q x_i}\right)}
$$
We find that the partition function $Z_N$ can be computed for an arbitrary choice of $n,m$ in terms of the {\it shuffle product} $*$. For two symmetric functions $F(x_1\dots x_k)$ and $G(x_1\dots x_l)$ we have:
\begin{align}\label{eq:sh_sub_intro}
F(x_1\dots x_k) * G(x_1\dots x_l) =
\sum_{\substack{S\subseteq [1\dots k+l]\\ |S|=k}}
F(x_S)G(x_{S^c}) \prod_{\substack{i\in S\\j\in S^c}} \zeta\left(\frac{x_i}{x_j}\right)
\end{align}
where $S^c$ denotes the complement of the subset $S$ and  $|S|=k$ means that the sum runs over subsets of length $k$; the function $\zeta$ is given by:
\begin{align}
    \label{eq:zeta-intro}
    \zeta(x) := \frac{(1-q x)(1-t^{-1}x)}{(1-x)(1- qt^{-1} x)}
\end{align}
We define the {\it shuffle-exponential}: $\exp_*(A):= 1 + A + \frac{1}{2!}  A*A + \frac{1}{3!}  A*A*A +\cdots $.
\begin{thm}\label{thm:Z_intro}
    The generating function:
    \begin{align}\label{eq:Zv}
        Z(v) = \sum_{N=0}^{\infty}v^N Z_N
    \end{align}
is given by:
    \begin{align}
    \label{eq:Zexp}
        Z(v) = \exp_* \left(\sum_{k>0}\frac{v^k}{k}\left(
        \sum_{i=1}^{m} w_i^k
        -
         \sum_{i=1}^n z_i^k 
         -
         \frac{q^k-t^k}{1-t^k}z_0^k
        \right) L_k\right)
    \end{align}
where $L_k=L_k(x_1\dots x_k)$ is a two-colour (one bosonic and one fermionic) conic partition function which can be defined by a formula similar to \eqref{eq:Z} but with a different loop counting weight:
\begin{align}
    \label{eq:L}
    L_k := \sum_{C\in \Omega_k^{(1,1)}:
    ~\Lambda(C)=k}(-1)^{\Lambda_2(C)} \Lambda_2(C) \times W_C
\end{align}
\end{thm}
This result implies that $Z_N$ is symmetric in $(z_1\ldots z_n)$ and separately in $(w_1\ldots w_m)$ since the dependence on $z$'s and $w$'s in the exponent in \eqref{eq:Zexp} is given by the power sum symmetric functions $p_k(z)$ and $p_k(w)$. The dependence of $Z_N$ on the spectral parameters $(x_1\ldots x_N)$ enters through the shuffle product of the functions $L_k(x_1\ldots x_k)$ in the view of the definition of exp$_*$.
Thus \eqref{eq:Zexp} allows us to express the partition function $Z_N$ (which is associated to the vertex model of  $U_t(\widehat{sl}_{n+1|m}$) in terms of the partition functions $L_k$ (which is associated to the vertex model of $U_t(\widehat{sl}_{1+1|1}$). The function $L_k$ in turn can be decomposed further:
\begin{align}
    L_{k} = \sum_{j=0}^k (-1)^j j Z^b_{k-j}
    *Z^f_{j}
\end{align}
where $Z^b_k$ and $Z^f_k$ are the six vertex bosonic and fermionic domain-wall partition functions.
The functions $L_k$ admit an explicit symmetrization formula: 
\begin{align}
    \label{eq:Psum_intro}
        L_k = \frac{-(1-q)^k(1-t)^k}{(1-q^k)(q-t)^{k}}
    {\normalfont\text{Sym}}\left(
    \frac{\sum_{j=0}^{k-1} (q/t)^j x_{j+1}/x_{1}}{\prod_{j=1}^{k-1} \left(1-(q/t) x_{j+1}/x_{j}\right)}
    \prod_{1\leq i<j\leq k}
    \zeta(x_i/x_j)
    \right)
\end{align}
where Sym is defined in \eqref{eq:Sym}. This symmetrization formula appeared in \cite{Ng_rev} as an explicit expression for a family of elements of $\Ac$ which are in some sense analogous to the power sums symmetric functions. Because of this connection $Z(v)$ in \eqref{eq:Zexp} can be viewed as a generalized mixed Cauchy kernel (see \cite{FHSSY}).

\subsection{Skew Macdonald functions}
The shuffle algebra $\mathcal{A}$ is well studied and has several representations \cite{FHHSY,SV-Hall,FT-shuffle,Ng_rev}. Consider a skew Young diagram of $\mu/\nu$ (in the English convention) with $k$ boxes $\square\in \mu/\nu$ which are labelled by $1\dots k$ in the reading order. We identify $\square=(a,b)$ where $a$ is the column index and $b$ is the row index of the box $\square\in \mu/\nu$. Let $\chi_\square = q^{a-1} t^{1-b}$ be the content of the box $\square\in \mu/\nu$ and denote by $\chi_1\dots \chi_k$ the contents of all $k$ boxes. The algebra $\mathcal{A}$ has a matrix representation $f\mapsto M$ in which an element $f(x_1\dots x_k)$ is mapped to a matrix $M$ whose non-zero matrix elements $M_{\mu,\nu}$ are those for which $\mu/\nu$ represents a skew Young diagram with $k$ boxes and:
$$
M_{\mu,\nu} = f(\chi_1\dots \chi_k)
$$
This representation of the shuffle algebra is intimately related to the theory of Macdonald functions $P_\lambda$. Using this representation we obtain the following connection between $Z_N$ and Macdonald functions.
\begin{thm}
Let $\mu/\nu$ be a skew Young diagram with $N$ boxes and $Z_N(x_1\dots x_N;z,w)$ be the partition function defined by \eqref{eq:Z} then:
\begin{align}
    \label{eq:Z_Mac}
    Z_N(\chi_1\dots \chi_N;z,w) =a_{\mu,\nu}^{-1} P_{\mu/\nu}\left[w-z-\frac{q-t}{1-t}z_0\right]
\end{align}
where we used the plethystic notation \eqref{eq:pleth} and $a_{\mu,\nu}$ is a combinatorial coefficient given in \eqref{eq:a}. 
When $z_i=0$ the above formula implies:
\begin{align}
    \label{eq:Mac_Z}
    P_{\mu/\nu}(w_1\dots w_m) =a_{\mu,\nu} Z_N(\chi_1\dots \chi_N;0,w)
\end{align}
\end{thm}
This theorem connects the lattice partition function on the cone $Z_N$ and the skew Macdonald functions. We may ask for a conic partition function with a fixed colour content of the paths. This amounts to asking for the coefficient of a specific monomial in $z$'s and $w$'s in \eqref{eq:Z}. Set $z_i=0$ and for a non-negative integer composition $\lambda=(\lambda_1\ldots \lambda_m)$ denote $Z_{\lambda}:=Z_{N,\varnothing,\lambda}$.  We obtain a lattice partition function representation of the monomial expansion of the skew Macdonald functions:
\begin{align}\label{eq:PskewZ}
P_{\mu/\nu}(w_1\dots w_m) =a_{\mu,\nu} \sum_{\lambda:~ |\lambda|= N}  
  w^\lambda 
    Z_{\lambda}(\chi_1\dots \chi_N)
\end{align}

Let us consider an example (see Section \ref{sec:skew_example}). We compute $P_{(2,1)/(1)}(w_1,w_2)$ with the above formula:
\begin{align*}
    a_{(2,1),(1)}^{-1} P_{(2,1)/(1)}(w_1,w_2) =
    Z_{(2,0)}(q, t^{-1}) w_1^2
    +
    Z_{(0,2)}(q, t^{-1}) w_2^2 
    +
    Z_{(1,1)}(q, t^{-1}) w_1 w_2
\end{align*}
Each of the monomial coefficients corresponds to specific loop configurations. We represent them using planar diagrams which are equivalent to the diagrams on the cone\footnote{If we rotate the planar pictures by 135 degrees counterclockwise we can place them on the lattice drawn on the cone \eqref{eq:cone}.}:
\begin{align}
a_{(2,1),(1)}^{-1} P_{(2,1)/(1)}(w_1,w_2) =
&\left(
\begin{tikzpicture}[scale=0.6, baseline=(current  bounding  box.center)]
\draw (1,0.5) -- node[pos=1,above] {$\scriptstyle 0$} (1,2.5);
\draw (2,0.5) -- node[pos=1,above] {$\scriptstyle 0$} (2,2.5);
\foreach\i/\lab in {2/,1/\delta_N}
\draw (0.5,\i) -- (2.5,\i);
\draw [rounded corners=5pt,red, ultra thick]  (2,0.5) -- (2,0.25) -- (2.75,0.25) -- (2.75,1) -- (2.5,1);
\draw[red,ultra thick] (2,0.5)--(2,1)--(2.5,1);
\draw [rounded corners=5pt,red,ultra thick]  (1,0.5) -- (1,0.5-0.5) -- (2+0.5+0.5,0.5-0.5) -- (2+0.5+0.5,2) -- (2.5,2);
\draw[red,ultra thick] (2,0.5)--(2,1)--(2.5,1);
\draw[red,ultra thick] (1,0.5)--(1,1)--(2,1)--(2,2)--(2.5,2);
\foreach\i/\lab in {2/0 ,1/0}
\node at (0.3,\i) {$\scriptstyle \lab$};
\end{tikzpicture}
+
\begin{tikzpicture}[scale=0.6, baseline=(current  bounding  box.center)]
\draw (1,0.5) -- node[pos=1,above] {$\scriptstyle 0$} (1,2.5);
\draw (2,0.5) -- node[pos=1,above] {$\scriptstyle 0$} (2,2.5);
\foreach\i/\lab in {2/,1/\delta_N}
\draw (0.5,\i) -- (2.5,\i);
\draw [rounded corners=5pt,red, ultra thick]  (2,0.5) -- (2,0.25) -- (2.75,0.25) -- (2.75,1) -- (2.5,1);
\draw[red,ultra thick] (2,0.5)--(2,1)--(2.5,1);
\draw [rounded corners=5pt,red,ultra thick]  (1,0.5) -- (1,0.5-0.5) -- (2+0.5+0.5,0.5-0.5) -- (2+0.5+0.5,2) -- (2.5,2);
\draw[red,ultra thick] (2,0.5)--(2,1)--(2.5,1);
\draw[red,ultra thick] (1,0.5)--(1,2)--(2.5,2);
\foreach\i/\lab in {2/0 ,1/0}
\node at (0.3,\i) {$\scriptstyle \lab$};
\end{tikzpicture}
\right) w^{2}_{1}\
+
\left(
\begin{tikzpicture}[scale=0.6, baseline=(current  bounding  box.center)]
\draw (1,0.5) -- node[pos=1,above] {$\scriptstyle 0$} (1,2.5);
\draw (2,0.5) -- node[pos=1,above] {$\scriptstyle 0$} (2,2.5);
\foreach\i/\lab in {2/,1/\delta_N}
\draw (0.5,\i) -- (2.5,\i);
\draw [rounded corners=5pt,gcol, ultra thick]  (2,0.5) -- (2,0.25) -- (2.75,0.25) -- (2.75,1) -- (2.5,1);
\draw[gcol,ultra thick] (2,0.5)--(2,1)--(2.5,1);
\draw [rounded corners=5pt,gcol,ultra thick]  (1,0.5) -- (1,0.5-0.5) -- (2+0.5+0.5,0.5-0.5) -- (2+0.5+0.5,2) -- (2.5,2);
\draw[gcol,ultra thick] (2,0.5)--(2,1)--(2.5,1);
\draw[gcol,ultra thick] (1,0.5)--(1,1)--(2,1)--(2,2)--(2.5,2);
\foreach\i/\lab in {2/0 ,1/0}
\node at (0.3,\i) {$\scriptstyle \lab$};
\end{tikzpicture}
+
\begin{tikzpicture}[scale=0.6, baseline=(current  bounding  box.center)]
\draw (1,0.5) -- node[pos=1,above] {$\scriptstyle 0$} (1,2.5);
\draw (2,0.5) -- node[pos=1,above] {$\scriptstyle 0$} (2,2.5);
\foreach\i/\lab in {2/,1/\delta_N}
\draw (0.5,\i) -- (2.5,\i);
\draw [rounded corners=5pt,gcol, ultra thick]  (2,0.5) -- (2,0.25) -- (2.75,0.25) -- (2.75,1) -- (2.5,1);
\draw[gcol,ultra thick] (2,0.5)--(2,1)--(2.5,1);
\draw [rounded corners=5pt,gcol,ultra thick]  (1,0.5) -- (1,0.5-0.5) -- (2+0.5+0.5,0.5-0.5) -- (2+0.5+0.5,2) -- (2.5,2);
\draw[gcol,ultra thick] (2,0.5)--(2,1)--(2.5,1);
\draw[gcol,ultra thick] (1,0.5)--(1,2)--(2.5,2);
\foreach\i/\lab in {2/0 ,1/0}
\node at (0.3,\i) {$\scriptstyle \lab$};
\end{tikzpicture}
\right)w^{2}_{2}\nonumber\\[1 em]
+
&
\left(
\begin{tikzpicture}[scale=0.6, baseline=(current  bounding  box.center)]
\draw (1,0.5) -- node[pos=1,above] {$\scriptstyle 0$} (1,2.5);
\draw (2,0.5) -- node[pos=1,above] {$\scriptstyle 0$} (2,2.5);
\foreach\i/\lab in {2/,1/\delta_N}
\draw (0.5,\i) -- (2.5,\i);
\draw [rounded corners=5pt,gcol, ultra thick]  (2,0.5) -- (2,0.25) -- (2.75,0.25) -- (2.75,1) -- (2.5,1);
\draw[gcol,ultra thick] (2,0.5)--(2,1)--(2.5,1);
\draw [rounded corners=5pt,red,ultra thick]  (1,0.5) -- (1,0.5-0.5) -- (2+0.5+0.5,0.5-0.5) -- (2+0.5+0.5,2) -- (2.5,2);
\draw[red,ultra thick] (1,0.5)--(1,1)--(2,1)--(2,2)--(2.5,2);
\foreach\i/\lab in {2/0 ,1/0}
\node at (0.3,\i) {$\scriptstyle \lab$};
\end{tikzpicture}+
\begin{tikzpicture}[scale=0.6, baseline=(current  bounding  box.center)]
\draw (1,0.5) -- node[pos=1,above] {$\scriptstyle 0$} (1,2.5);
\draw (2,0.5) -- node[pos=1,above] {$\scriptstyle 0$} (2,2.5);
\foreach\i/\lab in {2/,1/\delta_N}
\draw (0.5,\i) -- (2.5,\i);
\draw [rounded corners=5pt,gcol, ultra thick]  (2,0.5) -- (2,0.25) -- (2.75,0.25) -- (2.75,1) -- (2.5,1);
\draw[gcol,ultra thick] (2,0.5)--(2,1)--(2.5,1);
\draw [rounded corners=5pt,red,ultra thick]  (1,0.5) -- (1,0.5-0.5) -- (2+0.5+0.5,0.5-0.5) -- (2+0.5+0.5,2) -- (2.5,2);
\draw[red,ultra thick] (1,0.5)--(1,2)--(2.5,2);
\foreach\i/\lab in {2/0 ,1/0}
\node at (0.3,\i) {$\scriptstyle \lab$};
\end{tikzpicture}
+
\begin{tikzpicture}[scale=0.6, baseline=(current  bounding  box.center)]
\draw (1,0.5) -- node[pos=1,above] {$\scriptstyle 0$} (1,2.5);
\draw (2,0.5) -- node[pos=1,above] {$\scriptstyle 0$} (2,2.5);
\foreach\i/\lab in {2/,1/\delta_N}
\draw (0.5,\i) -- (2.5,\i);
\draw [rounded corners=5pt,red, ultra thick]  (2,0.5) -- (2,0.25) -- (2.75,0.25) -- (2.75,1) -- (2.5,1);
\draw[red,ultra thick] (2,0.5)--(2,1)--(2.5,1);
\draw [rounded corners=5pt,gcol,ultra thick]  (1,0.5) -- (1,0.5-0.5) -- (2+0.5+0.5,0.5-0.5) -- (2+0.5+0.5,2) -- (2.5,2);
\draw[gcol,ultra thick] (1,0.5)--(1,1)--(2,1)--(2,2)--(2.5,2);
\foreach\i/\lab in {2/0 ,1/0}
\node at (0.3,\i) {$\scriptstyle \lab$};
\end{tikzpicture}
+
\begin{tikzpicture}[scale=0.6, baseline=(current  bounding  box.center)]
\draw (1,0.5) -- node[pos=1,above] {$\scriptstyle 0$} (1,2.5);
\draw (2,0.5) -- node[pos=1,above] {$\scriptstyle 0$} (2,2.5);
\foreach\i/\lab in {2/,1/\delta_N}
\draw (0.5,\i) -- (2.5,\i);
\draw [rounded corners=5pt,red, ultra thick]  (2,0.5) -- (2,0.25) -- (2.75,0.25) -- (2.75,1) -- (2.5,1);
\draw[red,ultra thick] (2,0.5)--(2,1)--(2.5,1);
\draw [rounded corners=5pt,gcol,ultra thick]  (1,0.5) -- (1,0.5-0.5) -- (2+0.5+0.5,0.5-0.5) -- (2+0.5+0.5,2) -- (2.5,2);
\draw[gcol,ultra thick] (1,0.5)--(1,2)--(2.5,2);
\foreach\i/\lab in {2/0 ,1/0}
\node at (0.3,\i) {$\scriptstyle \lab$};
\end{tikzpicture}\right)w_{1}w_{2}
\nonumber
\end{align}
Using the Boltzmann weights \eqref{tikz:coloredvertices_intro}, with $n=0$ and the appropriate choice of the spectral parameters, and the coefficient $a_{(2,1),(1)}$, given in \eqref{eq:acoef_example}, we recover the skew Macdonald function $P_{(2,1)/(1)}(w_1,w_2)$:
$$
P_{(2,1)/(1)}(w_1,w_2) = 
w_1^2 + 
\frac{(1-t)  (2+q+t+2 q t)}{1-q t^2}w_1 w_2 
+w_2^2
$$

\subsection{Overview of the paper}
In Section \ref{sec:sym} we give the background on the theory of symmetric functions. In Section \ref{sec:shuffle} we review some aspects of the trigonometric Feigin--Odesskii shuffle algebra. In Section \ref{sec:six_vertex} we compute the conic partition for the six vertex model which is based on the $R$-matrix of $U_t(\widehat{sl}_2)$. In Section \ref{sec:supersymmetric} we extend the results of Section \ref{sec:six_vertex} to the case of $U_t(\widehat{sl}_{n+1|m})$. In Section \ref{sec:skew} we apply our results to the problem of computing the skew Macdonald functions.


\section{Background on symmetric functions}\label{sec:sym}
In this section we recall some basic facts from the theory of symmetric functions which will be required in the rest of the paper. All of the background material in this section can be found in  \cite{Macdonald}.

\subsection{Partitions} 
Let $\lambda=(\lambda_1,\lambda_2,\ldots )$, s.t. $\lambda_i\geq \lambda_{i+1}$ for all $i$, be an integer partition of $N$ and write $\lambda \vdash N$. The length of $\lambda$ is equal to the number of non-zero parts of $\lambda$ and is denoted by $\ell(\lambda)$. The sum of all parts $\lambda_i$ of $\lambda\vdash N$ is denoted by $|\lambda|$ and is equal to $N$. The multiplicity vector $m(\lambda)=(m_1(\lambda),m_2(\lambda)\ldots)$ is composed of integers $m_k(\lambda)$ which count how many parts of $\lambda$ are equal to $k$. In the context of integer sequences the notation $k^l$ means the sequence $k\ldots k$ which has $l$ repeats of $k$. We can write $\lambda = (\lambda_1^{m_{\lambda_1}}\ldots 2^{m_2}, 1^{m_1})$, where $m_k= m_k(\lambda)$. Integer partitions can be partially ordered using the dominance ordering: $\lambda\geq \mu$ when $\lambda_1+\cdots + \lambda_k \geq \mu_1+\cdots + \mu_k $ for all $k>0$.

A partition $\lambda$ is identified with the Young diagram where rows of boxes are placed horizontally and are non-increasing from top to bottom. By $\lambda'$ we denote the dual partition to the partition $\lambda$ which corresponds to the Young diagram which has rows and columns exchanged compared to $\lambda$. A box of a Young diagram is denoted by $\square$ and is identified with its coordinate $\square=(i,j)$, where the row index $i$ increases downwards and column index $j$ increases rightwards. The arm and leg functions $a_\lambda(\square)$, $l_\lambda(\square)$ are defined by 
$$
a_\lambda(\square) = \lambda_i-j,
\quad
l_\lambda(\square) = \lambda'_j-i
$$
The summations or products over $\square \in \lambda$ mean that $\square$ runs over all the boxes in the Young diagram of the partition $\lambda$. For two partitions $\lambda$ and $\mu$ such the $\lambda_i\geq \mu_i$, for all $i$, we define the skew partition $\lambda/\mu$ and the corresponding Young diagram.

\subsection{Basic symmetric functions}
Let $q,t$ be two formal variables. Consider the ring $\Lambda$ of symmetric functions in the alphabet $(x)=(x_1,x_2\ldots )$ with infinitely many variables and with coefficients in $\mathbb{F}:=\mathbb{Q}(q,t)$. A basic family of symmetric functions in this ring is the set of monomial symmetric functions $m_\lambda(x)$ which are labelled by integer partitions $\lambda$:  
\begin{align}\label{eq:monomials}
   m_{\lambda}=\sum_{\alpha}x^{\alpha}
\end{align}
where the sum runs over all distinct permutations $\alpha$ of $(\lambda_1\ldots \lambda_{\ell(\lambda)},0\ldots)$ and $x^{\alpha}:=\prod^{\infty}_{i=1}x^{\alpha_{i}}_{i}$.
As in \eqref{eq:monomials} we will often drop the explicit dependence on the alphabet. Another important family is the power sums symmetric functions:
\begin{align}\label{eq:power_sums}
    p_r = \sum_{i} x_i^r, 
    \qquad p_\lambda = p_{\lambda_1} \cdots p_{\lambda_{\ell(\lambda)}}
\end{align}
The two families $m_\lambda$ and $p_\lambda$ form bases in the ring $\Lambda$. Using the power sums we can generate new families of symmetric functions via exponential generating functions of the form:
$$
\exp\left(\sum_{r>0}c_r v^{r} p_r\right)
$$
for some choices of coefficients $c_r$. For our purposes we will need to consider three such families of functions $e_j$, $g_j$ and $g^*_j$, with $j=0,1,2\ldots$:
\begin{align}
    \label{eq:generating-e}
    &\sum_{j=0}^\infty v^j e_j = \exp\left(\sum_{r>0}\frac{(-1)^{r+1}}{r} v^{r} p_r\right)\\
        \label{eq:generating-g}
    &\sum_{j=0}^\infty v^j g_j = \exp\left(\sum_{r>0}\frac{1}{r}\frac{1-t^r}{1-q^r} v^{r} p_r\right)\\
        \label{eq:generating-g-dual}
    &\sum_{j=0}^\infty v^j g^*_j = \exp\left(-\sum_{r>0}\frac{1}{r}\frac{1-t^r}{1-q^r} v^{r} p_r\right)
\end{align}
These formulas define the elementary symmetric functions $e_j$ and two symmetric functions $g_j$ and $g_j^*$ (see \cite[Ch.\textrm{VI, \S 2}]{Macdonald}). By expanding the exponentials we can write explicit expansions of $e_j,g_j$ and $g_j^*$ in the power sums basis:
\begin{align*}
e_j = \sum_{\lambda \vdash j}    \frac{(-1)^{j+\ell(\lambda)}}{\lambda!}\prod_{r\in \lambda}\frac{1}{r}\cdot
p_\lambda,
\qquad
g_j = \sum_{\lambda \vdash j}    \frac{1}{\lambda!}\prod_{r\in \lambda}\frac{1}{r}\frac{1-t^r}{1-q^r}
\cdot p_\lambda,\qquad
g^*_j = \sum_{\lambda \vdash j}    \frac{(-1)^{\ell(\lambda)}}{\lambda!}\prod_{r\in \lambda}\frac{1}{r}\frac{1-t^r}{1-q^r}
\cdot p_\lambda
\end{align*}
where $\lambda! :=m_1(\lambda)! m_2(\lambda)!\cdots $ and the products over $r\in \lambda$ run over the parts of $\lambda$. 
The symmetric functions $e_j,g_j$ and $g_j^*$ can be used to write new bases:
\begin{align}
    \label{eq:bases:egg}
     e_\lambda = e_{\lambda_1} \cdots e_{\lambda_{\ell(\lambda)}},
     \qquad
      g_\lambda = g_{\lambda_1} \cdots g_{\lambda_{\ell(\lambda)}}, \qquad
       g^*_\lambda = g^*_{\lambda_1} \cdots g^*_{\lambda_{\ell(\lambda)}}
\end{align}

\subsection{Macdonald functions}
Define the scalar product using the power sums basis\footnote{In \eqref{eq:scalar} $\delta_{\lambda,\mu}$ is the standard Kronecker delta. We also define $\delta_{\text{True}}:=1$ and $\delta_{\text{False}}:=0$ for later use.}:
\begin{align}\label{eq:scalar}
    \langle p_{\lambda},p_{\mu}\rangle_{q,t}=\lambda! \prod_{r\in\lambda}r\frac{1-q^r}{1-t^r}\cdot \delta_{\lambda,\mu}
\end{align}
Recall \cite[Ch.\textrm{VI, \S 4}]{Macdonald}  that the Macdonald functions $P_{\lambda}$ are the unique symmetric functions in $\Lambda$ the Macdonald functions in finitely any variables that satisfy the following conditions:
\begin{align*}
    P_{\lambda}(x;q,t)=m_{\lambda} + \sum_{\mu < \lambda} c_{\lambda,\mu}m_{\mu},
    \qquad 
    \langle P_{\mu},P_{\nu}\rangle=0,
    \qquad \text{for }\mu \neq \nu
\end{align*}
where $c_{\lambda,\mu}\in \mathbb F$ are some coefficients and $\mu< \lambda$ in the sense of the dominance partial ordering.  The Macdonald functions are self-dual w.r.t. the scalar product \eqref{eq:scalar} up to a constant:
\begin{align}\label{eq:scalar_PQ}
\langle P_{\mu},Q_{_\lambda}\rangle =\delta_{\mu,\lambda},
\qquad Q_\lambda := b_\lambda P_\lambda, 
\end{align}
where the coefficient $b_\lambda\in \mathbb F$ is defined in terms two other coefficients $c_\lambda$ and $c'_\lambda$: 
\begin{align}\label{eq:bc-coef}
    b_\lambda := \frac{c_\lambda}{c'_\lambda}
    ,\qquad
    c_\lambda :=\prod_{\square \in \lambda}(1-q^{a_\lambda(\square)} t^{l_{\lambda}(\square)+1}) ,
    \qquad
    c'_\lambda := \prod_{\square \in \lambda}(1-q^{a_\lambda(\square)+1} t^{l_{\lambda}(\square)})
\end{align}

Consider $P_\lambda$ with the Young diagram of $\lambda$ given by a single column, in this case it is known that:
\begin{align}\label{eq:P-e}
P_{(1^{k})}=e_{k}
\end{align}
We will require the following {\it Pieri formula} \cite[Ch.\textrm{VI, \S 6}]{Macdonald} for Macdonald functions:
\begin{align}
    \label{eq:Pieri}
    e_{j} P_{\mu} =\sum_{\lambda}\psi'_{\lambda/\mu}P_{\lambda}
\end{align}
where the summation runs over $\lambda$ such that the skew partitions $\lambda/\mu$ are all vertical strips with $j$ boxes (i.e. $\lambda/\mu$ does not contain more than one box in a single row). The expansion coefficients $\psi'_{\lambda/\mu}$ are defined by:
\begin{align}\label{eq:psi-prime}
    \psi'_{\lambda/\mu} := \prod_{i,j} 
    \frac{(1-q^{\mu_i-\mu_j} t^{j-i-1})
        (1-q^{\lambda_i-\lambda_j} t^{j-i+1})}
        {(1-q^{\mu_i-\mu_j} t^{j-i})
        (1-q^{\lambda_i-\lambda_j} t^{j-i})}
\end{align}
where the product is taken over all $i,j$ such that $i<j$ and $\lambda_i=\mu_i$, $\lambda_j=\mu_j+1$.

Next we recall \cite[Ch.\textrm{VI, \S 7}]{Macdonald} the Littlewood--Richardson coefficients $f^{\mu}_{\lambda,\nu}=f^{\mu}_{\lambda,\nu}(q,t)$ and  the skew Macdonald functions $P_{\mu/\nu}$. The Littlewood--Richardson coefficients are the expansion coefficients of $P_{\lambda}P_{\nu}$ in the Macdonald basis:
\begin{align}\label{eq:LR-ceof}
P_{\lambda}P_{\nu}=\sum_{\mu}f^{\mu}_{\lambda,\nu}P_{\mu}
\end{align}
With these coefficients we can define the skew Macdonald functions which are symmetric functions labelled by skew partitions:
\begin{align}
    \label{eq:skew_M}
    P_{\mu/\nu} = \sum_{\lambda} \frac{b_\lambda b_\nu}{b_\mu}  f^\mu_{\lambda,\nu}  P_\lambda
\end{align}

The skew Macdonald functions $P_{\mu/\nu}(w_1\dots w_m)$ appearing in \eqref{eq:Mac_Z} depend on finitely many arguments. These functions, together with the Macdonald functions $P_{\mu}(w_1\dots w_m)$ and the finite alphabet $(w_1\dots w_m)$ counterparts of the other symmetric functions from this section, belong to the ring of symmetric polynomials $\Lambda_m=\mathbb F[w_1\ldots w_m]^{\mathcal{S}_m}$. We note, in particular, that the polynomials $P_{\mu}(w_1\dots w_m)$ with $|\mu|=k$  and $\ell(\mu)\leq m$ form a basis of $\Lambda_m^k$ which is the degree $k$ component of the graded ring $\Lambda_m$ consisting of the homogeneous symmetric polynomials of degree $k$.

\subsection{Cauchy kernel}
Consider two alphabets $(x)=(x_1\ldots x_k)$ and $(y)=(y_1\ldots y_l)$ and define the Macdonald {\it Cauchy kernel}  \cite[Ch.\textrm{VI, \S 2}]{Macdonald}:
\begin{align}\label{eq:Cauchy-kernel}
    \Pi(x,y):=
        \exp\left(\sum_{r>0}\frac{1}{r}\frac{1-t^r}{1-q^r} p_r(x)p_r(y)\right)
\end{align}
For any pair of bases $u_\lambda$ and $u^*_\lambda$ which are dual w.r.t. the scalar product \eqref{eq:scalar}, i.e.
$
\langle u_{\lambda},u^*_{\mu}\rangle_{q,t}=\delta_{\lambda,\mu}
$, 
we have:
\begin{align}\label{eq:Cauchy-resolution}
    \Pi(x,y)=
        \sum_{\lambda} u_\lambda(x) u^*_\lambda(y)
\end{align}
An important role in our paper is played by two pairs of dual bases. One pair is given by $m_\lambda$ and $g_\lambda$ and the other one by the Macdonald functions $P_\lambda$ and $Q_\lambda$:
\begin{align}
    \label{eq:g-m-P}
    \Pi(x,y) = \sum_{\lambda} m_\lambda(x) g_\lambda(y)
    = \sum_{\lambda} P_\lambda(x) Q_\lambda(y)
\end{align}

\section{The shuffle algebra}\label{sec:shuffle}
In this section, we introduce the trigonometric Feigin--Odesskii shuffle algebra \cite{SV-Hall,FT-shuffle,FHHSY,Ng_rev} which we denote by $\mathcal A$. This is an algebra of symmetric rational functions with a multiplication which is non-commutative in general. The algebra $\mathcal A$ contains a commutative subalgebra \cite{FHHSY}, which we denote by $\Ac$. We will focus on this commutative subalgebra. We present several examples of families of elements of $\Ac$ which were considered in the literature in various contexts. Among these examples we have the elements (completely factorized products) which were used in \cite{FHHSY} for computations related to the commutative algebra $\Ac$ and the elements (Izergin-type determinants) which are related to domain-wall partition functions as discussed in Sections \ref{sec:six_vertex} and \ref{sec:supersymmetric}. 
In this section we recall a particular representation of the algebra $\mathcal{A}$ \cite{SV-Hall,FT-shuffle,FHHSY} which gives rise to an isomorphism between $\Ac$ and the ring of symmetric functions $\Lambda$. This isomorphism is a key tool which will help us to relate the conic partition function $Z_N$ \eqref{eq:Z} with the skew Macdonald functions in Section \ref{sec:skew}.

In this section it will be convenient to use three parameters which are related to $q$ and $t$:
\begin{align}
    \label{eq:q123}
    q_1=q, \quad q_2 =t^{-1}, \quad q_3=t q^{-1}
\end{align}

\subsection{Definition of the shuffle algebra \texorpdfstring{$\mathcal{A}$}{}}
The shuffle algebra $\mathcal{A}$ is a vector space whose elements are symmetric rational functions. Their properties are determined by the function:
\begin{align}
    \label{eq:zeta}
    \zeta(x) := \frac{(1-q x)(1-t^{-1}x)}{(1-x)(1- qt^{-1} x)}
\end{align}
\begin{dfn}
Consider the vector space of symmetric rational functions:
\begin{align}\label{eq:V-space}
    \mathcal{V}= \bigoplus_{k\geq 0} \mathbb{F}(x_1\ldots x_k)^{\mathcal{S}_k}
\end{align}
Endow $\mathcal{V}$ with an algebra structure given by the shuffle product $*$. For $F(x_1\dots x_k)\in \mathcal{V}$ and $G(x_1\dots x_l)\in \mathcal{V}$ we have:
\begin{align}\label{eq:sh}
    F(x_1\dots x_k) * G(x_1\dots x_l) = 
    \frac{1}{k! l!}
    {\normalfont\text{Sym}}~ F(x_1\dots x_k) G(x_{k+1}\dots x_{k+l}) \prod_{\substack{i\in {1,\dots k} \\ j\in {k+1,\dots k+l}}}
    \zeta\left(\frac{x_i}{x_j}\right)
\end{align}
where:
\begin{align}\label{eq:Sym}
    {\normalfont\text{Sym}} ~P(x_1 \dots  x_k) = \sum_{\sigma \in \mathcal{S}_k} P(x_{\sigma(1)} \dots x_{\sigma(k)})
\end{align}
The shuffle algebra $\mathcal{A}\subset \mathcal{V}$ is defined as the set of rational functions of the form:
\begin{align}
\label{eq:A_element}
    F(x_1 \dots  x_k) =
    \frac{f(x_1\dots x_k)}{\prod_{1\leq i\neq j\leq k}(x_i-q t^{-1} x_j)},
    \qquad 
    f(x_1\dots x_k) \in \mathbb{F}[x_1^{\pm 1}\dots x_k^{\pm 1}]^{\mathcal{S}_k}
\end{align}
where $f(x_1\dots x_k)$ satisfies the wheel conditions:
\begin{align}
\label{eq:p_wheel}
f(x_1\dots x_k) = 0
\quad \text{if} \quad 
(x_i,x_j,x_k) = (x,q t^{-1} x, t^{-1} x) 
\quad \text{or} \quad
(x_i,x_j,x_k) = (x,q t^{-1} x, q x)
\end{align}
\end{dfn}
The shuffle product \eqref{eq:sh} is such that the product $F*G$ satisfies the wheel conditions if $F$ and $G$ do.  Note that \eqref{eq:sh} can also be written using a sum over subsets:
\begin{align}\label{eq:sh_sub}
F(x_1\dots x_k) * G(x_1\dots x_l) =
\sum_{\substack{S\subseteq [1\dots k+l]\\ |S|=k}}
F(x_S)G(x_{S^c}) \prod_{\substack{i\in S\\j\in S^c}} \zeta\left(\frac{x_i}{x_j}\right)
\end{align}
where the condition $|S|=k$ fixes the size of the subsets, $S^c\subseteq [1\dots k+l]$ refers to the subset complement to $S$ and its size must be equal to $l$. The algebra $\mathcal{A}$ is graded by the number of arguments $\mathcal{A}=\bigoplus_{k\geq 0} \mathcal{A}_k$, $F(x_1 \dots x_k)\in \mathcal{A}_k$. Consider a subalgebra $\mathcal{A}^\circ\subset \mathcal{A}$ of the elements $F\in \mathcal{A}_k$ for which the two limits:
\begin{align}
\label{eq:limit1}
    &\lim_{\xi \rightarrow 0}F(\xi x_1 \dots \xi x_{r},  x_{r+1} \dots  x_{k})\\
\label{eq:limit2}
    &\lim_{\xi \rightarrow \infty}
    F(\xi x_1 \dots \xi x_{r},  x_{r+1} \dots  x_{k})
\end{align}
exist and coincide for all fixed $r=1\dots k$. This subalgebra splits into components of fixed degree in the same way as $\mathcal{A}$:
$\mathcal{A}^\circ = \bigoplus_{k\in \mathbb{Z}_{\ge 0}} \mathcal{A}_{k}^\circ$. We have the following Proposition due to \cite{FHHSY}. 
\begin{prop}
The algebra $(\mathcal{A}^\circ, *)$ is commutative and the dimension of the graded subspace $\Ac_k$ is equal to $p(k)$, the number of partitions of $k$.
\end{prop}

\subsection{Basic elements of \texorpdfstring{$\mathcal{A}^\circ$}{}}
Let us define the elements $S_{k}$ of $\Ac$ which play the same role as the power sums in the ring of symmetric functions. These functions were introduced and studied in \cite{Ng_rev}. 
\begin{dfn}\label{def:sh_P}
    Define $S_k=S_k(x_1\ldots x_k)\in \Ac_k$:
\begin{align}
    \label{eq:Psum}
    S_k := \frac{(1-q)^k(1-t^{-1})^k}{(t-q)^k(1-t^{-k})}
    {\normalfont\text{Sym}}\left(
    \frac{\sum_{j=0}^{k-1} (q/t)^j x_{j+1}/x_{1}}{\prod_{j=1}^{k-1} \left(1-(q/t) x_{j+1}/x_{j}\right)}
    \prod_{1\leq i<j\leq k}
    \zeta(x_i/x_j)
    \right),
    \qquad
    k=0,1,2,\dots
\end{align}
\end{dfn}
\begin{dfn}\label{def:E-el}
Define $E_k(q_a)=E_k(q_a;x_1\ldots x_k)\in \Ac_k$:
\begin{align}\label{eq:E_def}
    E_k(q_a) := \prod_{1\leq i< j\leq k} \frac{(x_i-q_a x_j)(x_i-q_a^{-1} x_j)}{(x_i-q t^{-1} x_j)(x_i-t q^{-1} x_j)},
    \qquad a=1,2,3
\end{align}
\end{dfn}
We note that $E_k(t/q)=1$ but as an element in $\Ac_k$ it has to be viewed as a function of $k$ arguments $x_1\ldots x_k$. The factorized elements $E_k(q_a)$ and their elliptic generalizations were proposed in \cite{FO}. Since $E_k(q_a)$ are completely factorized they are very useful for computational purposes (see \cite{FHHSY}). In particular, it is easy to check that $E_k(q_a)\in \mathcal{A}_k$. Such factorized elements also play important roles in other commutative shuffle algebras  \cite{FT_Bethe,FJM_ssym}.

\begin{dfn}\label{def:H-el}
Let $(a,b,c)$ be a permutation of $(1,2,3)$, define $H_k(q_a)=H_k(q_a;x_1\ldots x_k)\in \Ac_k$:
\begin{align}\label{eq:H_def}
    H_k(q_a) := (q_a q t^{-1})^{k(k-1)/2}
    \frac{\prod_{1\leq i,j\leq k} (x_i - q_b x_j)(x_j - q_c x_i)}{\prod_{1\leq i\neq j\leq k}(x_i-x_j)(x_i-q t^{-1}x_j)}
    \det_{1\leq i,j\leq k} 
    \frac{1}{(x_i -q_b x_j)(x_j - q_c x_i)}
\end{align}
\end{dfn}
The elements $H_k(q_a)$ where discussed in \cite{GZJ} in connection with coloured lattice models. From Definition \ref{def:H-el} one can see that these element are members of $\mathcal{A}_k$ by expanding the determinant and applying the wheel conditions \eqref{eq:p_wheel}. Alternatively one can express $H_k(q_a)$ in terms of $E_k$ as in \eqref{eq:Hp_E} below. All families of functions listed in Definitions \ref{def:sh_P}-\ref{def:H-el} are members of $\Ac$ and thus they mutually commute with respect to the shuffle product. Therefore they can be shuffle multiplied and ordered and thus give bases in $\Ac$. For $\lambda\vdash n$, $a\in\{1,2,3\}$, we have:
$$
v_\lambda = v_{\lambda_1}*v_{\lambda_2}*\cdots * v_{\lambda_{\ell(\lambda)}},
\qquad v=S,E(q_a),H(q_a)
$$
In order to exponentially generate elements of $\Ac$ we introduce the shuffle version of the exponential function:
\begin{align}
    \label{eq:sh-exp}
    \exp_*(A) : = 1+ A +\frac{1}{2!} A*A+\frac{1}{3!} A*A*A+\dots
\end{align}
Define the generating functions:
\begin{align}
    \label{eq:genfs}
    E(v;q_a):= \sum_{k=0}^\infty v^k E_k(q_a),
    \qquad
    H(v;q_a):= \sum_{k=0}^\infty v^k H_k(q_a)
\end{align}
\begin{lemma}\label{lem:EH_gen}
Let $a\in\{1,2,3\}$. The generating functions $E_k(q_a)$ and $H_k(q_a)$ are equal to shuffle-exponentials:
\begin{align}
    \label{eq:E-gen}
    E(v;q_a)&= \exp_*
    \left(\sum_{r>0}\frac{(-1)^{r+1}}{r}
    \frac{1-q_a^r}{1-q^r}\frac{(t-q)^r}{(1-q_a)^r}
    v^{r} S_r
    \right)\\
    \label{eq:H-gen}
    H(v;q_a)&=\exp_*
    \left(\sum_{r>0}\frac{1}{r}
    \frac{1-q_a^r}{1-q^r}\frac{(t-q)^r}{(1-q_a)^r}
    v^{r} S_r
    \right)
\end{align}    
\end{lemma}
\begin{proof}
The generating function $E(v;p)$ was computed and written in the exponential form \eqref{eq:E-gen} in \cite{Ng_moduli}. The generating function $H(v;q_a)$ follows from the quadratic identity:
\begin{align}\label{eq:Hp_E}
    H_k(q_a) =   
    \sum_{r=0}^k 
    q_c^{k-r}
    \left(\frac{1-q_b}{1-q_b q_c}\right)^{k-r}
    \left(\frac{1-q_c}{1-q_b q_c}\right)^{r}
    E_{k-r}(q_b)*E_{r}(q_c)
\end{align}
where $(a,b,c)$ is a permutation of $(1,2,3)$. We  prove \eqref{eq:Hp_E} in Appendix \ref{app:EH}. By summing \eqref{eq:Hp_E} over $k$ with $v^k$ we obtain a product of two shuffle-exponentials on the right hand side of the resulting equation. These exponentials combine and produce the r.h.s. of \eqref{eq:H-gen}.     
\end{proof}
\subsection{Evaluation representation of the shuffle algebra \texorpdfstring{$\mathcal{A}$}{}}
Recall that $\square=(a,b)$ denotes a box in the Young diagram of $\lambda\vdash k$ located in the $a$-th column and $b$-th row. Define $\chi_\square$ to be the content of the box $\square$:
\begin{align}\label{eq:chi}
    \chi_\square = q^{a-1} t^{1-b},
    \qquad \square \in \lambda
\end{align}
Consider an example of $\lambda=(532)$ and in each box of the Young diagram of $\lambda$ write its content:
\begin{align}\label{yd:qt}
\ytableausetup
{mathmode, boxframe=normal, boxsize=2.1em}
\begin{ytableau}
\ss        1 & \ss  q & \ss  q^2 & \ss  q^3 & \ss  q^4 \\
\ss        t^{-1} &  \ss  q t^{-1} & \ss  q^2 t^{-1}  \\
\ss        t^{-2} & \ss  q t^{-2} 
\end{ytableau}
\end{align}
\begin{dfn}
Let $\lambda\vdash k$ be a partition, we say that $\square\in \lambda$ is the $i$-th box of $\lambda$ if  $\square$ is located on the $i$-th position in the reading order\footnote{Since we will be working with symmetric functions $f(x_1 \dots x_k)$ the order in which we associate $x_i$ with a particular $\square$ in $\lambda$ does not matter but it is convenient to fix it.}. For a function $f(x_1 \ldots x_k)$ we define  evaluations {\normalfont ev$_\lambda$}:
\begin{align}\label{eq:ev}
    {\normalfont \ev_\lambda}: \quad f(\dots x_i \dots) \mapsto 
    f(\dots \chi_\square \dots)
\end{align}    
which means that each $x_i$ is replaced with the content of the $i$-th box of $\lambda$.
\end{dfn}
For example, if we need to compute $\ev_\lambda\left( f(\dots x_i \dots)\right)$ with $\lambda=(532)\vdash 10$ we first assign $x$'s to the boxes as follows:
\begin{align}\label{yd:x}
\ytableausetup
{mathmode, boxframe=normal, boxsize=1.5em}
\begin{ytableau}
       \ss x_1 & \ss x_2 & \ss x_3 & \ss x_4 & \ss x_5 \\
       \ss x_6 &  \ss x_7 & \ss x_8  \\
       \ss x_9 & \ss x_{10}
\end{ytableau}
\end{align}
and then substitute for $x$'s the values of the contents of their boxes. We would like to apply $\ev_\lambda$ to symmetric functions in $\mathcal{A}$, however, due to the poles at $x_i=q t^{-1} x_j$ in \eqref{eq:A_element} we need to take extra care. This can be realized using an intermediate step:
\begin{align}\label{eq:evxy}
    \ev_\lambda (f(x_1 \ldots x_k)) =
    \ev^y \left(\ev^x_\lambda (f(x_1\ldots x_k))\right) 
    ,    \qquad \lambda \vdash k
   \end{align}
where
\begin{align}\label{eq:ev_x}
    \ev^x_\lambda (f(x_1\ldots x_k))  &= f(y_1, q y_1  \ldots  q^{\lambda_1-1} y_1 ,y_2, q y_2  \ldots q^{\lambda_2-1} y_2  \ldots
    y_{\ell(\lambda)}, q y_{\ell(\lambda)} \ldots q^{\lambda_{\ell(\lambda)}-1} y_{\ell(\lambda)}  \\
    \label{eq:ev_y}
    \ev^y (g(y_1\ldots y_j))  &= g(y,t^{-1} y \ldots t^{j-1}  y )
\end{align}
Note that for $F\in \mathcal{A}_k$, $\ev^x_\lambda (F(x_1 \ldots x_k))$ will produce a function of $y_1 \ldots  y_{\ell(\lambda)}$ which has no poles at $y_i=t^{-1} y_j$ due to the wheel conditions \eqref{eq:p_wheel}. The evaluation maps \eqref{eq:ev}, \eqref{eq:ev_x} and \eqref{eq:ev_y} can be extended to the skew diagrams $\lambda/ \mu$, s.t. $|\lambda|-|\mu|=k$:
\begin{align}\label{eq:ev_skew}
    \ev_{\lambda/ \mu}:\quad f(\ldots x_i \ldots) \mapsto 
    f(\dots \chi_\square \dots)
\end{align}
and in the skew analogue of \eqref{yd:x} $x$'s are distributed in the reading order similar to \eqref{yd:x}. 

Let o$(\lambda)$ and i$(\lambda)$ be the sets of coordinates of boxes corresponding to outer and inner corners of a Young diagram of $\lambda$ respectively. In order to explain this more precisely we consider an example with $\lambda=(532)$. The locations of the o$(\lambda)$ boxes and i$(\lambda)$ boxes are indicated on the left and on the right pictures respectively:
\begin{center}
\ytableausetup
{mathmode, boxframe=normal, boxsize=1em}
o$(\lambda): \quad$ 
\begin{ytableau}
        \, & \, & \, & \, & \, & \none
       \\
        \, &  \, & \, & \none & \none & \none[\bullet]
       \\
        \, & \,  & \none & \none[\bullet]
       \\ 
       \none & \none & \none[\bullet]
\end{ytableau}
\qquad \, \qquad
i$(\lambda): \quad$
\begin{ytableau}
        \, & \, & \, & \, & \, & \none[\bullet]
       \\
        \, &  \, & \, & \none[\bullet]
       \\
        \, & \,  & \none[\bullet]
       \\ 
       \none[\bullet]
\end{ytableau}
\end{center}
In this example we have o$(\lambda)=\{(2,6),(3,4),(4,3)\}$ and i$(\lambda)=\{(1,6),(2,4),(3,3),(4,1)\}$.
\begin{dfn}
Define a combinatorial factor $d_{\lambda/\mu}$:
\begin{align*}
    d_{\lambda/\mu}:=
    \left(\frac{(1-q)(1-t^{-1})}{1-qt^{-1}}\right)^{|\lambda/\mu|}
    \prod_{\square \in \lambda/\mu}\frac{\prod_{\square' \in {\normalfont{\text{o}}}(\lambda)} (1-\chi_{\square'}/\chi_\square)}{\prod_{\square' \in {\normalfont\text{i}}(\lambda)} (1-\chi_{\square'}/\chi_\square)}
\end{align*}
\end{dfn}
Let us turn to the \emph{evaluation representation} \cite{FT-shuffle,SV-Hall,Ng_moduli} of the shuffle algebra $\mathcal{A}$ on a graded vector space $\mathcal{F}$ whose basis vectors (assumed to be orthonormal) are labelled by integer partitions and the grading is determined by the weight of partitions.
\begin{prop}
To $F(x_1 \ldots x_k)\in \mathcal{A}_k$ we associate an infinite dimensional matrix whose rows and columns are labelled by partitions $\lambda$ and $\mu$:
\begin{align}\label{eq:sh_rep}
    \bra{\lambda} F(x_1 \ldots x_k) \ket{\mu}
    =\delta_{|\lambda/\mu|=k}\,d_{\lambda/\mu}{\normalfont\ev}_{\lambda/ \mu}(F(x_1 \ldots x_k))
\end{align}
We set the right hand side of \eqref{eq:sh_rep} to zero if $\lambda/\mu$ is not a skew partition. The map from $\mathcal{A}_k$ to such matrices defines a representation of $\mathcal{A}$.
\end{prop}
Let us sketch the proof of this statement.  
We need to show that:
\begin{align}\label{eq:sh_hom}
    \bra{\lambda} 
    F(x_1\dots x_k) * G(x_1\dots x_l) 
    \ket{\mu} =
    \sum_{\nu}
    \bra{\lambda} 
    F(x_1\dots x_k)\ket{\nu} 
    \bra{\nu}
    G(x_1\dots x_l) 
    \ket{\mu}
\end{align}
Due to \eqref{eq:sh_rep} the summation on the r.h.s. of \eqref{eq:sh_hom} contains non-zero terms only for such partitions $\nu$ that both $\lambda/\nu$ and $\nu/ \mu$ correspond to skew Young diagrams with $k$ and $l$ boxes respectively. In order to see that the same summation occurs on the l.h.s. write the shuffle product in the form \eqref{eq:sh_sub} and consider the map ev$_{\lambda/\mu}$ applied to each summand:
\begin{align}\label{eq:ev_FG}
\ev_{\lambda/\mu} 
\left(F(x_1\dots,x_k) * G(x_1\dots,x_l) \right)=
\sum_{\substack{S\subseteq [1\dots k+l]\\ |S|=k}}
\ev_{\lambda/\mu}\left(
F(x_S)G(x_{S^c}) \prod_{\substack{i\in S\\j\in S^c}} \zeta\left(\frac{x_i}{x_j}\right)\right)
\end{align}
Fix an $S\in [1\dots k+l]$ and compute the corresponding term on the r.h.s. of \eqref{eq:ev_FG}. The evaluation of $\zeta$ factors reads:
$$
\ev_{\lambda/\mu} \zeta\left(\frac{x_i}{x_j}\right)
= 
\zeta\left(\frac{\chi_{\square_i}}{\chi_{\square_j}}\right)
=
\zeta\left(q^{a_i-a_j}t^{b_j-b_i}\right)
=0\quad \text{if }
\begin{cases} 
    a_j-a_i=1 \quad  b_i=b_j \\
    b_j-b_i=1 \quad   a_i= a_j 
  \end{cases},
  \qquad i\in S, ~~j\in S^c
$$ 
where $\square_i=(a_i,b_i)$ and $\square_j=(a_j,b_j)$ are the $i$-th and $j$-th boxes of $\lambda/\mu$ respectively.
This means that if the boxes $\square_i$ and $\square_j$ are located in the same row ($b_i=b_j$) and are bordering each other then $\square_i$ must be on the right to $\square_j$ for the $\zeta$ factor not to vanish. It also means that if the boxes $\square_i$ and $\square_j$ are located in the same column ($a_i=a_j$) and are bordering each other then $\square_i$ must be on the bottom to $\square_j$ for the $\zeta$ factor not to vanish. Together these conditions mean that, if the $x$'s are distributed over the boxes of $\lambda/\mu$ in the reading order, the non-zero terms of ev$_{\lambda/\mu}(F*G)$ are such that the sets $S$ and $S^c$ split the boxes of $\lambda/\mu$ into two regions whose boundary defines another partition $\nu$. For $\lambda=(442)$ and $\mu=(31)$, $k=3$ and $l=3$ an example of a term, on the r.h.s. of \eqref{eq:ev_FG}, which  gives a non-zero contribution under $\ev_{\lambda/ \mu}$ is:
\begin{align}\label{yd:S}
\ytableausetup
{mathmode, boxframe=normal, boxsize=1.5em}
\begin{ytableau}
       \none  & \none  & \none  & *(gray)\ss x_1 \\
       \none  & *(gray)\ss  x_2 &\ss   x_3 &\ss  x_4  \\
       *(gray)\ss  x_5 &\ss  x_6
\end{ytableau}
\qquad S=\{3,4,6\}, \quad S^c=\{1,2,5\}
\end{align}
where in the diagram we shaded those $x$'s which correspond to the set $S^c$. The partition $\nu$ corresponding to this term is $\nu = (421)$. Continuing with \eqref{eq:ev_FG} in this particular case we have: 
\begin{multline*}
\ev_{(442)/(31)} 
\left(F(x_1,x_2, x_3) * G(x_1,x_2, x_3) \right)=
\cdots +
\ev_{(442)/(31)}
\left(F(x_S)G(x_{S^c}) \prod_{\substack{i\in S\\j\in S^c}} \zeta\left(\frac{x_i}{x_j}\right)\right)+\cdots\\
=
\cdots +
\ev_{(442)/(421)}\left(F(x_1,x_2,x_3)\right)
\ev_{(421)/(31)}\left(G(x_1,x_2,x_3)\right) \ev_{(442)/(31)}\left(
\prod_{\substack{i\in S\\j\in S^c}} 
\zeta\left(\frac{x_i}{x_j}\right)\right)+\cdots
\end{multline*}
In order to verify \eqref{eq:sh_hom} it remains to check that:
$$
\ev_{\lambda/\mu}
\left(\prod_{\substack{i\in S\\j\in S^c}}\zeta\left(\frac{x_i}{x_j}\right) \right)
=\frac{d_{\lambda/\nu} d_{\nu/\mu}}{d_{\lambda/\mu}}
$$

\subsection{Isomorphism of \texorpdfstring{$\Ac$}{TEXT} and \texorpdfstring{$\Lambda$}{TEXT}.}
Let us now restrict our attention to $\mathcal{A}^\circ$ and derive a graded algebra isomorphism between $\Ac$ and $\Lambda$ \cite{FT-shuffle}. The representation \eqref{eq:sh_rep} allows us to construct the vector space $\mathcal{F}$ as a module starting from the {\it vacuum vector} $\ket{\varnothing}$. More precisely, for any partition $\lambda\vdash k$ one can construct shuffle algebra elements $F_\lambda \in \mathcal{A}^\circ_k$ such that $F_\lambda \ket{\varnothing} = \ket{\lambda}$. Indeed, by \eqref{eq:sh_rep} we have:
\begin{align}\label{eq:G_vac}
    G(x_1 \ldots x_k) \ket{\varnothing} = \sum_{\lambda\vdash k} 
    d_{\lambda}\ev_{\lambda}(G(x_1 \ldots x_k)) \ket{\lambda},
    \qquad G\in \mathcal{A}^\circ_k
\end{align}
By choosing a basis in $\mathcal{A}^\circ_k$, for example $E_\mu(x;q)$ for all $\mu\vdash k$, we can write $p(k)$ such equations (see explicit evaluations of $E_k(x;q)$ in \eqref{eq:ev_E_prod}). Thus we can solve this system for  $\ket{\lambda}$ in terms of the elements of $\Ac_k$. In other words there exists an element  $F_\lambda \in \mathcal{A}^\circ_k$ such that:
\begin{align}
    \label{eq:F}
    F_\lambda(x_1\dots x_k) \ket{\varnothing} = \ket{\lambda}, \qquad \lambda \vdash k
\end{align}
Applying the dual vector $\bra{\mu}$ to this equation and using (\ref{eq:sh_rep}) gives us:
\begin{align}\label{eq:evF}
    \ev_\mu(F_\lambda(x_1\dots x_k))= \delta_{\lambda,\mu} \frac{1}{d_\lambda}
\end{align}
In particular, these equations can be used to compute $F_\lambda$. From \eqref{eq:G_vac} and \eqref{eq:F} it follows that for any element of $\Ac_k$ we have the expansion in the basis of $F_\lambda$:
\begin{align}\label{eq:P_F}
    G(x_1\ldots x_k)  = \sum_{\lambda\vdash k} 
    d_{\lambda}\ev_{\lambda}(G(x_1\ldots x_k)) F_\lambda,
    \qquad G\in \mathcal{A}^\circ_k
\end{align}
Recall the coefficients $\psi'_{\lambda/\mu}$ and $c_\lambda$  from \eqref{eq:psi-prime} and \eqref{eq:bc-coef} respectively and define:
\begin{align}
\label{eq:n-coef}    
n(\lambda):= \sum_{i=1}^{\ell(\lambda)} (i-1)\lambda_i
\end{align}
We have the following two Lemmas \cite{FT-shuffle}.
\begin{lemma}\label{lem_EF}
The shuffle product $E_k(q)*F_\mu$ expands in the basis of $F_\lambda$ as follows:
\begin{align}\label{eq:E_F}
    E_k(q)*F_\mu=
    \sum_{\lambda} \phi_{\lambda/\mu} F_\lambda,
\end{align}
where the summation runs over $\lambda$ such that the skew partitions $\lambda/\mu$ are all vertical strips with $k$ boxes and 
\begin{align}\label{eq:phi_psi}
    \phi_{\lambda/\mu}:= (1-t)^{k}q^{n(\lambda')-n(\mu')} \frac{c_\mu}{c_\lambda} \psi'_{\lambda/\mu}
\end{align}
\end{lemma} 

Let us sketch the proof of this statement. Due to \eqref{eq:F} showing \eqref{eq:E_F} is equivalent to showing:
\begin{align}
    \label{eq:E_F_vac}
    E_k(q)\ket{\mu}=
    \sum_{\lambda} \phi_{\lambda/\mu} \ket{\lambda}
\end{align}
By \eqref{eq:sh_rep} we have:
\begin{align}
    \label{eq:phi}
    \phi_{\lambda/\mu}  = d_{\lambda/\mu} \ev_{\lambda/\mu}\left(E_k(q) \right)
\end{align}
Recall that $\ev$ is computed in two steps \eqref{eq:ev_x} and \eqref{eq:ev_y}. Compute $\ev^x_{\lambda/\mu} $:
\begin{align*}
    \ev^x_{\lambda/\mu} 
   \left( \prod_{1\leq i< j\leq k} \frac{(x_i-q x_j)(x_i-q^{-1} x_j)}{(x_i-q t^{-1} x_j)(x_i-t q^{-1} x_j)}\right)
    = 0\quad \text{if $\exists$ $\square,\square'\in \lambda/\mu$, s.t.: $\chi_\square/\chi_{\square'}=q$}
\end{align*}
This implies that the skew partition $\lambda/\mu$ cannot contain more than one box in a single row and thus is a vertical strip. This determines that the summation set over $\lambda$ in \eqref{eq:E_F_vac} and \eqref{eq:E_F} must be given by the set of all vertical strips with $k$ boxes. Computing ev$_{\lambda/\mu}$ of $ E_k(q)$ gives:
\begin{align}\label{eq:ev_E_prod}
    \ev_{\lambda/\mu} \left(E_k(q) \right)= 
    t^{-k(k-1)/2}\frac{(1-q t^{-1})^k}{(1-q)^k}
    \prod_{\square,\square' \in \lambda/\mu}
    \frac{1-
    q \chi_{\square'}/\chi_{\square}}{1-
    q t^{-1}  \chi_{\square'}/\chi_{\square}}
\end{align}
In order to complete the proof one needs to multiply the expression on the r.h.s. in \eqref{eq:ev_E_prod} by $d_{\lambda/\mu}$ and compare it with the definition of $\phi$ from \eqref{eq:phi_psi}.

\begin{lemma}\label{lem:iso}
We have an isomorphism $\iota$ of algebras $\mathcal{A}^\circ$ and $\Lambda$ given by matching $F_\lambda$ with Macdonald functions $P_\lambda$:
\begin{align}
    \label{eq:iso}
    \iota:\quad F_\lambda \mapsto 
    \frac{c_\lambda }{q^{n(\lambda')}(1-t)^{|\lambda|}}
    P_\lambda
\end{align}
\end{lemma}
\begin{proof}
We consider \eqref{eq:iso} to be a linear map and show that it is an isomorphism. A particular case of the map \eqref{eq:iso} is when $\lambda=(1^k)$. In this case the Macdonald function coincides with the elementary symmetric function $P_{(k)} = e_k$. Consider  \eqref{eq:E_F} with $\mu=\varnothing$, in this case the r.h.s. of \eqref{eq:E_F} contains a single term with $\lambda=(1^k)$. By \eqref{eq:F} we have $F_{\varnothing}=1$ and computing $\phi_{(1)^k/\varnothing}$ gives:
\begin{align}
    \label{eq:F1n}
    F_{(1^k)} = \prod_{i=1}^k\frac{1-t^i}{1-t} E_k(q)
\end{align}
Therefore a special case of the map \eqref{eq:iso} is:
\begin{align}\label{eq:iota_e}
    \iota: \quad E_k(q) \mapsto  e_k
\end{align}
where the factors depending on $t$ in \eqref{eq:F1n} canceled with the factors from \eqref{eq:iso}. Next we compute:
\begin{align*}
    \iota \left(E_k(q) * F_\mu \right)  
    &= \sum_\lambda \phi_{\lambda/\mu}  \iota\left( F_\lambda \right)   = 
    \frac{c_\mu }{q^{n(\mu')}(1-t)^{|\mu|}}
     \sum_\lambda \psi'_{\lambda/\mu} P_\lambda
    =
    \frac{c_\mu }{q^{n(\mu')}(1-t)^{|\mu|}}
    e_k P_\mu = \iota \left(E_k(q)\right)  \iota\left( F_\mu \right)
\end{align*}
where we used \eqref{eq:E_F} in the first equality, \eqref{eq:phi_psi} and \eqref{eq:iso} in the second equality, the Pieri formula \eqref{eq:Pieri} in the third equality and \eqref{eq:iso} and \eqref{eq:iota_e} in the fourth equality. Consider an expansion of $E_\lambda(q)$ in the basis $F_\mu$ and let $C_{\lambda,\mu}$ denote the expansion coefficients, then compute using linearity of $\iota$ and the above equation:
\begin{multline*}
\iota \left( E_k(q) * E_\lambda(q) \right)= \sum_\mu C_{\lambda,\mu}\, \iota \left( E_k(q) * F_\mu \right)= 
 \sum_\mu C_{\lambda,\mu}\, \iota \left( E_k(q) \right) *\iota\left( F_\mu\right)\\ = \iota \left( E_k(q) \right) *\iota\left(\sum_\mu C_{\lambda,\mu} F_\mu\right) 
 =\iota \left( E_k(q) \right)\iota \left( E_\lambda(q) \right)
\end{multline*}
From this we have $\iota(E_\lambda(q))\mapsto e_\lambda$ and by dimensionality argument we have an isomorphism.
\end{proof}
\begin{cor}\label{cor:iso}
For the functions $S_k$, $E_k(t^{-1}),H_k(t^{-1})$ we have:
\begin{align}
    \label{eq:iota_S}
    &\iota: \quad S_k \mapsto \frac{(1-q)^k}{(t-q)^k} p_k 
     \\
\label{eq:iota_E}    
    &\iota: \quad E_k(t^{-1}) \mapsto \frac{(1-q)^k}{(1-t)^k} g_k  \\
\label{eq:iota_H}    
    &\iota: \quad H_k(t^{-1}) \mapsto \frac{(1-q)^k}{(t-1)^k} g_k^*  
\end{align}
\end{cor}
\begin{proof}
The first map \eqref{eq:iota_S} follows from \eqref{eq:iota_e} and matching the generating function of $E_k(q)$ and the generating function \eqref{eq:generating-e} of the elementary symmetric functions $e_k$:
\begin{align*}
\iota    \left(    E(v;q)\right)= \exp
    \left(
    \sum_{r>0}\frac{(-1)^{r+1}}{r} v^{r} p_r\right)
\end{align*}
Due to \eqref{eq:E-gen} and the fact that $\iota$ is an isomorphism it implies \eqref{eq:iota_S}. 
The two other maps \eqref{eq:iota_E} and \eqref{eq:iota_H} can be verified by applying $\iota^{-1}$ to the generating functions of $E_k(t^{-1})$ and $H_k(t^{-1})$ and then comparing the result with the generating functions \eqref{eq:generating-g} and \eqref{eq:generating-g-dual} of $g_k$ and $g_k^*$. 
\end{proof}
In Lemma \ref{lem:iso} and Corollary \ref{cor:iso} we omitted the dependence of symmetric functions of $\Lambda$ on the alphabet. This will be important in a later section. Consider a pair $F\in \Ac$ and $f\in \Lambda$ such that $\iota(F)=f$ and let $\Lambda$ be the ring of symmetric function in the alphabet $(z)=(z_1,z_2\ldots )$, then we will write:
$$
\iota_z(F)= f(z_1,z_2\ldots )
$$

\section{Six vertex model and the shuffle algebra \texorpdfstring{$\mathcal{A}^\circ$}{}}\label{sec:six_vertex}
In this section we explain one of the main results of the paper using the example of the six vertex model. We start with the Boltzmann weights of the fundamental $R$-matrix of $U_{t}(\widehat{sl}_2)$ and consider the associated square lattice partition functions. The configurations of these partition functions involve lattice paths of a single colour, labelled ``1''. The conic partition function, which is discussed in the introduction, can be expressed as a sum of planar partition functions with identified boundary conditions. In the algebraic language the conic partition function equals to the trace of a product of $R$-matrices. We use the notion of the {\it $F$-basis} to rewrite this trace as a symmetrization, w.r.t. the spectral parameters $x_1\ldots x_N$, of a fixed partition function with ordered labels on the boundaries. This produces a shuffle product formula for the conic partition function and leads us to a proof of Theorem \ref{thm:Z_intro} in the special case of $m=0$ and $n=1$.
\subsection{The six vertex model}
We consider a grid made up of a finite number of horizontal lines oriented from right to left and the same number of vertical lines oriented from top to bottom. We attach {\it spectral} the parameter $x_i(y_i)$  to each $i^{th}$ horizontal (vertical) line counting from the top (left):
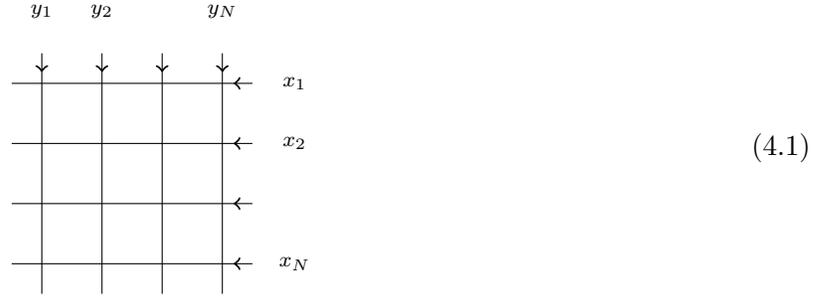
\begin{align}
\label{tikz:latticemodel}
\begin{tikzpicture}[scale=0.8,baseline=(current  bounding  box.center)]
\foreach\j/\lab in {1/,2/,3/,4/}
\draw[invarrow=0.95] (\j,0.5) -- node[pos=1,above] {$\scriptstyle \lab$} (\j,4.5);
\foreach\i/\lab in {4/,3/,2/,1/}
\draw[invarrow=0.95] (0.5,\i) -- node[pos=1,right] {$\scriptstyle \lab$} (4.5,\i);
\foreach\i/\lab in {1/x_N,2/ ,3/x_2 ,4/x_1}
\node at (5.2,\i) {$\scriptstyle \lab$};
\foreach\i/\lab in {1/y_1,2/y_2,3/ ,4/y_N}
\node at (\i,5.2) {$\scriptstyle \lab$};
\foreach\i/\lab in {4/,3/,2/ ,1/}
\node at (0.2,\i) {$\scriptstyle \lab$};
\foreach\i/\lab in {4/,3/,2/ ,1/}
\node at (\i,0.3) {$\scriptstyle \lab$};
\end{tikzpicture}
\end{align}
We are interested in the special case $y_i = q x_i$, however, the more general case which involves the $y$ parameters will be useful for computational purposes. We refer to an intersection of a horizontal and vertical line as a vertex. In this section every edge of a vertex can be labelled either $0$ or $1$ and to the edges carrying the label $1$ we will associate a red path. Specifying the boundary conditions in \eqref{tikz:latticemodel} means assigning labels to the $4N$ external edges. To every vertex, depending on the local configuration, we attach a Boltzmann weight:
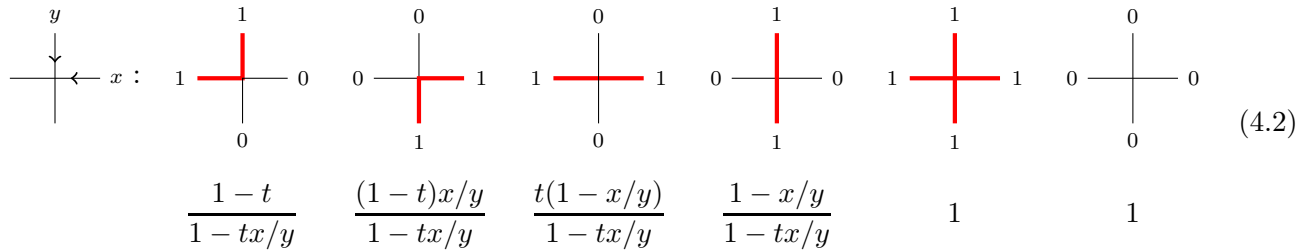
\begin{equation}\label{tikz:vertices}
\begin{tabular}{c@{\hskip 0.32cm}c@{\hskip 0.32cm}c@{\hskip 0.32cm}c@{\hskip 0.32cm}c@{\hskip 0.32cm}c@{\hskip 0.32cm}c}
\begin{tikzpicture}[scale=0.6,baseline=-2pt]
\draw[invarrow=0.75] (-1,0) --(1,0);
\draw[invarrow=0.75] (0,-1) --(0,1);
\node[right] at (1,0) {$\scriptstyle x$};
\node[above] at (0,1) {$\scriptstyle y$};
\end{tikzpicture}:
&
\begin{tikzpicture}[scale=0.6,baseline=-2pt]
\draw[red, ultra thick] (-1,0) node[left,black]{$\ss 1$}  -- (0,0)--(0,1) node[above,black]{$\ss 1$};
\draw (0,-1) node[below,black]{$\ss 0$} --(0,0)-- (1,0)node[right,black]{$\ss 0$};
\end{tikzpicture}
&
\begin{tikzpicture}[scale=0.6,baseline=-2pt]
\draw (-1,0) node[left,black]{$\ss 0$}  -- (0,0)--(0,1) node[above,black]{$\ss 0$};
\draw[red, ultra thick] (0,-1) node[below,black] {$\ss 1$} --(0,0)-- (1,0) node[right,black]{$\ss 1$};
\end{tikzpicture}
&
\begin{tikzpicture}[scale=0.6,baseline=-2pt]
\draw[red, ultra thick] (-1,0)  node[left,black]{$\ss 1$}-- (1,0) node[right,black]{$\ss 1$};
\draw (0,-1)  node[below,black]{$\ss 0$}--(0,1) node[above,black]{$\ss 0$};
\end{tikzpicture}
&
\begin{tikzpicture}[scale=0.6,baseline=-2pt]
\draw (-1,0) node[left,black]{$\ss 0$}  -- (1,0) node[right,black]{$\ss 0$};
\draw[red, ultra thick] (0,-1)node[below,black]{$\ss 1$}  --(0,1) node[above,black]{$\ss 1$};
\end{tikzpicture}
&
\begin{tikzpicture}[scale=0.6,baseline=-2pt]
\draw[ red,ultra thick] (-1,0) node[left,black]{$\ss 1$}  -- (1,0) node[right,black]{$\ss 1$};
\draw[red, ultra thick] (0,-1)node[below,black]{$\ss 1$}  --(0,1)node[above,black]{$\ss 1$};
\end{tikzpicture}&
\begin{tikzpicture}[scale=0.6,baseline=-2pt]
\draw (-1,0) node[left,black]{$\ss 0$}  -- (1,0) node[right,black]{$\ss 0$};
\draw (0,-1)node[below,black]{$\ss 0$}  --(0,1)node[above,black]{$\ss 0$};
\end{tikzpicture}
\\[3em]
& $\dfrac{1-t}{1-tx/y}$ & $\dfrac{(1-t)x/y}{1-t x/y}$& $\dfrac{t(1-x/y)}{1-t x/y}$&$\dfrac{1-x/y}{1-t x/y}$&$1$&$1$
\end{tabular}
\end{equation}
and the Boltzmann weights of all other local path configurations being zero. The leftmost vertex in \eqref{tikz:vertices} with oriented lines and unspecified edge labels represents the collection of all vertices with their Boltzmann weights and will be referred to as the graphical representation of the $\check R$-matrix of the six vertex model. When the values of the external edges of the graphical $\check R$-matrix are specified we will identify them with the corresponding Boltzmann weights in \eqref{tikz:vertices}. We can join edges of several graphical $\check R$-matrices together in which case we will assume that the values at the joined edges are summed over. 
Consider an example of \eqref{tikz:latticemodel} with $N=2$ and a choice of boundary conditions:
\begin{align}\label{eq:example_config}
\begin{tikzpicture}[scale=0.6,baseline=-2pt]
\node[right] at (3+0.5,0+0.5) {$\scriptstyle x_1$};
\node[right] at (3+0.5,-1+0.5) {$\scriptstyle x_2$};
\node[above] at (1,1+1) {$\scriptstyle y_1$};
\node[above] at (2,1+1) {$\scriptstyle y_2$};
\node[right] at (3,0+0.5) {$\ss 1$};
\node[right] at (3,-1+0.5) {$\ss 1$};
\node[left] at (0,0+0.5) {$\ss 1$};
\node[left] at (0,-1+0.5) {$\ss 0$};
\node[above] at (1,1+0.5) {$\ss 1$};
\node[above] at (2,1+0.5) {$\ss 0$};
\node[below] at (1,-2+0.5) {$\ss 1$};
\node[below] at (2,-2+0.5) {$\ss 1$};
\draw[invarrow=0.875] (0,0+0.5) --(3,0+0.5);
\draw[invarrow=0.875] (0,-1+0.5) --(3,-1+0.5);
\draw[arrow=0.125] (1,1+0.5) --(1,-2+0.5);
\draw[arrow=0.125] (2,1+0.5) --(2,-2+0.5);
\end{tikzpicture}
=
\begin{tikzpicture}[scale=0.6,baseline=-2pt]
\node[right] at (3,0+0.5) {$\ss 1$};
\node[right] at (3,-1+0.5) {$\ss 1$};
\node[left] at (0,0+0.5) {$\ss 1$};
\node[left] at (0,-1+0.5) {$\ss 0$};
\node[above] at (1,1+0.5) {$\ss 1$};
\node[above] at (2,1+0.5) {$\ss 0$};
\node[below] at (1,-2+0.5) {$\ss 1$};
\node[below] at (2,-2+0.5) {$\ss 1$};
\draw (0,0+0.5) --(3,0+0.5);
\draw (0,-1+0.5) --(3,-1+0.5);
\draw (1,1+0.5) --(1,-2+0.5);
\draw (2,1+0.5) --(2,-2+0.5);
\draw[red, ultra thick] (0,0+0.5) --(3,0+0.5);
\draw[red, ultra thick]  (1,1+0.5) --(1,-2+0.5);
\draw[red, ultra thick] (2,-2+0.5) --(2,-1+0.5) --(3,-1+0.5);
\end{tikzpicture}
&+
\begin{tikzpicture}[scale=0.6,baseline=-2pt]
\node[right] at (3,0+0.5) {$\ss 1$};
\node[right] at (3,-1+0.5) {$\ss 1$};
\node[left] at (0,0+0.5) {$\ss 1$};
\node[left] at (0,-1+0.5) {$\ss 0$};
\node[above] at (1,1+0.5) {$\ss 1$};
\node[above] at (2,1+0.5) {$\ss 0$};
\node[below] at (1,-2+0.5) {$\ss 1$};
\node[below] at (2,-2+0.5) {$\ss 1$};
\draw (0,0+0.5) --(3,0+0.5);
\draw (0,-1+0.5) --(3,-1+0.5);
\draw (1,1+0.5) --(1,-2+0.5);
\draw (2,1+0.5) --(2,-2+0.5);
\draw[red, ultra thick] (0,0+0.5) --(1,0+0.5) --(1,1+0.5);
\draw[red, ultra thick] (1,-2+0.5) -- (1,-1+0.5) -- (3,-1+0.5);
\draw[red, ultra thick] (2,-2+0.5) --(2,0+0.5) --(3,0+0.5);
\end{tikzpicture} \\
=
  \frac{t \left(1-x_1/y_2\right) 
  \left(1-x_2/y_1\right)(1-t) x_2/y_2}
   {\left(1- t x_1/y_2\right) \left(1- t x_2/y_1\right) \left(1-t x_2/y_2\right)}
 &  +
    \frac{(1-t)^3 x_1/y_2 \,x_2/y_1 }{\left(1-t x_1/y_1\right)  \left(1-t x_1/y_2\right)
    \left(1-t x_2/y_1\right)}
\nonumber
\end{align}
In this case, if we sum over all possible labels of the internal edges on the l.h.s. of \eqref{eq:example_config} we will find two configurations with non-zero  Boltzmann weights which we computed using \eqref{tikz:vertices} in the second line.
\begin{dfn}\label{def:Z}
Let $\alpha,\beta,\gamma,\delta\in\{0,1\}^N$ be collections of $0$'s and $1$'s. The six vertex partition function $Z_{\alpha,\gamma}^{\beta,\delta}(x;y):=Z_{\alpha,\gamma}^{\beta,\delta}(x_1\dots x_N;y_1\dots y_N)$ is defined as the rational function in the spectral parameters equal to the weighted sum over all possible six vertex configurations with the boundary conditions as specified below: 
\begin{align}\label{eq:Z_def}
Z_{\alpha,\gamma}^{\beta,\delta}(x;y)
=
\begin{tikzpicture}[scale=0.8, baseline=(current  bounding  box.center)]
\foreach\j/\lab in {1/\beta_1,2/\beta_2,3/,4/\beta_N}
\draw[invarrow=0.95] (\j,0.5) -- node[pos=1,above] {$\scriptstyle \lab$} (\j,4.5);
\foreach\i/\lab in {4/\delta_1,3/\delta_2,2/,1/\delta_N}
\draw[invarrow=0.95] (0.5,\i) -- node[pos=1,right] {$\scriptstyle \lab$} (4.5,\i);
\foreach\i/\lab in {4/x_1,3/x_2,2/ ,1/x_N}
\node at (5.5,\i) {$\scriptstyle \lab$};
\foreach\i/\lab in {1/y_1,2/y_2,3/ ,4/y_N}
\node at (\i,5.3) {$\scriptstyle \lab$};
\foreach\i/\lab in {4/\alpha_1,3/\alpha_2,2/ ,1/\alpha_N}
\node at (0.2,\i) {$\scriptstyle \lab$};
\foreach\i/\lab in {1/\gamma_1,2/\gamma_2,3/ ,4/\gamma_N}
\node at (\i,0.3) {$\scriptstyle \lab$};
\end{tikzpicture}
\end{align}
\end{dfn} 
Let us turn to the algebraic picture. We define the six vertex $\check R$-matrix:
\begin{align}\label{eq:Rc}
    \check R(x/y) := \left(
\begin{array}{cccc}
 1 & 0 & 0 & 0 \\
 0 & \dfrac{(1-t) x/y}{1-t x/y} & \dfrac{1-x/y}{1-t x/y} & 0 \\
 0 & \dfrac{t (1-x/y)}{1-t x/y} & \dfrac{1-t}{1-t x/y} & 0 \\
 0 & 0 & 0 & 1 \\
\end{array}
\right)
\end{align}
This matrix acts in $V_y\otimes V_x$ with $V_y,V_x \simeq \mathbb C^2$. Let $P$ be the permutation matrix acting in $\mathbb C^2\otimes \mathbb C^2$ by swapping the basis vectors, then we define the $R$-matrix:
\begin{align}\label{eq:R}
R(x) = P\check R(x)    
\end{align}
Let $\ket{0}=(1,0)^T$ and $\ket{1}=(0,1)^T$ denote the standard basis in $\mathbb C^2$ and $\ket{i_1\dots i_N}$, with $i_1,\dots, i_N\in \{0,1\}$, its generalization to the $N$-fold tensor product of $\mathbb C^2$. Define similarly the dual basis, then we have:
\begin{align}
    \label{eq:R-Rg}
    \check R(x/y) = \sum_{a,b,c,d=0,1}
\left[
\begin{tikzpicture}[scale=0.6,baseline=-2pt]
\draw[invarrow=0.75] (-1,0) --(1,0);
\draw[invarrow=0.75] (0,-1) --(0,1);
\node[right] at (1+0.5,0) {$\scriptstyle x$};
\node[above] at (0,1+0.5) {$\scriptstyle y$};
\node[below] at (0,-1) {$\scriptstyle c$};
\node[left] at (-1,0) {$\scriptstyle a$};
\node[right] at (1,0) {$\scriptstyle d$};
\node[above] at (0,1) {$\scriptstyle b$};
\end{tikzpicture}
\right]
 \ket{a,c}\bra{b,d}
\end{align}
Let $\id$ be the identity matrix in $\mathbb C^2$ then $\check R$ acts in $\otimes_{i=1}^N V_{x_i}$ by:
\begin{align}\label{eq:R_i}
    \check R_i(x_{i+1}/x_i) := \underbrace{\id \otimes \cdots \otimes\id}_{i-1} \otimes 
    \check R(x_{i+1}/x_i) \otimes \underbrace{\id \otimes \cdots \otimes\id}_{N-i-1}  
\end{align}
and similarly we define $R_i(x_{i+1}/x_i)\in \text{End}\left(\otimes_{i=1}^N V_{x_i}\right)$ and the permutation matrix $P_i$. We have the Yang--Baxter equation:
\begin{align}
    \label{eq:YB}
    \check R_{i}(z/y)
    \check R_{i+1}(z/x)
    \check R_{i}(y/x)
    =
    \check R_{i+1}(y/x)
    \check R_{i}(z/x)
    \check R_{i+1}(z/y)
\end{align}
and the unitarity relation:
\begin{align}
    \label{eq:unitarity}
    \check R_{i}(x/y)
    \check R_{i}(y/x)
    = \text{id}\otimes \text{id}
\end{align}
The graphical notation associated to the matrix $\check R_i(x_{i+1}/x_i)$ is:
\begin{align}
\label{eq:R_i-graph}
\check R_i(x_{i+1}/x_i):\quad
\begin{tikzpicture}[baseline=0.5cm]
\draw[invarrow] (0,0) -- (0,1);
\node at (1,0.5) {$\cdots$};
\draw[invarrow] (2,0) -- (2,1);
\draw[invarrow=0.75] (3,0) -- node[pos=0.75,right] {$\scriptstyle x_{i+1}$} (4,1);
\draw[invarrow=0.75] (4,0) -- node[pos=0.75,left] {$\scriptstyle x_i$} (3,1);
\draw[invarrow] (5,0) -- (5,1);
\node at (6,0.5) {$\cdots$};
\draw[invarrow] (7,0) -- (7,1);
\draw[decorate,decoration=brace] (2,-0.2) -- node[below] {$i-1$} (0,-0.2);
\draw[decorate,decoration=brace] (7,-0.2) -- node[below] {$N-i-1$} (5,-0.2);
\end{tikzpicture}
\end{align}
where the rotation of the cross compared to \eqref{tikz:vertices} does not present an ambiguity due to the presence of the arrows. The arrows also help us to keep track of the ordering of operators: moving forward w.r.t. the orientation of a line is reading an expression right to left. With this convention the graphical version of the Yang--Baxter equation \eqref{eq:YB} reads:
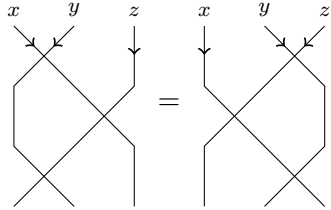
\[
\begin{tikzpicture}[scale=0.8, baseline=0.5cm]
\node[above] at (0,3-1) {$\scriptstyle x$};
\node[above] at (1,3-1) {$\scriptstyle y$};
\node[above] at (2,3-1) {$\scriptstyle z$};
\draw[arrow=0.125] (0,3-1) -- (2,1-1) -- (2,0-1);
\draw[arrow=0.125] (1,3-1) -- (0,2-1) -- (0,1-1) -- (1,0-1);
\draw[arrow=0.125] (2,3-1) -- (2,2-1) -- (0,0-1);
\end{tikzpicture}
=
\begin{tikzpicture}[scale=0.8, baseline=0.5cm]
\node[above] at (0,3-1) {$\scriptstyle x$};
\node[above] at (1,3-1) {$\scriptstyle y$};
\node[above] at (2,3-1) {$\scriptstyle z$};
\draw[arrow=0.125] (0,3-1) -- (0,2-1) -- (2,0-1);
\draw[arrow=0.125] (1,3-1) -- (2,2-1) -- (2,1-1) --(1,0-1);
\draw[arrow=0.125] (2,3-1) -- (0,1-1) -- (0,0-1);
\end{tikzpicture}
\]

Using the correspondence of the $\check R_i$-matrix in \eqref{eq:R_i} and its graphical counterpart in \eqref{eq:R_i-graph}, we can view the (rotated counterclockwise by $\pi/4$) object in \eqref{tikz:latticemodel} as a tensor. In what follows, we define three such tensors: $Z_N(x;y),W_N(x;y)$ and $W_N(x)$, where $W_N(x)$ is a special case of $W_N(x;y)$ which in turn is a projection of $Z_N(x;y)$. The reason for this is that we will be interested in the matrix elements of $W_N(x)$ and will require their properties, which can be determined from the more general objects $Z_N(x;y)$ and $W_N(x;y)$.

\begin{dfn}
For two sets of parameters $(x)=(x_1\dots x_N)$ and $(y)=(y_1\dots y_N)$ we define the following tensors:
\begin{align}\label{eq:defZ}
Z_N(x;y)&
:=\check R_N(x_N/y_1)\check R_{N-1}(x_{N-1}/y_1)\check R_{N+1}(x_N/y_2)
\cdots
\check R_N(x_1/y_N)\\
\label{eq:defWxy}
W_N(x;y)&:=\left(\bra{0^N}\otimes \cdot\right) Z_N(x;y) \left(\ket{0^N}\otimes \cdot \right)\\
\label{eq:defW}
W_N(x)&:=W(x_1\ldots x_N;q x_1\ldots q x_N)
\end{align}
\end{dfn}
The graphical representation for the matrix elements of $Z_N(x;y)$ follows from Definition \ref{def:Z}:
\begin{align}
    \label{eq:Zel}
Z_N(x;y) &= 
\sum_{\alpha,\beta,\gamma,\delta}
Z_{\alpha,\gamma}^{\beta,\delta}(x;y)
\ket{\alpha,\gamma}\bra{\beta,\delta}
\end{align}
By setting $W_{\gamma}^{\delta}(x;y):=Z_{(0^N),\gamma}^{(0^N),\delta}(x;y)$ and $W_{\gamma}^{\delta}(x):=W_{\gamma}^{\delta}(\ldots x_i \ldots;\ldots q x_i \ldots)$ we also have:
\begin{align}
    \label{eq:Wel}
W_N(x;y) = 
\sum_{\gamma,\delta}
W_{\gamma}^{\delta}(x;y)
\ket{\gamma}\bra{\delta},
\qquad
W_N(x) = 
\sum_{\gamma,\delta}
W_{\gamma}^{\delta}(x)
\ket{\gamma}\bra{\delta}
\end{align}
and the graphical representations of $W_{\gamma}^{\delta}(x;y)$ and $W_{\gamma}^{\delta}(x)$ are given by \eqref{eq:Z_def} but with the specialized labels $\alpha_i=0$ and $\beta_i=0$ and parameters $y_i=q x_i$ in the case of $W_{\gamma}^{\delta}(x)$. 
\begin{lemma}\label{lem:ZW-exchange}
For $i=1\ldots N-1$,
the tensor $Z_N(x;y)$ satisfies the following exchange relations:
\begin{align}
    \label{eq:exx-six}
&\check R_i(x_{i+1}/x_i) Z_N(x;y)=
    Z_N(\ldots x_{i+1},x_i \ldots;y) \check R_{N+i}(x_{i+1}/x_i)
\\
    \label{eq:exy-six}
&Z_N(x;y)\check R_i(y_i/y_{i+1})
=
\check R_{N+i}(y_i/y_{i+1}) Z_N(x;\ldots y_{i+1},y_{i} \ldots)
\end{align}
The tensor $W_N(x;y)$ satisfies:
\begin{align}
    \label{eq:exWxy1-six}
W_N(x;y)
&=    W_N(\ldots x_{i+1},x_i \ldots;y) \check R_{i}(x_{i+1}/x_i)\\
    \label{eq:exWxy2-six}
 W_N(x;y)
&=  \check R_{i}(y_i/y_{i+1})  W_N(x;\ldots y_{i+1},y_i\ldots)     
\end{align}
and for $W_N(x)$ we have:
\begin{align}
\label{eq:exW-six}
\check R_{i}(x_{i+1}/x_i) W_N(x)
=    W_N(\ldots x_{i+1},x_i\ldots) \check R_{i}(x_{i+1}/x_i)
\end{align}
\end{lemma}
\begin{proof}
The two equations \eqref{eq:exx-six} and \eqref{eq:exy-six} are well-known exchange relations which are a consequence of the Yang--Baxter equation \eqref{eq:YB} (see e.g. \cite{GZJ}). The exchange relations for $W_N(x;y)$ \eqref{eq:exWxy1-six} and \eqref{eq:exWxy2-six} follow from \eqref{eq:exx-six} and \eqref{eq:exy-six} by applying the projection in \eqref{eq:defWxy} and using: 
$$
\bra{0^N}\check R_i(x)=\bra{0^N}, 
\qquad
\check R_i(x)\ket{0^N}=\ket{0^N}
$$
After relabelling the vector spaces $N+i\rightarrow i$ we get  \eqref{eq:exWxy1-six} and \eqref{eq:exWxy2-six}. 
The last equation  \eqref{eq:exW-six} is obtained by equating the right hand sides of \eqref{eq:exWxy1-six} and \eqref{eq:exWxy2-six}, swapping $x_i \leftrightarrow x_{i+1}$:
$$
W_N(x;y) \check R_{i}(x_i/x_{i+1})
= \check R_{i}(y_i/y_{i+1})  W_N(\ldots x_{i+1},x_i \ldots;\ldots y_{i+1},y_i\ldots)     
$$
multiplying both sides by $\check R_{i}(y_{i+1}/y_i)$ on the left and by $\check R_{i}(x_{i+1}/x_i)$ on the right, using the unitarity property \eqref{eq:unitarity} and then setting $y_i=q x_i$ for all $i$. 
\end{proof}

\subsection{The \texorpdfstring{$F$}{}-matrix and  transformed tensors}\label{sec:F}
In this section we introduce the $F$-matrix  \cite{MdS,ABFR}, following the conventions of \cite{BWcoloured}. Then we will use it to transform the tensors $Z_N(x;y),W_N(x;y)$ and $W_N(x)$ and establish their properties.
\begin{dfn}
The $2$-site $F$-matrix reads:
\begin{align}\label{eq:F2-matrix}
    F_{2}(x_1,x_2) =\left(
\begin{array}{cccc}
 1 & 0 & 0 & 0 \\
 0 & 1 & 0 & 0 \\
 0 & \frac{(1-t) x_2/x_1}{1-t x_2/x_1} & \frac{1-x_2/x_1}{1-t x_2/x_1} & 0 \\
 0 & 0 & 0 & 1 \\
\end{array}
\right)
\end{align}
This matrix satisfies the property:
\begin{align}\label{eq:F2-prop}
F_2(x_1,x_2) = P F_2(x_2,x_1)\check R(x_2/x_1) 
\end{align}
\end{dfn}
The 2-site $F$-matrix and the property \eqref{eq:F2-prop} are associated to the vector space $V_{x_1}\otimes V_{x_2}$. The $N$-site $F$-matrix and the analogue of \eqref{eq:F2-prop} correspond to $\otimes_{i=1}^N V_{x_i}$. This generalization is done
with the help of the matrices $\check R_\sigma$ and $R_\sigma$ for $\sigma \in \mathcal{S}_N$. Consider $\sigma \in \mathcal{S}_N$ and its decomposition into simple transpositions. For some $i$ and $\sigma' \in \mathcal{S}_N$ we can write\footnote{In our conventions for $\sigma=s_{i} \sigma'$ and $j=1\ldots N$ we have $\sigma(j)=s_i(\sigma'(j))$.}:
\begin{align}
    \sigma=s_{i} \sigma'
\end{align}
Define recursively the matrices $\check R_{\sigma}$ and $R_{\sigma}$:
\begin{align}
\label{eq:Rch-sigma}
&\check R_{\text{id}} = 1,\qquad    
\check R_{\sigma} = \check R_{i}(x_{\sigma'(i+1)}/x_{\sigma'(i)})
   \check R_{\sigma'}    \\
\label{eq:R-sigma}
&R_{\text{id}} = 1,\qquad    R_{\sigma} = R_{\sigma'(i),\sigma'(i+1)}(x_{\sigma'(i+1)}/x_{\sigma'(i)})
    R_{\sigma'}
\end{align}
where the two-index notation $R_{i,j}(x)$ is defined by setting $R_{i,i+1}(x)=R_i(x)$ and $R_{i,j+1}(x)=P_{j}R_{i,j}(x)P_j$. Next we consider an explicit decomposition $\sigma=s_{i_k}\cdots s_{i_1}$ and define:
\begin{align}
    \label{eq:Psig}
    P_\sigma := P_{i_1}\cdots P_{i_k}
\end{align}
This factorization is opposite to $\sigma=s_{i_k}\cdots s_{i_1}$ which means that when $P_\sigma$ acts on vectors $\ket{\alpha}$ in $\otimes_{i=1}^N V_{x_i}$ it permutes the label $\alpha$ according to $\sigma^{-1}$. With the above definitions we have $P_\sigma=R_\sigma|_{x_1=\dots=x_N}$ and the relation:
\begin{align}\label{eq:check_R-R_sigma}
    \check{R}_\sigma = P_{\sigma^{-1}} R_\sigma
\end{align}
\begin{dfn}
For $k,l\in \{0,1\}$ and $r=1\dots N$, let $E^{(k l)}_{r}\in \normalfont{ \text{End}}\left(\otimes_{i=1}^N V_{x_i}\right)$ be the matrix acting non-trivially on the $r$-th tensor space $V_{x_r}$ by the matrix unit $E^{(k l)}$ (the two by two matrix with a single non-zero entry in the $k$-th row and $j$-th column which is equal to $1$). The $N$-site $F$-matrix reads:
\begin{align}\label{eq:F-matrix}
    F_{N}(x_1\dots x_N) = 
    \sum_{\sigma \in \mathcal{S}_N} \sum_{(k_1\dots k_N)\in \mathcal{J}_\sigma} \prod_{i=1}^N E^{(k_i k_i)}_{\sigma(i)} R_\sigma
\end{align}
where the set $\mathcal{J}_\sigma$ is defined by:
\begin{align}\label{eq:J-set}
    \mathcal{J}_\sigma = 
    \Big{\{}
    0\leq k_1\leq \dots \leq k_N\leq 1:
    k_i<k_{i+1} ~~~\text{if}~~~ \sigma(i) >\sigma(i+1)
    \Big{\}}
\end{align}
This matrix satisfies the property:
\begin{align}\label{eq:F-prop}
F_N(x_1\dots x_N) = P_\sigma F_N(x_{\sigma(1)}\ldots x_{\sigma(N)})\check R_\sigma
\end{align}
\end{dfn}
\begin{lemma}
The $F$-matrix has the properties:
\begin{align}
    \label{eq:F-diag}
     \bra{0^{N-k}1^k} F_N(x) =  \bra{0^{N-k}1^k},
 \qquad
 F_N(x) \ket{1^k0^{N-k}} = \prod_{i=1}^k
 \prod_{j=k+1}^N \frac{x_i-x_j}{x_i-t x_j}\ket{1^k0^{N-k}} 
\end{align}
\end{lemma}
These are known properties which are a consequence of  the explicit form of the $F$-matrix \eqref{eq:F-matrix} (for more details see e.g. \cite{McW}).
\begin{dfn}
Introduce transformed tensors $\widetilde Z_N(x;y)$, $\widetilde W_N(x;y)$ and $\widetilde W_N(x)$:
\begin{align}
    \label{eq:Ztilde}
\widetilde{Z}_{N}(x;y) &:= 
\left(F_{N}(x)\otimes
F_{N}(y)\right)
Z_{N}(x;y) 
\left(F^{-1}_{N}(y)\otimes 
F^{-1}_{N}(x)\right)\\
    \label{eq:Wytilde}
\widetilde{W}_{N}(x;y) &:= 
F_{N}(y)
W_{N}(x;y) 
F^{-1}_{N}(x)\\
    \label{eq:Wtilde}
\widetilde{W}_{N}(x) &:= 
F_{N}(x)
W_{N}(x) 
F^{-1}_{N}(x)
\end{align}    
\end{dfn}
The definition of $\widetilde W_N(x;y)$ is consistent with \eqref{eq:defWxy} because of the property \eqref{eq:F-diag} at $k=0$, in other words we have:
\begin{align}\label{eq:tildeWZ}
\widetilde W_N(x;y)=\left(\bra{0^N}\otimes \cdot\right) 
\widetilde Z_N(x;y) \left(\ket{0^N}\otimes \cdot \right)
\end{align}
Due to the conjugation by the $F$-matrices the tensors $\widetilde Z_N(x;y), \widetilde W_N(x;y)$ and $\widetilde W_N(x)$ obey new exchange relations.
\begin{prop}\label{prop:tilde-ex}
For $i=1 \ldots N-1$,
the tensor $\widetilde Z_N(x;y)$ satisfies the following exchange relations:
\begin{align}\label{eq:Ztilde-exx}
P_i \widetilde  Z(x;y)&=
    \widetilde  Z(\ldots x_{i+1},x_i \ldots;y) P_{N+i}
\\
    \label{eq:Ztilde-exy}
\widetilde  Z(x;y) P_i
&=
P_{N+i}\widetilde  Z(x;\ldots y_{i+1},y_{i} \ldots) 
\end{align}
the tensor $\widetilde W_N(x;y)$ satisfies:
\begin{align}
    \label{eq:Wytilde-ex1}
\widetilde  W_N(x;y)P_{i}
&=   \widetilde  W_N(\ldots x_{i+1},x_i \ldots;y)\\
    \label{eq:Wytilde-ex2}
P_i\widetilde  W_N(x;y)
&=   \widetilde  W_N(x;\ldots y_{i+1},y_i\ldots)
\end{align}
and $\widetilde  W_N(x)$ satisfies:
\begin{align}
\label{eq:Wtilde-ex}
P_i \widetilde W_N(x)
&=\widetilde  W_N(\ldots x_{i+1},x_i \ldots)P_i
\end{align}
\end{prop}
\begin{proof}
Consider two permutations $\sigma,\tau \in \mathcal{S}_{N}$ and their decompositions into simple transpositions, denote $x_\sigma =(x_{\sigma(1)}\ldots x_{\sigma(N)})$ and $y_\tau =(y_{\tau(1)}\ldots x_{\tau(N)})$, then we have:
\begin{align}\label{eq:Z-sigma}
    \left(\check R_\sigma(x) \otimes \check R_\tau(y) \right)Z_N(x;y)=
    Z_N(x_\sigma;y_\tau)
    \left(\check R_\tau(y) \otimes \check R_\sigma(x) \right)
\end{align}
This follows from the recursive definition of $\check R_\sigma$ and repetitive application of \eqref{eq:exx-six} and \eqref{eq:exy-six}. Multiply both sides of \eqref{eq:Z-sigma} on the left and on the right by $F$-matrices as follows:
$$
\left(F_N(x_\sigma)
\otimes 
F_N(y_\tau)\right) \eqref{eq:Z-sigma} \left(
F_N^{-1}(y) \otimes 
F_N^{-1}(x) \right)
$$
and use \eqref{eq:F-prop} and \eqref{eq:Ztilde}, the outcome is:
\begin{align}\label{eq:Zt-sigma}
 \left(P_{\sigma^{-1}}\otimes P_{\tau^{-1}}\right)\widetilde Z_N(x;y)=
    \widetilde Z_N(x_\sigma;y_\tau)
    \left(P_{\tau^{-1}}  \otimes P_{\sigma^{-1}} \right)
\end{align}
where on the r.h.s. we used \eqref{eq:F-prop} in a rearranged form:
\begin{align*}
\check R_\sigma(x) F_N^{-1}(x)=F_N^{-1}(x_\sigma)P_{\sigma^{-1}}
\end{align*}
Clearly \eqref{eq:Zt-sigma} is equivalent to the pair of equations \eqref{eq:Ztilde-exx} and \eqref{eq:Ztilde-exy}. The remaining equations \eqref{eq:Wytilde-ex1}, \eqref{eq:Wytilde-ex2} and \eqref{eq:Wtilde-ex} are special cases of the above due to \eqref{eq:tildeWZ} and $\widetilde W_N(x) = \widetilde W_N(\ldots x_i\ldots;\ldots q x_i\ldots )$.
\end{proof}
In the remaining part of this subsection we compute the matrix elements of $\widetilde W_N(x;y)$. The formula which we obtain is a product consisting of several building blocks of fully factorized terms and a {\it domain-wall} partition function: 
\begin{align}
\label{eq:DW}
D_{M}(x;y):=W_{(1^M)}^{(1^M)}(x_1\ldots x_M;y_1\ldots y_M).    
\end{align}
The following Lemma is a well-known result about the six vertex domain-wall partition function \cite{Kor82,Iz87}.
\begin{lemma}\label{lem:IK}
$D_{M}(x;y)$ is the six vertex domain-wall partition function which can be computed using a determinant formula:
\begin{align}
    \label{eq:IK_det}
    D_{M}(x;y) = \frac{\prod_{i,j=1}^M(x_i-y_j)}{\prod_{1\leq i<j\leq M}(x_i-x_j)(y_j-y_i)}
    \det_{1\leq i,j \leq M} \frac{(1-t)x_i}{(x_i-y_j)(y_j-t x_i)}
\end{align}    
\end{lemma}
\begin{prop}
    \label{prop:Wy}
Fix $k \in \{0\ldots N\}$ and let $\alpha,\beta\in\{0,1\}^N$ be the labels of the matrix elements $\widetilde W_{\alpha}^{\beta}(x;y)$ of $\widetilde W_N(x;y)$ both having $k$ number of $1$'s. Let $P,S\subseteq [N]$ denote the positions of $1$'s in $\alpha$ and $\beta$ respectively, then:
\begin{align}
    \label{eq:Wy-mat}
\widetilde W_{\alpha}^{\beta}(x;y)= 
D_{k} (x_S;y_P) \prod_{i\in S}\prod_{j\in S^c}
\frac{x_i-t x_j}{x_i-x_j}
\prod_{i\in P}\prod_{j\in S^c}
\frac{y_i-x_j}{y_i-t x_j}
\end{align}
\end{prop}
\begin{proof}
The proof is based on the properties \eqref{eq:Wytilde-ex1} and \eqref{eq:Wytilde-ex2} which imply that for a fixed $k$ it is sufficient to compute any one matrix element $\widetilde W_{\alpha}^{\beta}(x;y)$. There exists a choice of $\alpha$ and $\beta$ for which 
the matrix element $\widetilde W_{\alpha}^{\beta}(x;y)$ can be easily shown to be of the form \eqref{eq:Wy-mat}. Let us explain this in detail. 

Consider two permutations $\sigma,\tau \in \mathcal{S}_N$ such that $\alpha=\sigma^{-1}(0^{N-k}1^k)$ and $\beta =\tau^{-1}(0^{N-k}1^k)$. We can represent the vectors $\bra{\alpha}$ and $\ket{\beta}$ as follows:
$$
\bra{\alpha} = 
\bra{0^{N-k}1^k}P_{\sigma^{-1}},
\qquad
\ket{\beta}= 
P_{\tau}\ket{0^{N-k}1^k}
$$
therefore:
\begin{align*}
\widetilde W_{\alpha}^{\beta}(x;y)
= 
\bra{\alpha}\widetilde W_N(x;y) \ket{\beta}
&=
\bra{0^{N-k}1^k}P_{\sigma^{-1}}
\widetilde W_N(x;y)
P_{\tau}\ket{0^{N-k}1^k}\\
&=
\bra{0^{N-k}1^k}
\widetilde W_N(x_\tau;y_\sigma) \ket{0^{N-k}1^k}
=\widetilde W_{(0^{N-k}1^k)}^{(0^{N-k}1^k)}(x_\tau;y_\sigma)
\end{align*}
where in the second line we used the properties \eqref{eq:Wytilde-ex1} and \eqref{eq:Wytilde-ex2} of $\widetilde W_{\alpha}^{\beta}(x;y)$ (cf. \eqref{eq:Zt-sigma}). This shows that computing one matrix element of $\widetilde W_N(x;y)$  with a fixed number of $1$'s is sufficient to recover all of them. Consider next the matrix element $\widetilde W_{(0^{N-k} 1^k)}^{(1^k0^{N-k})}(x;y)$ and write it in terms of $W_{(0^{N-k} 1^k)}^{(1^k0^{N-k})}(x;y)$ using \eqref{eq:Wytilde}:
\begin{align*}
    \widetilde W_{(0^{N-k} 1^k)}^{(1^k0^{N-k})}(x;y)
    &=\bra{0^{N-k}1^k}F_{N}(y)W_{N}(x;y) F^{-1}_{N}(x) \ket{1^k0^{N-k}}\\
    =\prod_{i=1}^k\prod_{j=k+1}^N \frac{x_i-t x_j}{x_i-x_j}
    &\times
    \bra{0^{N-k}1^k}W_{N}(x;y)\ket{1^k0^{N-k}}
    =
    \prod_{i=1}^k\prod_{j=k+1}^N \frac{x_i-t x_j}{x_i-x_j}\times
    W_{(0^{N-k} 1^k)}^{(1^k0^{N-k})}(x;y)
\end{align*}
where in the second line we used the properties \eqref{eq:F-diag} of the $F$-matrix. The matrix element $W_{(0^{N-k} 1^k)}^{(1^k0^{N-k})}(x;y)$ can be evaluated as follows. For demonstration, we choose $N=4$ and $k=2$. The graphical depiction of the partition function $W_{(0011)}^{(1100)}(x;y)$ shows several regions which are frozen and one region with domain-wall configurations (enclosed in the dotted square below):
\begin{align}
\nonumber
W_{(0011)}^{(1100)}(x;y)
&=
\begin{tikzpicture}[scale=0.8, baseline=(current  bounding  box.center)]
\foreach\j/\lab in {1/0,2/0,3/0,4/0}
\draw[invarrow=0.95] (\j,0.5) -- node[pos=1,above] {$\scriptstyle \lab$} (\j,4.5);
\foreach\i/\lab in {4/1,3/1,2/0,1/0}
\draw[invarrow=0.95] (0.5,\i) -- node[pos=1,right] {$\scriptstyle \lab$} (4.5,\i);
\foreach\i/\lab in {4/x_1,3/x_2,2/x_3 ,1/x_4}
\node at (5.5,\i) {$\scriptstyle \lab$};
\foreach\i/\lab in {1/y_1,2/y_2,3/y_3,4/y_4}
\node at (\i,5.3) {$\scriptstyle \lab$};
\foreach\i/\lab in {4/0,3/0,2/0,1/0}
\node at (0.2,\i) {$\scriptstyle \lab$};
\foreach\i/\lab in {1/0,2/0,3/1,4/1}
\node at (\i,0.3) {$\scriptstyle \lab$};
\node at (3,1.5) {$\scriptstyle 1$};
\node at (4,1.5) {$\scriptstyle 1$};
\node at (3,2.5) {$\scriptstyle 1$};
\node at (4,2.5) {$\scriptstyle 1$};
\foreach\i in {1,2,3}
\node at (\i+0.5,1) {$\scriptstyle 0$};
\foreach\i in {1,2,3}
\node at (\i+0.5,2) {$\scriptstyle 0$};
\foreach\i in {1,2}
\node at (\i+0.5,3) {$\scriptstyle 0$};
\foreach\i in {1,2}
\node at (\i+0.5,4) {$\scriptstyle 0$};
\foreach\i in {1,2,3}
\node at (1,\i+0.5) {$\scriptstyle 0$};
\foreach\i in {1,2,3}
\node at (2,\i+0.5) {$\scriptstyle 0$};
\draw[dotted] (2+0.5-0.1,2+0.5-0.1) -- (2+0.5-0.1,4+1) -- (4+1,4+1) -- (4+1,2+0.5-0.1) -- (2+0.5-0.1,2+0.5-0.1);
\end{tikzpicture}\\
=
\prod_{i=\{3,4\}}\prod_{j=\{3,4\}}
\frac{y_i-x_j}{y_i-t x_j}
&\times 
\begin{tikzpicture}[scale=0.8, baseline=(current  bounding  box.center)]
\foreach\j/\lab in {1/0,2/0}
\draw[invarrow=0.9] (\j,0.5) -- node[pos=1,above] {$\scriptstyle \lab$} (\j,2.5);
\foreach\i/\lab in {2/1,1/1}
\draw[invarrow=0.9] (0.5,\i) -- node[pos=1,right] {$\scriptstyle \lab$} (2.5,\i);
\foreach\i/\lab in {2/x_1 ,1/x_2}
\node at (3.2,\i) {$\scriptstyle \lab$};
\foreach\i/\lab in {1/y_3,2/y_4}
\node at (\i,3.2) {$\scriptstyle \lab$};
\foreach\i/\lab in {2/0,1/0}
\node at (0.2,\i) {$\scriptstyle \lab$};
\foreach\i/\lab in {1/1,2/1}
\node at (\i,0.3) {$\scriptstyle \lab$};
\end{tikzpicture}
=
\prod_{i\in\{3,4\}}\prod_{j\in\{3,4\}}
\frac{y_i-x_j}{y_i-t x_j}
D_{2}(x_1,x_2;y_3,y_4)
\label{eq:W0011}
\end{align}
In the first line of the above equation we observed that the left half of the lattice must have only trivial local configurations corresponding to the last vertex in \eqref{tikz:vertices} whose weight is equal to $1$; the bottom right $2\times 2$ region contains four local configurations all given by the $4$-th vertex in \eqref{tikz:vertices} and therefore they contribute four factors appearing in the second line; finally the top right region, enclosed in the dotted square, equals to the partition function $D_{2}(x_1,x_2;y_3,y_4)$ \eqref{eq:DW}. The labels of $x$'s and $y$'s which appear in \eqref{eq:W0011} in $D_2(x_1,x_2;y_3,y_4)$ are determined by the locations of $1$'s on right and bottom boundaries respectively while the labels of $x$'s and $y$'s in the product appearing in \eqref{eq:W0011} are determined by the positions of $0$'s on the right boundary and positions of $1$'s on the bottom boundary respectively. Thanks to the property:
\begin{align*}
\widetilde W_{\sigma^{-1}(0^{N-k}1^k)}^{\tau^{-1}(0^{N-k}1^k)}(x;y)
=\widetilde W_{(0^{N-k}1^k)}^{(0^{N-k}1^k)}(x_\tau;y_\sigma)
\end{align*}
one needs to keep track of the positions of $1$'s to be able to assign the correct labels to $x$'s and $y$'s. This explains the labelling by the sets $S$ and $P$ given in the statement of the proposition. 
\end{proof}

\subsection{Trace of \texorpdfstring{$W_N(x)$}{} and the shuffle product}\label{sec:trace}
In this section, we define the partition function $T_N(x)$ as the trace of the matrix $W_N(x)$. Because of the cyclicity of the trace we will be able to replace $W_N(x)$ with $\widetilde W_N(x)$ since they are related by a conjugation \eqref{eq:Wtilde}. The trace represents a summation over indices of the matrix $\widetilde W_N(x)$, using the exchange relation \eqref{eq:Wtilde-ex} we can replace this summation by a symmetrization over the parameters $x_1\ldots x_N$ acting on specific diagonal elements of $\widetilde W_N(x)$. These diagonal elements can be computed via the connection with $\widetilde W_N(x;y)$. The result of this computation, summarized in Theorem \ref{thm:T-shuffle}, gives a formula for $T_N(x)$ in terms of the shuffle product of $\Ac$. 
\begin{dfn}\label{def:T}
Let $z_0,z_1$ be two indeterminates. Define the following function:
    \begin{align}\label{eq:T_W}
        T_N(x;z_0,z_1) := \sum_{\alpha\in\{0,1\}^N} z_{\alpha_1}\cdots z_{\alpha_N}
        \bra{\alpha} W_N(x)\ket{\alpha}
    \end{align}
\end{dfn}
Note that $W_N(x)$ is block diagonal with blocks labelled by the number of $1$'s in $\alpha$ and $\beta$ in the matrix elements $W_{\alpha}^\beta(x)$. This means that $T_N(x;z_0,z_1)$ is the trace of $W_N(x)$ where the monomials in $z_0,z_1$ parameterize different blocks of $W_N(x)$.
We can write $T_N(x;z_0,z_1)$ in terms of the matrix elements of $W_N(x)$:
\begin{align}
        T_N(x;z_0,z_1) = \sum_{\alpha\in\{0,1\}^N} z_{\alpha_1}\cdots z_{\alpha_N} W_{\alpha}^{\alpha}(x)
\end{align}
and using the graphics \eqref{eq:Z_def} we interpret $T_N(x;z_0,z_1)$ as the following lattice partition function:
\begin{align}\label{eq:T_def}
T_N(x;z_0,z_1)
=
\begin{tikzpicture}[scale=0.8, baseline=(current  bounding  box.center)]
\foreach\j/\lab in {1/0,2/0,3/,4/0}
\draw[invarrow=0.95] (\j,0.5) -- node[pos=1,above] {$\scriptstyle \lab$} (\j,4.5);
\foreach\i/\lab in {4/\delta_1,3/\delta_2,2/,1/\delta_N}
\draw[invarrow=0.95] (0.5,\i) -- (4.5,\i);
\foreach\i/\lab in {4/x_1,3/x_2,2/ ,1/x_N}
\node at (4.5,0.2+\i) {$\scriptstyle \lab$};
\foreach\i/\lab in {1/q x_1,2/q x_2,3/ ,4/q x_N}
\node at (\i,5.2) {$\scriptstyle \lab$};
\foreach\i/\lab in {4/0,3/0,2/ ,1/0}
\node at (0.2,\i) {$\scriptstyle \lab$};
\foreach\i/\j in {4/1,3/2,2/3 ,1/4}
\draw [rounded corners=5pt]  (\i,0.5) -- (\i,0.5-0.25*\j) -- (4+0.5+0.25*\j,0.5-0.25*\j) -- (4+0.5+0.25*\j,\j) -- (4.5,\j);
\end{tikzpicture}
\end{align}
If we rotate this picture by 135 degrees counterclockwise we can match it with the lattice drawn on the cone in \eqref{eq:cone}. Therefore $T_N(x;z_0,z_1)$ coincides with the partition function $Z_N$ from the introduction \eqref{eq:Z} where we need to choose $n=1$ and $m=0$.
\begin{lemma}
The function $T_N(x;z_0,z_1)$ is symmetric in $(x_1\ldots x_N)$.
\label{lem:Tsymmetric}
\end{lemma}
\begin{proof}
We can rewrite \eqref{eq:exW-six} in the form:
$$
W_N(\ldots x_{i+1},x_i \ldots)=
\check R_i(x_{i+1}/x_i) W_N(x)\check R_i(x_{i}/x_{i+1})
$$
Using this equation we compute:
\begin{align*}
        T_N(\ldots x_{i+1},x_i \ldots;z_0,z_1)&= \sum_{\alpha\in\{0,1\}^N} z_{\alpha_1}\cdots z_{\alpha_N}
        \bra{\alpha} W_N(\ldots x_{i+1},x_i \ldots)\ket{\alpha}    
        \\
        &= \sum_{\alpha\in\{0,1\}^N} z_{\alpha_1}\cdots z_{\alpha_N}
        \bra{\alpha} 
        \check R_i(x_{i+1}/x_i) W_N(x)\check R_i(x_{i}/x_{i+1})
        \ket{\alpha}   \\
        &= \sum_{\alpha\in\{0,1\}^N} z_{\alpha_1}\cdots z_{\alpha_N}
        \bra{\alpha} 
        \check R_i(x_{i}/x_{i+1}) \check R_i(x_{i+1}/x_i) W_N(x)
        \ket{\alpha} \\
        &=
        \sum_{\alpha\in\{0,1\}^N} z_{\alpha_1}\cdots z_{\alpha_N}
        \bra{\alpha}  W_N(x)\ket{\alpha}     = T_N(x;z_0,z_1)
\end{align*}
where we also used the cyclicity of the trace and the unitarity \eqref{eq:unitarity}.
\end{proof}
\begin{thm}
The function $T_N(x;z_0,z_1)$ can be written as a shuffle product:
\label{thm:T-shuffle}
\begin{align}
    \label{eq:T-sh}
            T_N(x;z_0,z_1)=
         \sum_{k=0}^N z_0^{N-k}z_1^k \frac{(1-t^{-1})^k}{(1-q t^{-1})^k}
         H_{k}(t^{-1}) * E_{N-k}(t q^{-1})
\end{align}
\end{thm}
\begin{proof}
Let us use the relation \eqref{eq:Wtilde} between $W_N(x)$ and $\widetilde W_N(x)$ in Definition \ref{def:T} of $T_N(x;z_0,z_1)$:
\begin{align*}
            T_N(x;z_0,z_1)= \sum_{\alpha\in\{0,1\}^N} z_{\alpha_1}\cdots z_{\alpha_N}
        \bra{\alpha} F^{-1}_{N}(x)\widetilde W_{N}(x) F_{N}(x)\ket{\alpha}
\end{align*}
Using the cyclicity of the trace we get:
\begin{align}
    \label{eq:T-Wt}
            T_N(x;z_0,z_1)=
         \sum_{\alpha\in\{0,1\}^N} z_{\alpha_1}\cdots z_{\alpha_N}
        \widetilde W_\alpha^\alpha(x)
\end{align}
We can rewrite the summation over $\alpha\in\{0,1\}^N$ as a sum over  $k=0\ldots N$ and a sum over permutations of $(0^{N-k} 1^k)$:
\begin{align}
    \label{eq:Tsigma-Wt}
            T_N(x;z_0,z_1)=
         \sum_{k=0}^N \frac{1}{k! (N-k)!}  
         z_0^{N-k}z_1^k
         \sum_{\sigma \in \mathcal S_N}
         \widetilde W_{\sigma(0^{N-k}1^k)}^{\sigma(0^{N-k}1^k)}(x)
\end{align}
The division by the factorials compensates the over-counting when $\sigma$ permutes $0$'s or $1$'s. In the above formula the trace from \eqref{eq:T-Wt} is written as a sum over permutations $\sigma$ acting on the indices of the matrix elements of $W_N(x)$, below we will show that this action can be transferred to the action of $\sigma$ which instead permutes the spectral parameters. 
The diagonal elements $\widetilde W_{\alpha}^{\alpha}(x)$ in \eqref{eq:Tsigma-Wt} can be computed with the help of Proposition \ref{prop:Wy} in which we set $y_i=q x_i$:
\begin{align}
    \label{eq:W-mat}
\widetilde W_{\alpha}^{\alpha}(x)= 
D_{k}(x_S) 
\prod_{i\in S}\prod_{j\in S^c}
\frac{x_i-t x_j}{x_i-x_j}
\frac{q x_i-x_j}{q x_i-t x_j}=
D_k (x_S) 
\prod_{i\in S}\prod_{j\in S^c}\zeta(x_i/x_j)
\end{align}
where $D_{k}(x):= D_{k}(x_1\ldots x_k;q x_1\ldots q x_k)$, the subset $S\subseteq [N]$ records the positions of $1$'s in $\alpha$, $S^c\subseteq [N]$ is the complement subset to $S$ and thus it records the positions of $0$'s in $\alpha$. The domain-wall  partition function $D_{k}(x)$ can be computed as a determinant using Lemma \ref{lem:IK} in the special case $y_i = q x_i$. The determinant formula for $D_{k}(x)$ matches with the definition of $H_k(t^{-1})$ in \eqref{eq:H_def} up to a factor:
\begin{align}\label{eq:H-Wx}
        D_k (x)= 
        \frac{(1-t^{-1})^k}{(1-q t^{-1})^k}
        H_k(t^{-1}) 
\end{align}
Inserting \eqref{eq:H-Wx} into \eqref{eq:W-mat} and then inserting the result into \eqref{eq:Tsigma-Wt} produces:
\begin{multline}
    \label{eq:T-EH}
            T_N(x;z_0,z_1)=
         \sum_{k=0}^Nz_0^{N-k}z_1^k 
         \frac{(1-t^{-1})^k}{(1-q t^{-1})^k}\\
          \frac{1}{k! (N-k)!} 
         \sum_{\sigma \in \mathcal S_N}
         H_k(t^{-1}; x_{\sigma(N-k+1)}\ldots x_{\sigma(N)}) 
\prod_{j=1}^{N-k}
\prod_{i=N-k+1}^N
\zeta(x_{\sigma(i)}/x_{\sigma(j)})
\end{multline}
After reordering the summation over $\sigma$ we can match the expression in the second line with the shuffle product \eqref{eq:sh} of $H_k(t^{-1})$ and $E_{N-k}(t q^{-1})$ which verifies \eqref{eq:T-sh}.
\end{proof}
\begin{cor}
Consider the generating function of $T_N(x;z_0,z_1)$:
\begin{align}
    T(v;z_0,z_1) = \sum_{N=0}^{\infty} v^N T_N(x;z_0,z_1)
\end{align}
We have:
\begin{align}\label{eq:Tv}
    T(v;z_0,z_1) =
    \exp_*\left(\sum_{r>0}\frac{1}{r}\frac{1-t^r}{1-q^r}\left( -z_1^r-\frac{q^r-t^r}{1-t^r}z_0^r \right) v^{r}
           S_r\right)
\end{align}
\end{cor}
\begin{proof}
We compute $T(v)$ using \eqref{eq:T-sh}:
\begin{align*}
            \sum_{N=0}^{\infty}v^N T_N(x;z_0,z_1)&=
         \sum_{N=0}^{\infty}v^N \sum_{k=0}^N z_0^{N-k}z_1^k \frac{(1-t^{-1})^k}{(1-q t^{-1})^k}
         H_{k}(t^{-1}) * E_{N-k}(t q^{-1})\\
         &=
         \sum_{k=0}^{\infty} (v z_1)^k \frac{(1-t^{-1})^k}{(1-q t^{-1})^k}
         H_{k}(t^{-1}) * \sum_{l=0}^{\infty}  (v z_0)^{l} E_{l}(t q^{-1})\\
         &=\exp_*
    \left(-\sum_{r>0}\frac{1}{r}\frac{1-t^r}{1-q^r}v^{r} z_1^r  S_r
    \right)
    \exp_*
    \left(-\sum_{r>0}\frac{1}{r}
    \frac{q^r-t^r }{1-q^r} v^{r}z_0^r S_r\right)\\
        &=
          \exp_*\left(\sum_{r>0}\frac{1}{r}\frac{1-t^r}{1-q^r}\left( -z_1^r-\frac{q^r-t^r}{1-t^r}z_0^r \right) v^{r}
           S_r\right)
\end{align*}
where in the second line we reordered the two summations and in the third line we used the exponential generating functions \eqref{eq:E-gen} and \eqref{eq:H-gen} for $H_k(t^{-1})$ and $E_k(t q^{-1})$ and finally we combined the two shuffle-exponentials into a single shuffle-exponential. The last step is possible because of the shuffle-commutativity of $H_k(t^{-1})$ and $E_k(t q^{-1})$. 
\end{proof}

\section{The \texorpdfstring{$sl_{n+1|m}$}{TEXT} vertex model and the shuffle algebra \texorpdfstring{$\mathcal{A}^\circ$}{TEXT}}\label{sec:supersymmetric}
The aim of this section is to prove Theorem \ref{thm:Z_intro} in the full generality. We will follow the same logic as in Section \ref{sec:six_vertex}. The main input in  Section \ref{sec:six_vertex} is the six vertex $R$-matrix. In this section, we will replace it with the $R$-matrix of $U_{t}(\widehat{sl}_{n+1|m}$, a supersymmetric algebra whose representations are given by $\mathbb Z_2$ graded vector spaces. The grading splits vectors into two categories which graphically correspond to having bosonic and fermionic lattice paths. In this sense a red path in the six vertex model case is a bosonic path and these paths contribute the $H_k(t^{-1})$ part in the formula \eqref{eq:T-sh}. In this section we will obtain a generalization of the formula \eqref{eq:T-sh} which, in addition to the bosonic contributions $H_k(t^{-1})$, will have fermionic contributions of  $E_k(t^{-1})$.


\subsection{Coloured lattice models} 
We use the same diagrammatic formalism as in Section \ref{sec:six_vertex}. Consider a square lattice as in \eqref{tikz:latticemodel} where every edge of a vertex is labelled by $0\ldots n+m$ where $0$ denotes the absence of a path, $i>0$ stands for a coloured path of colour ``$i$''. If $i\leq n$ then the path is called bosonic and if $i> n$ the path is called fermionic. Specifying the boundary conditions in \eqref{tikz:latticemodel} means assigning labels to the $4N$ external edges which take values in $\{0\ldots n+m\}$. For colours $i<j$ and ``0'' being the greatest colour, we list all the possible types of vertices with their corresponding Boltzmann weights:
\begin{equation}\label{tikz:coloredvertices}
\begin{tabular}{c@{\hskip 0.32cm}c@{\hskip 0.32cm}c@{\hskip 0.32cm}c@{\hskip 0.32cm}c@{\hskip 0.32cm}c@{\hskip 0.32cm}c}
\begin{tikzpicture}[scale=0.6,baseline=-2pt]
\draw[invarrow=0.75] (-1,0) --(1,0);
\draw[invarrow=0.75] (0,-1) --(0,1);
\node[right] at (1,0) {$\scriptstyle x$};
\node[above] at (0,1) {$\scriptstyle y$};
\end{tikzpicture}:
&
\begin{tikzpicture}[scale=0.6,baseline=-2pt]
\draw[red, ultra thick] (-1,0) node[left,black]{$\ss i$}  -- (0,0)--(0,1) node[above,black]{$\ss i$};
\draw (0,-1) node[below,black]{$\ss 0$} --(0,0)-- (1,0)node[right,black]{$\ss 0$};
\end{tikzpicture}
&
\begin{tikzpicture}[scale=0.6,baseline=-2pt]
\draw (-1,0) node[left,black]{$\ss 0$}  -- (0,0)--(0,1) node[above,black]{$\ss 0$};
\draw[red, ultra thick] (0,-1) node[below,black] {$\ss i$} --(0,0)-- (1,0) node[right,black]{$\ss i$};
\end{tikzpicture}
&
\begin{tikzpicture}[scale=0.6,baseline=-2pt]
\draw[red, ultra thick] (-1,0)  node[left,black]{$\ss i$}-- (1,0) node[right,black]{$\ss i$};
\draw (0,-1)  node[below,black]{$\ss 0$}--(0,1) node[above,black]{$\ss 0$};
\end{tikzpicture}
&
\begin{tikzpicture}[scale=0.6,baseline=-2pt]
\draw (-1,0) node[left,black]{$\ss 0$}  -- (1,0) node[right,black]{$\ss 0$};
\draw[red, ultra thick] (0,-1)node[below,black]{$\ss i$}  --(0,1) node[above,black]{$\ss i$};
\end{tikzpicture}
&
\begin{tikzpicture}[scale=0.6,baseline=-2pt]
\draw (-1,0) node[left,black]{$\ss 0$}  -- (1,0) node[right,black]{$\ss 0$};
\draw (0,-1)node[below,black]{$\ss 0$}  --(0,1)node[above,black]{$\ss 0$};
\end{tikzpicture}
\\[3em]
&
\begin{tikzpicture}[scale=0.6,baseline=-2pt]
\draw[red, ultra thick] (-1,0) node[left,black]{$\ss i$}  -- (0,0)--(0,1) node[above,black]{$\ss i$};
\draw[gcol, ultra thick] (0,-1) node[below,black]{$\ss j$} --(0,0)-- (1,0)node[right,black]{$\ss j$};
\end{tikzpicture}
&
\begin{tikzpicture}[scale=0.6,baseline=-2pt]
\draw[gcol, ultra thick] (-1,0) node[left,black]{$\ss j$}  -- (0,0)--(0,1) node[above,black]{$\ss j$};
\draw[red, ultra thick] (0,-1) node[below,black] {$\ss i$} --(0,0)-- (1,0) node[right,black]{$\ss i$};
\end{tikzpicture}
&
\begin{tikzpicture}[scale=0.6,baseline=-2pt]
\draw[red, ultra thick] (-1,0)  node[left,black]{$\ss i$}-- (1,0) node[right,black]{$\ss i$};
\draw[gcol, ultra thick] (0,-1)  node[below,black]{$\ss j$}--(0,1) node[above,black]{$\ss j$};
\end{tikzpicture}
&
\begin{tikzpicture}[scale=0.6,baseline=-2pt]
\draw[gcol, ultra thick] (-1,0) node[left,black]{$\ss j$}  -- (1,0) node[right,black]{$\ss j$};
\draw[red, ultra thick] (0,-1)node[below,black]{$\ss i$}  --(0,1) node[above,black]{$\ss i$};
\end{tikzpicture}
&
\begin{tikzpicture}[scale=0.6,baseline=-2pt]
\draw[ red,ultra thick] (-1,0) node[left,black]{$\ss i$}  -- (1,0) node[right,black]{$\ss i$};
\draw[red, ultra thick] (0,-1)node[below,black]{$\ss i$}  --(0,1)node[above,black]{$\ss i$};
\end{tikzpicture}
\\[3em]
& $\dfrac{1-t}{1-tx/y}$ & $\dfrac{(1-t)x/y}{1-t x/y}$& $\dfrac{t(1-x/y)}{1-t x/y}$&$\dfrac{1-x/y}{1-t x/y}$&$
 \begin{cases}
  \qquad 1  & i \leq n \\
  \dfrac{x/y-t}{1-tx/y}  & i> n
\end{cases}$
\end{tabular}
\end{equation}
The weights of all the other vertices are set to $0$. Note that the difference between fermionic and bosonic paths shows in the last vertex in \eqref{tikz:coloredvertices}, i.e. when two paths of the same colour touch each other.

Consider an example of \eqref{tikz:latticemodel} with $N=2$ with one bosonic (red) and one fermionic (green) path and a choice of boundary conditions:
\begin{align}\label{eq:example_config_coloured}
\begin{tikzpicture}[scale=0.6,baseline=-2pt]
\node[right] at (3+0.5,0+0.5) {$\scriptstyle x_1$};
\node[right] at (3+0.5,-1+0.5) {$\scriptstyle x_2$};
\node[above] at (1,1+1) {$\scriptstyle y_1$};
\node[above] at (2,1+1) {$\scriptstyle y_2$};
\node[right] at (3,0+0.5) {$\ss 1$};
\node[right] at (3,-1+0.5) {$\ss 2$};
\node[left] at (0,0+0.5) {$\ss 0$};
\node[left] at (0,-1+0.5) {$\ss 2$};
\node[above] at (1,1+0.5) {$\ss 0$};
\node[above] at (2,1+0.5) {$\ss 2$};
\node[below] at (1,-2+0.5) {$\ss 1$};
\node[below] at (2,-2+0.5) {$\ss 2$};
\draw[invarrow=0.875] (0,0+0.5) --(3,0+0.5);
\draw[invarrow=0.875] (0,-1+0.5) --(3,-1+0.5);
\draw[arrow=0.125] (1,1+0.5) --(1,-2+0.5);
\draw[arrow=0.125] (2,1+0.5) --(2,-2+0.5);
\end{tikzpicture}
=
\begin{tikzpicture}[scale=0.6,baseline=-2pt]
\node[right] at (3,0+0.5) {$\ss 1$};
\node[right] at (3,-1+0.5) {$\ss 2$};
\node[left] at (0,0+0.5) {$\ss 0$};
\node[left] at (0,-1+0.5) {$\ss 2$};
\node[above] at (1,1+0.5) {$\ss 0$};
\node[above] at (2,1+0.5) {$\ss 2$};
\node[below] at (1,-2+0.5) {$\ss 1$};
\node[below] at (2,-2+0.5) {$\ss 2$};
\draw (0,0+0.5) --(3,0+0.5);
\draw (0,-1+0.5) --(3,-1+0.5);
\draw (1,1+0.5) --(1,-2+0.5);
\draw (2,1+0.5) --(2,-2+0.5);
\draw[gcol, ultra thick] (0,-1+0.5) --(3,-1+0.5);
\draw[red, ultra thick] (1,-2+0.5) --(1,0.5) --(3,0.5);
\draw[gcol, ultra thick] (2,-2+0.5) -- (2,1.5);
\end{tikzpicture} 
&+
\begin{tikzpicture}[scale=0.6,baseline=-2pt]
\node[right] at (3,0+0.5) {$\ss 1$};
\node[right] at (3,-1+0.5) {$\ss 2$};
\node[left] at (0,0+0.5) {$\ss 0$};
\node[left] at (0,-1+0.5) {$\ss 2$};
\node[above] at (1,1+0.5) {$\ss 0$};
\node[above] at (2,1+0.5) {$\ss 2$};
\node[below] at (1,-2+0.5) {$\ss 1$};
\node[below] at (2,-2+0.5) {$\ss 2$};
\draw (0,0+0.5) --(3,0+0.5);
\draw (0,-1+0.5) --(3,-1+0.5);
\draw (1,1+0.5) --(1,-2+0.5);
\draw (2,1+0.5) --(2,-2+0.5);
\draw[gcol, ultra thick] (0,-1+0.5) -- (1,-1+0.5) --(1,0.5) --(2,0.5) --(2,1.5);
\draw[red, ultra thick] (1,-2+0.5) --(1, -1+ 0.5) --(2, -1+ 0.5) --(2,  0.5) --(3,  0.5) ;
\draw[gcol, ultra thick] (2,-2+0.5) -- (2,-1+0.5) -- (3,-1+0.5);
\end{tikzpicture}
  \\
= \frac{(1-t) x_1/y_1 \left(1-x_2/y_1\right) t\left(1-x_1/y_2\right) \left(x_2/y_2-t\right)}{\left(1-t x_1/y_1 \right) \left(1-t x_2/y_1\right) \left(1-t x_1/y_2 \right) \left(1-t
   x_2/y_2\right)}
&+
   \frac{(1-t)^4 x_1/y_1\, x_2/y_1\, x_1/y_2}{\left(1-t x_1/y_1 \right) \left(1-t x_2/y_1\right) \left(1-t x_1/y_2 \right) \left(1-t
   x_2/y_2\right)}
\nonumber
\end{align}
Summing over all possible labels of the internal edges on the l.h.s. of \eqref{eq:example_config_coloured} gives two possible configurations with non-zero  Boltzmann weights which are computed with \eqref{tikz:coloredvertices} and presented in the second line.
\begin{dfn}\label{def:Z_coloured}
Let $\alpha,\beta,\gamma,\delta\in\{0,1\ldots n+m\}^N$ (cf. Definition \ref{def:Z}). The supersymmetric coloured partition function $Z_{\alpha,\gamma}^{\beta,\delta}(x_1\dots x_N;y_1\dots y_N)$ is defined as the rational function in the spectral parameters equal to the weighted sum over all possible configurations computed with \eqref{tikz:coloredvertices} and the boundary conditions specified below: 
\begin{align}\label{eq:Z_def_coloured}
Z_{\alpha,\gamma}^{\beta,\delta}(x;y)
=
\begin{tikzpicture}[scale=0.8, baseline=(current  bounding  box.center)]
\foreach\j/\lab in {1/\beta_1,2/\beta_2,3/,4/\beta_N}
\draw[invarrow=0.95] (\j,0.5) -- node[pos=1,above] {$\scriptstyle \lab$} (\j,4.5);
\foreach\i/\lab in {4/\delta_1,3/\delta_2,2/,1/\delta_N}
\draw[invarrow=0.95] (0.5,\i) -- node[pos=1,right] {$\scriptstyle \lab$} (4.5,\i);
\foreach\i/\lab in {4/x_1,3/x_2,2/ ,1/x_N}
\node at (5.5,\i) {$\scriptstyle \lab$};
\foreach\i/\lab in {1/y_1,2/y_2,3/ ,4/y_N}
\node at (\i,5.3) {$\scriptstyle \lab$};
\foreach\i/\lab in {4/\alpha_1,3/\alpha_2,2/ ,1/\alpha_N}
\node at (0.2,\i) {$\scriptstyle \lab$};
\foreach\i/\lab in {1/\gamma_1,2/\gamma_2,3/ ,4/\gamma_N}
\node at (\i,0.3) {$\scriptstyle \lab$};
\end{tikzpicture}
\end{align}
\end{dfn}
The supersymmetric $\check R$-matrix acts in $V_y\otimes V_x$ with $V_y,V_x \simeq \mathbb C^{n+m+1}$. Let $\ket{0}=(1,0\ldots 0)^T$ and $\ket{i}=(0\ldots 0,1,0\ldots 0)^T$, with $1$ being on the $i$-th position, denote the standard basis in $\mathbb C^{n+m+1}$ and $\ket{i_1\dots i_N}$, with $i_1,\dots, i_N\in \{0,1\ldots n+m+1\}$, its generalization to the $N$-fold tensor product of $\mathbb C^{n+m+1}$. Define similarly the dual basis, then we have:
\begin{align}
    \label{eq:R-Rg_sup}
    \check R(x/y) = \sum_{a,b,c,d=0,1\ldots n+m+1}
\left[
\begin{tikzpicture}[scale=0.6,baseline=-2pt]
\draw[invarrow=0.75] (-1,0) --(1,0);
\draw[invarrow=0.75] (0,-1) --(0,1);
\node[right] at (1+0.5,0) {$\scriptstyle x$};
\node[above] at (0,1+0.5) {$\scriptstyle y$};
\node[below] at (0,-1) {$\scriptstyle c$};
\node[left] at (-1,0) {$\scriptstyle a$};
\node[right] at (1,0) {$\scriptstyle d$};
\node[above] at (0,1) {$\scriptstyle b$};
\end{tikzpicture}
\right]
 \ket{a,c}\bra{b,d}
\end{align}
where the explicit form of the matrix elements is given by the Boltzmann weights \eqref{tikz:coloredvertices}. The permutation matrix $P$ acts in $\mathbb C^{n+m+1}\otimes \mathbb C^{n+m+1}$. The supersymmetric $R$-matrix is defined as before:
\begin{align}\label{eq:R_sup}
R(x) = P\check R(x)    
\end{align}
Let $\id$ be the identity matrix in $\mathbb C^{n+m+1}$ then $\check R$, $R$ and $P$ act in $\otimes_{i=1}^N V_{x_i}$ by $\check R_i(x_{i+1}/x_i)$, $R_i(x_{i+1}/x_i)$ and $P_i$ respectively (see e.g. \eqref{eq:R_i}). For these matrices we have the Yang--Baxter equation \eqref{eq:YB} and the same unitarity relation as before \eqref{eq:unitarity}. We will also use the same graphical representation for $\check R_i(x_{i+1}/x_i)$ as before. 

The tensors $Z_N(x;y),W_N(x;y)$ and $W_N(x)$ in the $sl_{n+1|m}$ case are defined algebraically by \eqref{eq:defZ}, \eqref{eq:defWxy} and \eqref{eq:defW} respectively and their matrix elements are
$$
Z_{\alpha,\gamma}^{\beta,\delta}(x;y),
\quad 
W_{\gamma}^{\delta}(x;y)
\quad \text{and}~~~
W_{\gamma}^{\delta}(x)
$$
with the indices $\alpha,\beta,\gamma,\delta\in\{0,1\ldots n+m\}^N$. The graphical interpretation of these matrix elements is given by Definition \ref{def:Z_coloured} and appropriate specialization of the indices and parameters in the case of $W_{\gamma}^{\delta}(x;y)$ and $W_{\gamma}^{\delta}(x)$. We have the analogue of Lemma \ref{lem:ZW-exchange}. 
\begin{lemma}\label{lem:ZW-exchange_coloured}
For $i=1 \ldots N-1$,
the tensor $Z_N(x;y)$ satisfies the following exchange relations:
\begin{align}
    \label{eq:exx}
&\check R_i(x_{i+1}/x_i) Z_N(x;y)=
    Z_N(\ldots x_{i+1},x_i \ldots;y) \check R_{N+i}(x_{i+1}/x_i)
\\
    \label{eq:exy}
&Z_N(x;y)\check R_i(y_i/y_{i+1})
=
\check R_{N+i}(y_i/y_{i+1}) Z_N(x;\ldots y_{i+1},y_{i} \ldots)
\end{align}
The tensor $W_N(x;y)$ satisfies:
\begin{align}
    \label{eq:exWxy1}
W_N(x;y)
&=    W_N(\ldots x_{i+1},x_i \ldots;y) \check R_{i}(x_{i+1}/x_i)\\
    \label{eq:exWxy2}
 W_N(x;y)
&=  \check R_{i}(y_i/y_{i+1})  W_N(x;\ldots y_{i+1},y_i\ldots)     
\end{align}
and for $W_N(x)$ we have:
\begin{align}
\label{eq:exW}
\check R_{i}(x_{i+1}/x_i) W_N(x)
=    W_N(\ldots x_{i+1},x_i\ldots) \check R_{i}(x_{i+1}/x_i)
\end{align}
\end{lemma}
This Lemma is based on the same algebraic identities as Lemma \ref{lem:ZW-exchange}.

\subsection{The \texorpdfstring{$sl_{n+1|m}$}{TEXT} \texorpdfstring{$F$}{TEXT}-matrix and  transformed tensors}
In this section we introduce the supersymmetric  $sl_{n+1|m}$ $F$-matrix \cite{YZZ}. This matrix satisfies the same properties as the six vertex $F$-matrix therefore we will immediately get the transformed tensors $\widetilde Z_N(x;y),\widetilde W_N(x;y)$ and $\widetilde W_N(x)$ with the desired properties establishes in Section \ref{sec:F}.
\begin{dfn}
For $k,l\in \{0\ldots n+m\}$ and $r=1\dots N$, let $E^{(k l)}_{r}\in \normalfont{ \text{End}}\left(\otimes_{i=1}^N V_{x_i}\right)$ be the matrix acting non-trivially on the $r$-th tensor space $V_{x_r}$ by the $(n+m+1)\times(n+m+1)$ matrix unit  $E^{(k l)}$. The $N$-site supersymmetric $sl_{n+1|m}$ $F$-matrix reads:
\begin{align}\label{eq:F-matrix-super}
    F_{N}(x_1\dots x_N) = 
    \sum_{\sigma \in \mathcal{S}_N} \sum_{(k_1\dots k_N)\in \mathcal{J}_\sigma} 
    \prod_{\substack{1\leq i< j\leq N\\ k_i=k_j>n}} \frac{x_{\sigma(i)}+x_{\sigma(j)}}{x_{\sigma(i)}-t x_{\sigma(j)}}\cdot
    \prod_{i=1}^N E^{(k_i k_i)}_{\sigma(i)} R_\sigma
\end{align}
where $R_\sigma$ is defined as before \eqref{eq:R-sigma} but based on the $sl_{n+1|m}$ $R$-matrix and the set $\mathcal{J}_\sigma$ is defined by:
\begin{align}\label{eq:J-set-super}
    \mathcal{J}_\sigma = 
    \Big{\{}
    0\leq k_1\leq \dots \leq k_N\leq n+m:
    k_i<k_{i+1} ~~~\text{if}~~~ \sigma(i) >\sigma(i+1)
    \Big{\}}
\end{align}
The $F$-matrix \eqref{eq:F-matrix-super} satisfies the property:
\begin{align}\label{eq:F-prop-super}
F_N(x_1\dots x_N) = P_\sigma F_N(x_{\sigma(1)}\ldots x_{\sigma(N)})\check R_\sigma
\end{align}
\end{dfn}
We note that the $F$-matrix in the supersymmetric case is modified compared to \eqref{eq:F-matrix} by the additional rational function in \eqref{eq:F-matrix-super}. Each factor of this rational function is associated to a pair of fermionic labels, therefore if $m=0$, this factor disappears and one recovers the standard $sl_{n+1}$ $F$-matrix (see Appendix \ref{app:F} for further discussion of this $F$-matrix).
\begin{lemma}\label{lem:F-prop-super}
Fix a composition of non-negative integers $(l_0,l_1\ldots l_{n+m})$ such that $N=l_0+\cdots +l_{n+m}$. For every such composition we define two ordered compositions $\lambda^-$ and $\lambda^+$:
\begin{align}\label{eq:lambda_pm}
\lambda^- :=     
(0^{l_0}1^{l_1}\ldots (n+m)^{l_{n+m}}), 
\qquad
\lambda^+ := 
((n+m)^{l_{n+m}}\ldots 1^{l_1} 0^{l_0})
\end{align}
The supersymmetric $F$-matrix has the properties:
\begin{align}
    \label{eq:F-diag-super}
     \bra{\lambda^-}F_N(x) = 
     \bra{\lambda^-} 
     \prod_{\substack{1\leq i< j\leq N\\ 
     \lambda^-_i= \lambda^-_j>n}}
     \frac{x_i+x_j}{x_i-t x_j},
     \qquad
F_N(x) 
\ket{\lambda^+}
=
 \prod_{\substack{1\leq i<j \leq N\\ \lambda^+_i \neq  \lambda^+_j}} 
       t^{\delta_{\lambda_j>0}}\dfrac{x_i-x_j}{x_i -t x_j}\cdot \prod_{\substack{1\leq i< j\leq N\\ 
     \lambda^+_i= \lambda^+_j>n}}
     \frac{x_i+x_j}{x_i-t x_j}
\ket{\lambda^+}     
\end{align}
\end{lemma}
The proof of this lemma is presented in Appendix \ref{app:proof_F-lemma}
\begin{dfn}
Let $F_N$ be the supersymmetric $F$-matrix \eqref{eq:F-matrix-super}. Introduce transformed tensors $\widetilde Z_N(x;y)$, $\widetilde W_N(x;y)$ and $\widetilde W_N(x)$:
\begin{align}
    \label{eq:Ztilde-super}
\widetilde{Z}_{N}(x;y) &:= 
\left(F_{N}(x)\otimes
F_{N}(y)\right)
Z_{N}(x;y) 
\left(F^{-1}_{N}(y)\otimes 
F^{-1}_{N}(x)\right)\\
    \label{eq:Wytilde-super}
\widetilde{W}_{N}(x;y) &:= 
F_{N}(y)
W_{N}(x;y) 
F^{-1}_{N}(x)\\
    \label{eq:Wtilde-super}
\widetilde{W}_{N}(x) &:= 
F_{N}(x)
W_{N}(x) 
F^{-1}_{N}(x)
\end{align}    
\end{dfn}
\begin{prop}
For $i=1 \ldots N-1$,
the tensor $\widetilde Z_N(x;y)$ satisfies the following exchange relations:
\begin{align}\label{eq:Ztilde-exx-super}
P_i \widetilde  Z(x;y)&=
    \widetilde  Z(\ldots x_{i+1},x_i \ldots;y) P_{N+i}
\\
    \label{eq:Ztilde-exy-super}
\widetilde  Z(x;y) P_i
&=
P_{N+i}\widetilde  Z(x;\ldots y_{i+1},y_{i} \ldots) 
\end{align}
the tensor $\widetilde W_N(x;y)$ satisfies:
\begin{align}
    \label{eq:Wytilde-ex1-super}
\widetilde  W_N(x;y)P_{i}
&=   \widetilde  W_N(\ldots x_{i+1},x_i \ldots;y)\\
    \label{eq:Wytilde-ex2-super}
P_i\widetilde  W_N(x;y)
&=   \widetilde  W_N(x;\ldots y_{i+1},y_i\ldots)
\end{align}
and $\widetilde  W_N(x)$ satisfies:
\begin{align}
\label{eq:Wtilde-ex-super}
P_i \widetilde W_N(x)
&=\widetilde  W_N(\ldots x_{i+1},x_i \ldots)P_i
\end{align}
\end{prop}
The proof of this proposition is the same as the proof of Proposition \ref{prop:tilde-ex} since the algebraic properties of the more general objects are the same as in the six vertex case. In the remaining part of this subsection we compute the matrix elements of $\widetilde W_N(x;y)$ generalizing the result of Proposition \ref{prop:Wy}. The formula is again a product consisting of fully factorized terms and domain-wall partition functions:
\begin{align}
\label{eq:DW-super}
D_{M}^{(k)}(x;y):=W_{(k^M)}^{(k^M)}(x_1\ldots x_M;y_1\ldots y_M).    
\end{align}
When the label $k$ is fermionic (i.e. $k>n$) the domain-wall partition function becomes the fermionic version of the six vertex domain-wall partition function and is known to have a factorized formula \cite{FW}. We summarize the formulas for $D_{M}^{(k)}(x;y)$ for different $k$ in the following Lemma.
\begin{lemma}\label{lem:IK-super}
$D_{M}^{(k)}(x;y)$ is the domain-wall partition function associated to the label $k$:
\begin{align}
    \label{eq:IK_det-super}
    D_M^{(k)} (x;y)= 
    \begin{cases}
        1 & k=0 \\
        D_M(x;y) & 0< k \leq n \\
        (-1+t)^M\dfrac{\prod_{i=1}^M x_i \prod_{j=i+1}^M(x_j-t x_i)(y_i-t y_j)}{\prod_{i,j=1}^M(t x_i-y_j)} & n< k \leq n+m
    \end{cases}
\end{align}    
where $D_M(x;y)$ is defined in \eqref{eq:DW}.
\end{lemma}
\begin{prop}
    \label{prop:Wy-super}
Fix a composition of non-negative integers $(l_0,l_1\ldots l_{n+m})$ such that $N=l_0+\cdots +l_{n+m}$. Let $\alpha,\beta$ be two permutations of $(0^{l_0}1^{l_1}\ldots (n+m)^{l_{n+m}})$ and  $P(k),S(k)\subseteq [N]$ be pairs of compositions which record the positions of $k$'s in $\alpha$ and $\beta$ respectively, then:
\begin{align}
\label{eq:Wy-mat-super}
\widetilde W_{\alpha}^{\beta}(x;y)= 
    \prod_{k=0}^{n+m}
    D_{l_k}^{(k)}(x_{S(k)},y_{P(k)})
    &\times 
    \prod_{0\leq k_1<k_2 \leq n+m}
    \prod_{j\in S(k_1)}
    \prod_{i\in S(k_2)}
    \frac{x_i-t x_j}{x_i-x_j}
     \cdot
    \prod_{j\in S(k_1)}
    \prod_{i\in P(k_2)}     
    \frac{y_{i}-x_{j}}{y_{i}-t x_{j}}
\\
    &\times
    \prod_{k=n+1}^{n+m}
    \prod_{\substack{i,j\in P(k)\\ i<j}}
     \frac{y_i+y_j}{y_i-t y_j}\cdot
    \prod_{\substack{i,j\in S(k)\\ i<j}} 
    \frac{x_i-t x_j}{x_i+x_j}
\nonumber    
\end{align}
\end{prop}
\begin{proof}
The proof is analogous to the proof of Proposition \ref{prop:Wy}. Define two compositions $\lambda^\pm$:
\begin{align*}
\lambda^- :=     
(0^{l_0}1^{l_1}\ldots (n+m)^{l_{n+m}}), 
\qquad
\lambda^+ := 
((n+m)^{l_{n+m}}\ldots 1^{l_1} 0^{l_0})
\end{align*}
For two permutations $\sigma,\tau \in \mathcal{S}_N$ let $\alpha=\sigma^{-1}(\lambda^-)$ and $\beta=\tau^{-1}(\lambda^-)$. Represent the vectors $\bra{\alpha}$ and $\ket{\beta}$ as follows:
$$
\bra{\alpha} = 
\bra{\lambda^-}P_{\sigma^{-1}},
\qquad
\ket{\beta}= 
P_{\tau}\ket{\lambda^-}
$$
Different matrix elements $\widetilde W_{\alpha}^{\beta}(x;y)$ are related to each other by appropriately permuting the spectral parameters:
\begin{align}\label{eq:Wt-Wt-lambda}
\widetilde W_{\alpha}^{\beta}(x;y)
= 
\bra{\alpha}\widetilde W_N(x;y) \ket{\beta}
=
\bra{\lambda^-}P_{\sigma^{-1}}
\widetilde W_N(x;y)
P_{\tau}\ket{\lambda^-}
=
\bra{\lambda^-}
\widetilde W_N(x_\tau;y_\sigma) \ket{\lambda^-}
=\widetilde W_{\lambda^-}^{\lambda^-}(x_\tau;y_\sigma)
\end{align}
where we used the properties \eqref{eq:Wytilde-ex1-super} and \eqref{eq:Wytilde-ex2-super} of $\widetilde W_{\alpha}^{\beta}(x;y)$. Consider next the matrix element $\widetilde W_{\lambda^-}^{\lambda^+}(x;y)$ and write it in terms of $W_{\lambda^-}^{\lambda^+}(x;y)$ using \eqref{eq:Wytilde-super}:
\begin{align}\nonumber
    \widetilde W_{\lambda^-}^{\lambda^+}(x;y)
    &=\bra{\lambda^-}F_{N}(y)W_{N}(x;y) F^{-1}_{N}(x) \ket{\lambda^+}
    \\
    &=
    \prod_{\substack{1\leq i<j \leq N\\ \lambda^+_i \neq  \lambda^+_j}} 
       t^{-\delta_{\lambda^+_j>0}}\dfrac{x_i-t x_j}{x_i - x_j}
       \prod_{\substack{1\leq i< j\leq N\\ 
     \lambda^+_i= \lambda^+_j>n}}
     \frac{x_i-t x_j}{x_i+x_j}
     \prod_{\substack{1\leq i< j\leq N\\ 
     \lambda^-_i= \lambda^-_j>n}}
     \frac{y_i+y_j}{y_i-t y_j}
     \times
    W_{\lambda^-}^{\lambda^+}(x;y)
    \label{eq:Wt-lamba}
\end{align}
where in the second line we used the properties \eqref{eq:F-diag-super} of the supersymmetric $F$-matrix. The matrix element $W_{\lambda^-}^{\lambda^+}(x;y)$ can be evaluated in a similar way as the matrix element $W_{(0^{N-k} 1^k)}^{(1^k0^{N-k})}(x;y)$ in the proof of Proposition \ref{prop:Wy}. For demonstration we choose $N=6$, $m=n=1$ and $l_0=l_1=l_2=2$. The graphical depiction of the partition function $W_{(001122)}^{(221100)}(x;y)$ shows several regions which are frozen and two regions which contain domain-wall type configurations which we enclosed in the dotted squares for the reference:
\begin{align}
\nonumber
&W_{(001122)}^{(221100)}(x;y)
=
\begin{tikzpicture}[scale=0.8, baseline=(current  bounding  box.center)]
\foreach\j/\lab in {1/0,2/0,3/0,4/0,5/0,6/0}
\draw[invarrow=0.975] (\j,0.5) -- node[pos=1,above] {$\scriptstyle \lab$} (\j,6.5);
\foreach\i/\lab in {6/2,5/2,4/1,3/1,2/0,1/0}
\draw[invarrow=0.975] (0.5,\i) -- node[pos=1,right] {$\scriptstyle \lab$} (6.5,\i);
\foreach\i/\lab in {6/x_1,5/x_2,4/x_3 ,3/x_4,2/x_5,1/x_6}
\node at (7.5,\i) {$\scriptstyle \lab$};
\foreach\i/\lab in {1/y_1,2/y_2,3/y_3,4/y_4,5/y_5,6/y_6}
\node at (\i,7.3) {$\scriptstyle \lab$};
\foreach\i/\lab in {6/0,5/0,4/0,3/0,2/0,1/0}
\node at (0.2,\i) {$\scriptstyle \lab$};
\foreach\i/\lab in {1/0,2/0,3/1,4/1,5/2,6/2}
\node at (\i,0.3) {$\scriptstyle \lab$};
\foreach\i in {1,2,3,4,5}
\node at (1,\i+0.5) {$\scriptstyle 0$};
\foreach\i in {1,2,3,4,5}
\node at (2,\i+0.5) {$\scriptstyle 0$};
\foreach\i/\lab in {1/1,2/1,3/,4/0,5/0}
\node at (3,\i+0.5) {$\scriptstyle \lab$};
\foreach\i/\lab in {1/1,2/1,3/,4/0,5/0}
\node at (4,\i+0.5) {$\scriptstyle \lab$};
\foreach\i/\lab in {1/2,2/2,3/2,4/2,5/}
\node at (5,\i+0.5) {$\scriptstyle \lab$};
\foreach\i/\lab in {1/2,2/2,3/2,4/2,5/}
\node at (6,\i+0.5) {$\scriptstyle \lab$};
\foreach\i in {1,2,3,4,5}
\node at (\i+0.5,1) {$\scriptstyle 0$};
\foreach\i in {1,2,3,4,5}
\node at (\i+0.5,2) {$\scriptstyle 0$};
\foreach\i/\lab in {1/0,2/0,3/,4/1,5/1}
\node at (\i+0.5,3) {$\scriptstyle \lab$};
\foreach\i/\lab in {1/0,2/0,3/,4/1,5/1}
\node at (\i+0.5,4) {$\scriptstyle \lab$};
\foreach\i in {1,2,3,4}
\node at (\i+0.5,5) {$\scriptstyle 0$};
\foreach\i in {1,2,3,4}
\node at (\i+0.5,6) {$\scriptstyle 0$};
\draw[dotted] (2+0.5-0.1,2+0.5-0.1) -- (2+0.5-0.1,4+0.5+0.1) -- (4+0.5+0.1,4+0.5+0.1) -- (4+0.5+0.1,2+0.5-0.1) -- (2+0.5-0.1,2+0.5-0.1);
\draw[dotted] (4+0.5-0.1,4+0.5-0.1) -- (4+0.5-0.1,6+1) -- (6+1,6+1) -- (6+1,4+0.5-0.1) -- (4+0.5-0.1,4+0.5-0.1);
\end{tikzpicture}\\
&=
\prod_{i\in\{3,4,5,6\}}\prod_{j\in\{5,6\}}
\frac{y_i-x_j}{y_i-t x_j}
\prod_{i\in\{5,6\}}\prod_{j\in\{3,4\}}
t \frac{y_i-x_j}{y_i-t x_j}
\times 
\begin{tikzpicture}[scale=0.8, baseline=(current  bounding  box.center)]
\foreach\j/\lab in {1/0,2/0}
\draw[invarrow=0.9] (\j,0.5) -- node[pos=1,above] {$\scriptstyle \lab$} (\j,2.5);
\foreach\i/\lab in {2/1,1/1}
\draw[invarrow=0.9] (0.5,\i) -- node[pos=1,right] {$\scriptstyle \lab$} (2.5,\i);
\foreach\i/\lab in {2/x_3 ,1/x_4}
\node at (3.2,\i) {$\scriptstyle \lab$};
\foreach\i/\lab in {1/y_3,2/y_4}
\node at (\i,3.2) {$\scriptstyle \lab$};
\foreach\i/\lab in {2/0,1/0}
\node at (0.2,\i) {$\scriptstyle \lab$};
\foreach\i/\lab in {1/1,2/1}
\node at (\i,0.3) {$\scriptstyle \lab$};
\end{tikzpicture}
\times 
\begin{tikzpicture}[scale=0.8, baseline=(current  bounding  box.center)]
\foreach\j/\lab in {1/0,2/0}
\draw[invarrow=0.9] (\j,0.5) -- node[pos=1,above] {$\scriptstyle \lab$} (\j,2.5);
\foreach\i/\lab in {2/2,1/2}
\draw[invarrow=0.9] (0.5,\i) -- node[pos=1,right] {$\scriptstyle \lab$} (2.5,\i);
\foreach\i/\lab in {2/x_1 ,1/x_2}
\node at (3.2,\i) {$\scriptstyle \lab$};
\foreach\i/\lab in {1/y_5,2/y_6}
\node at (\i,3.2) {$\scriptstyle \lab$};
\foreach\i/\lab in {2/0,1/0}
\node at (0.2,\i) {$\scriptstyle \lab$};
\foreach\i/\lab in {1/2,2/2}
\node at (\i,0.3) {$\scriptstyle \lab$};
\end{tikzpicture}
\nonumber\\
&=
\prod_{i\in\{3,4,5,6\}}\prod_{j\in\{5,6\}}
\frac{y_i-x_j}{y_i-t x_j}
\prod_{i\in\{5,6\}}\prod_{j\in\{3,4\}}
t
\frac{y_i-x_j}{y_i-t x_j}\times
D_2^{(1)}(x_3,x_4;y_3,y_4)
D_2^{(2)}(x_1,x_2;y_5,y_6)
\label{eq:W001122}
\end{align}
The partial freezing of the last two columns in the graphical representation given in the first line in \eqref{eq:W001122} is due to the fact that the $2$'s on the right boundary must take either left or down steps and drop to the bottom boundary in the last two columns. This leaves us with domain-wall configurations in the intersections of the first two rows and last two columns. The remaining part of the lattice is of the form \eqref{eq:W0011}. As a result we need to evaluate two domain-wall partition functions $D_{2}^{(1)}$ and $D_{2}^{(2)}$ as indicated in the second line in \eqref{eq:W001122} and take into account various factors coming from the vertices:
\begin{align*}
\begin{tikzpicture}[scale=0.6,baseline=-2pt]
\node[left] at (0,0+0.5) {$\ss 0$};
\node[right] at (2,0+0.5) {$\ss 0$};
\node[above] at (1,1+0.5) {$\ss 1$};
\node[below] at (1,-1+0.5) {$\ss 1$};
\draw (0,0+0.5) --(2,0+0.5);
\draw (1,1+0.5) --(1,-1+0.5);
\end{tikzpicture}    
\qquad
\begin{tikzpicture}[scale=0.6,baseline=-2pt]
\node[left] at (0,0+0.5) {$\ss 0$};
\node[right] at (2,0+0.5) {$\ss 0$};
\node[above] at (1,1+0.5) {$\ss 2$};
\node[below] at (1,-1+0.5) {$\ss 2$};
\draw (0,0+0.5) --(2,0+0.5);
\draw (1,1+0.5) --(1,-1+0.5);
\end{tikzpicture}    
\qquad
\begin{tikzpicture}[scale=0.6,baseline=-2pt]
\node[left] at (0,0+0.5) {$\ss 1$};
\node[right] at (2,0+0.5) {$\ss 1$};
\node[above] at (1,1+0.5) {$\ss 2$};
\node[below] at (1,-1+0.5) {$\ss 2$};
\draw (0,0+0.5) --(2,0+0.5);
\draw (1,1+0.5) --(1,-1+0.5);
\end{tikzpicture}    
\end{align*}
where the first two are of the fourth type in \eqref{tikz:coloredvertices} and the third one is of the third type in \eqref{tikz:coloredvertices} (recall that $0$ is considered the highest colour label). 

Let us view the lattice configurations in the first line in \eqref{eq:W001122} as a block matrix. 
The example of \eqref{eq:W001122} shows that the anti-diagonal blocks must have domain-wall configurations, above the anti-diagonal the configurations are trivial with all edges labelled $0$ and below the anti-diagonal we have blocks with vertices of type four and type three in \eqref{tikz:coloredvertices}. The latter blocks can be expressed in terms of $Z_{\alpha,\beta}^{\beta,\alpha}$ with $\alpha=(i^M)$ and $\beta=(j^K)$ for some $i< j$ and evaluated:
\begin{align}
    \label{eq:Zij}
    Z_{(i^M),(j^K)}^{(j^K),(i^M)}(x;y)
    =
    t^{\delta_{i>0}M K}
    \prod_{a=1}^M \prod_{b=1}^K
\frac{y_b-x_a}{y_b-t x_a}
\end{align}
We can write $W^{\lambda^+}_{\lambda^-}(x;y)$ in terms of such block matrix:
\begin{align}\label{eq:W-prod}
    W^{\lambda^+}_{\lambda^-}(x;y)
    =
    \prod_{a,b} \left(
\begin{array}{ccccc}
 1 & \ldots & 1 & \ldots & W_{m+n} \\
 \vdots & \vdots & \vdots & \iddots & \vdots \\
 1 & \ldots & W_k & \ldots & W _{k,m+n} \\
 \vdots & \iddots & \vdots & \vdots & \vdots \\
 W_0 & \ldots & W _{0,k} & \ldots & W _{0,m+n} \\
\end{array}
\right)_{a,b}
\end{align}
where we shortened the notations:
\begin{align}
    \label{eq:short_W-chi}
    W_k = D_{l_k}^{(k)}(x^{(k)},y^{(k)}),
\qquad W_{i,j} = Z_{(i^{l_i}),(j^{l_j})}^{(j^{l_j}),(i^{l_i})}(x^{(i)},y^{(j)})
\end{align}
and $(x^{(k)})$ and $(y^{(k)})$ denote $l_k$ spectral parameters $x_a$ and $y_b$ with $a$ and $b$ such that $\lambda^+_a=k$ and $\lambda^-_b=k$ for all $a,b$. By evaluating the product in \eqref{eq:W-prod} we get:
\begin{align}\label{eq:Wy-super}
    W_{\lambda^-}^{\lambda^+}(x;y)
    =
    \prod_{k=0}^{n+m}
    D_{l_k}^{(k)}(x^{(k)},y^{(k)})
    \prod_{0\leq i<j\leq n+m}
    Z_{(i^{l_i}),(j^{l_j})}^{(j^{l_j}),(i^{l_i})}(x^{(i)},y^{(j)})
\end{align}
This expression allows us to compute $\widetilde W_{\lambda^-}^{\lambda^+}(x;y)$ using \eqref{eq:Wt-lamba}. The powers of $t$ appearing in \eqref{eq:Wt-lamba} combine into $t^{-l_0(N-l_0)}$ and, after the substitution of \eqref{eq:Wy-super} into \eqref{eq:Wt-lamba}, are canceled against the powers of $t$ coming from \eqref{eq:Zij}. After inserting $\widetilde W_{\lambda^-}^{\lambda^+}(x;y)$ into \eqref{eq:Wt-Wt-lambda} we get the final expression \eqref{eq:Wy-mat-super} for all matrix elements $\widetilde W_{\alpha}^{\beta}(x;y)$.
\end{proof}

\subsection{Trace of \texorpdfstring{$W_N(x)$}{} and the shuffle product}\label{sec:trace-super}
In this section we define the supersymmetric partition function $T_N(x)$ as the trace of the matrix $W_N(x)$. We show that the trace can be rewritten as a symmetrization over the parameters $x_1\ldots x_N$ acting on specific diagonal elements of $\widetilde W_N(x)$. These diagonal elements can be computed using Proposition \ref{prop:Wy-super}. The result of this computation, summarized in Theorem \ref{thm:T-shuffle-super}, gives a formula for $T_N(x)$ in terms of the shuffle product of $\Ac$. 
\begin{dfn}\label{def:T-super}
Let $z_0,z_1\ldots z_n$ and $w_1\ldots w_m$ be two sets of indeterminates. Define the following function:
    \begin{align}\label{eq:T_W-super}
        T_N(x;z,w) := \sum_{\alpha\in\{0,1\ldots n+m\}^N} \prod_{i=0}^n z_{\alpha_i}
        \prod_{i=1}^{m} (-w_{\alpha_{n+i}})
        \bra{\alpha} W_N(x)\ket{\alpha}
    \end{align}
\end{dfn}
The graphical counterpart of \eqref{eq:T_W-super} is given by \eqref{eq:T_def} with  vertices interpreted as in \eqref{tikz:coloredvertices}. 
Using this we can see that the function $T_N(x;z,w)$ is equal to the conic partition function $Z_N(x;z,w)$ which was defined in the introduction \eqref{eq:Z} using lattice paths. The matrix $W_N(x)$ is block diagonal with blocks labelled by a collection of non-negative integers $(l_0,l_1\ldots l_{n+m})$ where $l_k$ is the number of $k$'s in $\alpha$ and $\beta$ in the matrix elements $W_{\alpha}^\beta(x)$. This means that $T_N(x;z,w)$ is the (graded) trace of $W_N(x)$ where the monomials in $z$ and $w$ parameterize different blocks of $W_N(x)$.
\begin{lemma}
The function $T_N(x;z,w)$ is symmetric in $(x_1\ldots x_N)$.
\label{lem:Tsymmetric-super}
\end{lemma}
The proof follows the same lines as the proof of Lemma \ref{lem:Tsymmetric}. We arrive at the main theorem of this section.
\begin{thm}
The function $T_N(x;z,w)$ can be written as a shuffle product:
\label{thm:T-shuffle-super}
\begin{align}
    \nonumber
            T_N(x;z,w)=
         \sum_{\substack{l_0\ldots l_{n+m}\geq 0 \\ l_0+\cdots +l_{n+m}=N}} 
         &\prod_{i=0}^n z_i^{l_i}
         \prod_{i=1}^{m} (-w_i)^{l_{n+i}}
         \frac{(1-t^{-1})^{N-l_0}}{(1-q t^{-1})^{N-l_0}}\\
         &E_{l_{n+m}}(t^{-1}) * \dots * E_{l_{n+1}}(t^{-1})
         * H_{l_n}(t^{-1}) * \cdots * H_{l_1}(t^{-1}) 
         * E_{l_0}(t q^{-1})
\label{eq:T-sh-super}         
\end{align}
\end{thm}
\begin{proof}
By the same arguments leading to \eqref{eq:Tsigma-Wt} we convert the trace in \eqref{eq:T_W-super} into a summation over the symmetric group:
\begin{align}
    \label{eq:Tsigma-Wt-super}
            T_N(x;z,w)=
         \sum_{\substack{l_0\ldots l_{n+m}\geq 0 \\ l_0+\cdots +l_{n+m}=N}} 
         \frac{1}{l_0!\cdots l_{n+m}!}  
         \prod_{i=0}^n z_i^{l_i}
         \prod_{i=1}^{m} (-w_i)^{l_{n+i}}
         \sum_{\sigma \in \mathcal S_N}
         \widetilde W_{\sigma(\lambda^-)}^{\sigma(\lambda^-)}(x)
\end{align}
where $\lambda^-=(0^{l_0}1^{l_1}\ldots (n+m)^{l_{n+m}})$. 
Let us show that the summation over $\sigma$ can be written as a shuffle product. The diagonal elements $\widetilde W_{\alpha}^{\alpha}(x)$ in \eqref{eq:Tsigma-Wt-super} can be computed with the help of Proposition \ref{prop:Wy-super} in which we set $y_i=q x_i$\footnote{We note that in the expression \eqref{eq:W-mat-super} the eigenvalues of the $F$-matrix from \eqref{eq:F-diag-super} appear in the form of a ratio. This ratio is the only explicit information about the $F$-matrix which is needed to compute $T_N$.}:
\begin{align}
    \label{eq:W-mat-super}
\widetilde W_{\alpha}^{\alpha}(x)= 
    \prod_{k=0}^{n+m}
    D_{l_k}^{(k)}(x_{S(k)})
    \times 
    \prod_{k_1=0}^{n+m}\prod_{k_2=k_1+1}^{n+m}
    \prod_{j\in S(k_1)}
    \prod_{i\in S(k_2)}    
    \frac{x_i-t x_j}{x_i-x_j}
    \frac{q x_{i}-x_{j}}{q x_{i}-t x_{j}}    
\end{align}
where $D_M^{(k)}(x):= D_M^{(k)}(x_1\ldots x_M;q x_1\ldots q x_M)$ and $S(k)\in [N]$ is the subset that records the positions of $k$'s in $\alpha$. The factors $D_{l_k}^{(k)}(x)$ are the domain-wall partition functions from Lemma \ref{lem:IK-super} in the special case $y_i = q x_i$. We can match these partition functions with the shuffle algebra functions:
\begin{align}\label{eq:EEH-Wx}
        D_{M}^{(k)}(x)= 
        \begin{cases}
        E_M(t q^{-1}) & k=0 \\
        \dfrac{(1-t^{-1})^M}{(1-q t^{-1})^M}
        H_M(t^{-1})  & 0< k \leq n \\
        \dfrac{(1-t^{-1})^M}{(1-q t^{-1})^M}E_M(t^{-1})
        & n< k \leq n+m
    \end{cases}
\end{align}
where we remind that $E_M(t q^{-1})=1$ but considered to be a function in $(x_1\ldots x_M)$. By inserting \eqref{eq:EEH-Wx} into \eqref{eq:W-mat-super} and recalling the definition of $\zeta$ in \eqref{eq:zeta} we get:
\begin{align}
    \label{eq:WX-super}
\widetilde W_{\alpha}^{\alpha}(x)= 
    \frac{(1-t^{-1})^{N-l_0}}{(1-q t^{-1})^{N-l_0}}
    E_{l_0}(t q^{-1};x_{S(0)})
    &\prod_{k=1}^{n} H_{l_k}(t^{-1};x_{S(k)})
    \prod_{k=n+1}^{n+m} E_{l_k}(t^{-1};x_{S(k)})\\
    &\prod_{k_1=0}^{n+m}\prod_{k_2=k_1+1}^{n+m}
    \prod_{i\in S(k_1)}
    \prod_{j\in S(k_2)}    
    \zeta(x_j/x_i)\nonumber
\end{align}
Inserting this formula for $\widetilde W_{\alpha}^{\alpha}(x)$ into \eqref{eq:Tsigma-Wt-super} gives an expression which can be rewritten using the shuffle product \eqref{eq:sh}. Note that the positions of the variables $x_j$ and $x_i$ in  $\zeta(x_j/x_i)$ decide the order in which the shuffle product is taken: the functions ($H$ or $E$) which depend on $S(k)$ with higher values of $k$ should be placed on the left. This explains the reason for the ordering of the factors in the shuffle product in \eqref{eq:T-sh-super}.
\end{proof}
\begin{cor}
Consider the generating function of $T_N(x;z,w)$:
\begin{align}\label{eq:Tv_def-super}
    T(v;z,w) = \sum_{N=0}^{\infty} v^N T_N(x;z,w)
\end{align}
We have:
\begin{align}\label{eq:Tv-super}
    T(v;z,w) =
    \exp_*\left(\sum_{k>0}\frac{1}{k}\frac{1-t^k}{1-q^k}\left(\sum_{i=1}^{m} w_i^k
    -\sum_{i=1}^n z_i^k
    -\frac{q^k-t^k}{1-t^k}z_0^k  \right) v^{k}S_k\right)
\end{align}
\end{cor}
\begin{proof}
We first note that the elements $E_k(t^{-1})$, $H_k(t^{-1})$ and $E_k(t q^{-1})$ belong to the commutative shuffle algebra $\Ac$. The proof of \eqref{eq:Tv-super} follows by writing \eqref{eq:Tv_def-super} with $T_N$ given by \eqref{eq:T-sh-super}. Then summing over $N$ produces a shuffle product of $m$ generating functions of $E_k(t^{-1})$, $n$ generating functions of $H_k(t^{-1})$ and a single generating function of $E_k( tq^{-1})$ where all generating functions have different generating parameters. Then we can use the exponential generating functions \eqref{eq:E-gen} and \eqref{eq:H-gen} which allows us to rewrite $T(v;z,w)$ in the form given in \eqref{eq:Tv-super}. The last step is possible because we are dealing with commutative elements.
\end{proof}
The partition function $T_N(v;z,w)$ is equal to $Z_N$ from the introduction \eqref{eq:Z}. In Theorem \ref{thm:Z_intro} we wrote $Z_N$ in terms of yet another conic partition function $L_N$. This result follows from \eqref{eq:L-S} below.
\begin{dfn}
Consider the matrix $W_N(x)$ associated to $U_t(\widehat{sl}_{2|1})$ vertex model, i.e. in the matrix elements $W_{\gamma}^\delta(x)$, with $\gamma,\delta\in \{0,1,2\}^N$, the label $1$ denotes the bosonic colour and $2$ denotes the fermionic colour.
We define $L_N(x)$ to be the following partition function:  \begin{align}\label{eq:L-trace}
        L_N(x) := \sum_{\alpha\in\{1,2\}^N} 
        (-1)^{m_2(\alpha)}
        m_2(\alpha)
        \bra{\alpha} W_N(x)\ket{\alpha}
\end{align}
where $m_2(\alpha)$ is the multiplicity of the fermionic index $2$ in $\alpha$. 
\end{dfn}
Using the same methods as in the proof of Theorem \ref{thm:T-shuffle-super} we can write $L_N(x)$ in terms of a shuffle product of $H_M(t^{-1})$ and $E_M(t^{-1})$ which by \eqref{eq:EEH-Wx} are equal to $D^{(1)}_M$ and $D^{(2)}_M$ up to a factor. Therefore the partition function $L_N(x)$ can be expressed as a combination of the fermionic and bosonic six vertex domain-wall partition functions:
\begin{align}
    \label{eq:L-DW}
    L_N(x) = \sum_{j=1}^N (-1)^j j D_{N-j}^{(1)}*D_j^{(2)}
\end{align}
The following proposition provides a lattice path realization of the shuffle algebra elements $S_N\in \Ac_N$.
\begin{prop}
The partition functions $L_N(x)$ and the elements $S_N\in\Ac_N$ are equal up to a factor:
\begin{align}
    \label{eq:L-S}
    L_N(x) = \frac{1-t^N}{1-q^N}S_N(x)
\end{align}
\end{prop}
\begin{proof}
Consider $T_N(x;z_0,z_1,w_1)$, i.e. the case $n=m=1$ of $T_N$.
Set $z_0=0,z_1=1$ and $w_1=w$. We compute $T_N(x;0,1,w)$ using the definition \eqref{eq:T_W-super} of $T_N$:
\begin{align}\label{eq:T_W11}
        T_N(x;0,1,w) = \sum_{\alpha\in\{1,2\}^N} 
        (-w)^{m_2(\alpha)}
        \bra{\alpha} W_N(x)\ket{\alpha}
\end{align}
By comparing \eqref{eq:L-trace} with \eqref{eq:T_W11} we see that the functions $L_N(x)$ and $T_N(x;0,1,w)$ are related by:
\begin{align}
    \label{eq:L-T}
    L_N(x)=\lim_{w\rightarrow 1} \frac{\partial}{\partial w} T_N(x;0,1,w)
\end{align}
Using \eqref{eq:Tv-super} we compute the exponential generating function of $T_N(x;0,1,w)$:
\begin{align}\label{eq:Tv11}
    \sum_{N=0}^\infty v^N T_N(x;0,1,w) =
    \exp_*\left(\sum_{k>0}\frac{1}{k}\frac{1-t^k}{1-q^k}\left(w^k-1\right) v^{k}S_k\right)
\end{align}
Next we take the sum over $N$ on both sides in \eqref{eq:L-T} with $v^N$ and compute the derivative and the limit on the r.h.s. using \eqref{eq:Tv11}. As a result we get \eqref{eq:L-S}.
\end{proof}

\section{Skew Macdonald functions and lattice paths}\label{sec:skew}
In this section we derive the connection between the trace elements $T_N$ and the skew Macdonald functions $P_{\mu/\nu}$. This connection is based on the mixed Cauchy kernel $K(x;y)$, which is defined by applying the isomorphism $\iota_z$ to the Cauchy kernel $\Pi(z,y)$ (see \cite{FHSSY}). Under the evaluation map $\ev_{\mu/\nu}$ the mixed Cauchy kernel $K(x;y)$ reproduces the skew Macdonald functions $P_{\mu/\nu}(y)$. Finally we relate the mixed Cauchy kernel with the generating function $T(v;z,w)$. 

\subsection{The mixed Cauchy kernel}
Let us recall the commuting shuffle algebra elements $F_\lambda\in \Ac$ and the evaluation map \eqref{eq:ev_skew} from Section \ref{sec:shuffle}. The functions $F_\lambda=F_\lambda(x_1\ldots x_k)$, with $\lambda\vdash k$, have an important property:
$$
\ev_\mu(F_\lambda)= \delta_{\lambda,\mu} \frac{1}{d_\lambda}
$$
Let us compute the evaluations $\ev_{\mu/\nu}$ of $F_\lambda$.
\begin{lemma}
\label{lem:shuffle-LR}
We have the product rule:
\begin{align}\label{eq:shuffle-LR-ceof}
 F_{\lambda}*F_{\nu}=\sum_{\mu}
 q^{n(\mu')-n(\lambda')-n(\nu')}\frac{c_\lambda c_\nu}{c_\mu}
 f^{\mu}_{\lambda,\nu}F_{\mu}
\end{align}
and evaluations:
\begin{align}
    \label{eq:ev_LR}
    \normalfont{\ev}_{\mu/\nu}\left(F_\lambda\right) = 
    q^{n(\mu')-n(\lambda')-n(\nu')}\frac{c_\lambda c_\nu}{c_\mu d_{\mu/\nu}}
 f^{\mu}_{\lambda,\nu}
\end{align}
\end{lemma}
\begin{proof}
We recall the product rule for Macdonald functions \eqref{eq:LR-ceof}:
$$
 P_{\lambda}P_{\nu}=\sum_{\mu}f^{\mu}_{\lambda,\nu}P_{\mu}
$$
and apply the isomorphism $\iota^{-1}$ to the l.h.s. and to the r.h.s. of this equation separately:
\begin{align*}
 \iota^{-1}\left(P_{\lambda}P_{\nu}\right)
 &=
 \frac{q^{n(\lambda')+n(\nu')}(1-t)^{|\lambda|+|\nu|}}{c_\lambda c_\nu} 
 F_\lambda *F_\nu \\
 \sum_{\mu}f^{\mu}_{\lambda,\nu}
 \iota^{-1}\left(P_{\mu}\right)
 &=\sum_{\mu}
 \frac{q^{n(\mu')}(1-t)^{|\mu|}}{c_\mu}
 f^{\mu}_{\lambda,\nu}F_{\mu}
\end{align*}
By comparing the two equations above and noting that $|\lambda|+|\mu|=|\nu|$ by the degree argument we verify \eqref{eq:shuffle-LR-ceof}. In order to prove \eqref{eq:ev_LR} we act with \eqref{eq:shuffle-LR-ceof} on $\ket{\varnothing}$ and recall \eqref{eq:F} that $F_\nu\ket{\varnothing}=\ket{\nu}$:
\begin{align*}
 F_{\lambda}\ket{\nu}=\sum_{\mu}
 q^{n(\mu')-n(\lambda')-n(\nu')}\frac{c_\lambda c_\nu}{c_\mu}
 f^{\mu}_{\lambda,\nu}\ket{\mu}
\end{align*}
This equation allows us to compute the matrix elements $\bra{\mu}F_{\lambda}\ket{\nu}$:
\begin{align*}
 \bra{\mu}F_{\lambda}\ket{\nu}=
 q^{n(\mu')-n(\lambda')-n(\nu')}\frac{c_\lambda c_\nu}{c_\mu}
 f^{\mu}_{\lambda,\nu}
\end{align*}
and combining this with \eqref{eq:sh_rep} leads to 
\eqref{eq:ev_LR}.
\end{proof}

Next we define the mixed Cauchy kernel $K(x;z)$ which we will expand in $F_\lambda(x)$ and apply to it the evaluation map ev$_{\mu/\nu}$ using the result \eqref{eq:ev_LR} of the above Lemma.
\begin{dfn}
Let $(z)$ and $(y)$ be two alphabets and $\Pi(z,y)$ the Cauchy kernel \eqref{eq:Cauchy-kernel}.
    Define the mixed Cauchy kernel:
    \begin{align}
        \label{eq:K-mixed}
        K(x;y) : = \iota_z^{-1} \left(\Pi(z,y)\right)
    \end{align}
\end{dfn}
\begin{remark}
For $K(x;y)$ we have the analogues of the exponential formula \eqref{eq:Cauchy-kernel} as well as bases expansions \eqref{eq:g-m-P} for $\Pi(z,y)$: 
\begin{align}\label{eq:K-exp}
    K(x;y)&=
        \exp_*\left(\sum_{r>0}\frac{1}{r}\frac{1-t^r}{1-q^r}
        \frac{(t-q)^r}{(1-q)^r}p_r(y)S_r \right)\\
        \label{eq:K-mon}
    K(x;y)&= \sum_{\lambda}\frac{(1-t)^{|\lambda|}}{(1-q)^{|\lambda|}} m_\lambda(y) E_\lambda(t^{-1}) \\
    \label{eq:K-Mac}
    K(x;y)&=  \sum_{\lambda} \frac{q^{n(\lambda')}(1-t)^{|\lambda|}}{c'_\lambda} P_\lambda(y) F_\lambda  
\end{align}
\end{remark}
\begin{proof}
    The exponential formula \eqref{eq:K-exp} is obtained by inserting the exponential formula \eqref{eq:Cauchy-kernel} for the Cauchy kernel in \eqref{eq:K-mixed} and applying $\iota_z$ to the power sums \eqref{eq:iota_S}. The monomial expansion \eqref{eq:K-mon} is a consequence of \eqref{eq:g-m-P} and the isomorphism formula \eqref{eq:iota_E}\footnote{The monomial expansion of $K(x;z)$ was studied in \cite{FHSSY} in relation with the tableaux formula for the Macdonald functions.}.  The Macdonald expansion \eqref{eq:K-Mac} follows from \eqref{eq:g-m-P}, \eqref{eq:bc-coef} and the isomorphism formula \eqref{eq:iso}.
\end{proof}
\begin{lemma}\label{lem:K-skew-Mac}
    Let $\mu/\nu$ be a skew partition and $k=|\mu|-|\nu|$. The mixed Cauchy kernel $K(x;y)$ reproduces the skew Macdonald function $P_{\mu/\nu}(y)$ under $\normalfont{\ev}_{\mu/\nu}$:
\begin{align}
    \label{eq:K-skew-Mac}
    P_{\mu/\nu}(y) 
    =\tilde a_{\mu,\nu}
    \normalfont{\ev}_{\mu/\nu} \left(K(x;y)\right),
    \qquad
     \tilde a_{\mu,\nu} : = \frac{q^{n(\nu')-n(\mu')}c_\mu' d_{\mu/\nu}}{(1-t)^{k} c_\nu'}
\end{align}    
\end{lemma}
\begin{proof}
We first note that the evaluation map $\ev_{\mu/\nu}$ acts on $F\in \Ac_j$ and gives zero unless $j=k$, the number of boxes in $\mu/\nu$. We apply $\ev_{\mu/\nu}$ to $K(x;y)$ written in the form \eqref{eq:K-Mac}:
\begin{align*}
    \ev_{\mu/\nu} \left(K(x;y)\right)= 
    &(1-t)^{k}\sum_{\lambda\vdash k} \frac{q^{n(\lambda')}}{c'_\lambda} P_\lambda(y) \ev_{\mu/\nu}  \left(F_\lambda  \right) \\
    = 
    &(1-t)^{k}
    \frac{q^{n(\mu')-n(\nu')} c_\nu}{c_\mu d_{\mu/\nu}}
    \sum_{\lambda\vdash k} b_\lambda
    f^{\mu}_{\lambda,\nu}P_\lambda(y)\\
    &=
    \frac{(1-t)^{k}q^{n(\mu')-n(\nu')} c_\nu'}{c_\mu' d_{\mu/\nu}} P_{\mu/\nu}(y)
\end{align*}
where in the second line we computed $\ev_{\mu/\nu}(F_\lambda)$ using \eqref{eq:ev_LR} and then used the definition of $b_\lambda$ \eqref{eq:bc-coef}. To get the result in the third line we recalled the formula for the skew Macdonald functions \eqref{eq:skew_M}:
$$
P_{\mu/\nu} = \sum_{\lambda} \frac{b_\nu b_\lambda}{b_\mu}  f^\mu_{\lambda,\nu}  P_\lambda
$$
This computation proves \eqref{eq:K-skew-Mac}. 
\end{proof}

\subsection{The trace elements and the skew Macdonald functions}\label{sec:Trace_Mac}
The elements of the shuffle algebra $S_r\in \Ac_r$ are mapped to the power sums \eqref{eq:iota_S} under $\iota$, therefore the exponential generating function $T(v;z,w)$ given in \eqref{eq:Tv-super} is related to the Cauchy kernel \eqref{eq:Cauchy-kernel} under this isomorphism.

The generating function $T(v;z,w)$ \eqref{eq:Tv-super} is a symmetric function in two finite alphabets $(w)=(w_1\ldots w_m)$ and $(z)=(z_1\ldots z_n)$ and can be written in terms of the power sums:
\begin{align}\label{eq:Tv-power_sums}
    T(v;z,w) =
    \exp_*\left(\sum_{r>0}\frac{1}{r}\frac{1-t^r}{1-q^r}\left(p_r(w)- p_r(z)-\frac{q^r-t^r}{1-t^r}z_0^r  \right) v^{r}S_r\right)
\end{align}
Therefore it is given by the mixed Cauchy kernel $K(x;y)$ to which one needs to apply the following transformation:
\begin{align}
    \label{eq:pi}
\pi_{w,z}:\quad p_r(y) \rightarrow p_r(w_1\ldots w_m) -p_r(z_1\ldots z_n) - \frac{q^r-t^r}{1-t^r}z_0^r 
\end{align}
The map $\pi_{w,z}$ is a  {\it plethystic substitution} and for a symmetric function $f(y)$ we write:
\begin{align}
    \label{eq:pleth}
\pi_{w,z}\left( f(y)  \right) = f\left[w -z -\frac{q-t}{1-t}z_0\right]
\end{align}
By comparing \eqref{eq:Tv-power_sums} with \eqref{eq:K-exp} we find:
\begin{align}
    \label{eq:T-K}
    T(c v;z,w) = \pi_{w,z}\left(K(x;y)\right),
    \qquad c = \frac{t-q}{1-q}
\end{align}
This identity together with Lemma \ref{lem:K-skew-Mac} leads to the following result:
\begin{thm}\label{thm:PT}
    Let $\mu/\nu$ be a skew partition and $N=|\mu|-|\nu|$, we have the identity:
\begin{align}
    \label{eq:T-skew-Mac}
    P_{\mu/\nu}\left[w -z -\frac{q-t}{1-t}z_0\right]
    =a_{\mu,\nu}
    \normalfont{\ev}_{\mu/\nu} \left(T_N(x;z,w)\right),
\end{align}    
where 
\begin{align}
    \label{eq:a}
    a_{\mu,\nu} : = \frac{(t-q)^N}{(1-q)^N (1-t)^{N}}\frac{q^{n(\nu')-n(\mu')}c_\mu' d_{\mu/\nu}}{ c_\nu'}
\end{align}
\end{thm}

\subsection{Example}\label{sec:skew_example}
In this section we show how to compute $P_{(2,1)/(1)}(w_1,w_2;z_{0})$ using Theorem \ref{thm:PT}. In this case we have $N=2$, $n=0$ and $m=2$, so according to Theorem \ref{thm:PT} we need to calculate:
$$
a_{(2,1),(1)}^{-1}P_{(2,1)/(1)}(w_1,w_2;z_{0})
=\ev_{(2,1)/(1)} \left(T_2(x_1,x_2;z_0,w_1,w_2)\right)= 
T_2(q,t^{-1};z_0,w_1,w_2)
$$
where $T_2(q,t^{-1};z_0,w_1,w_2)$ is the partition function of lattice path configurations with two fermionic colours and boundary conditions as in \eqref{eq:T_def}. Let the two colours be red and green where red corresponds to the smaller label and the local configurations are given by the Boltzmann weights of the $R$-matrix of $U_{t}(\widehat{ sl}_{1|2})$, given in \eqref{tikz:coloredvertices}. We can find all possible configurations in $T_2(q,t^{-1};z_0,w_1,w_2)$ and write:
\begin{align}\label{eq:loops}
&a_{(2,1),(1)}^{-1}P_{(2,1)/(1)}(w_1,w_2;z_{0})=
\begin{tikzpicture}[scale=0.6, baseline=(current  bounding  box.center)]
\draw[invarrow=0.95] (1,0.5) -- node[pos=1,above] {$\scriptstyle 0$} (1,2.5);
\draw[invarrow=0.95] (2,0.5) -- node[pos=1,above] {$\scriptstyle 0$} (2,2.5);
\foreach\i/\lab in {2/,1/\delta_N}
\draw[invarrow=0.95] (0.5,\i) -- (2.5,\i);
\draw [rounded corners=5pt]  (2,0.5) -- (2,0.5-0.25) -- (2+0.5+0.25,0.5-0.25) -- (2+0.5+0.25,1) -- (2.5,1);
\draw [rounded corners=5pt]  (1,0.5) -- (1,0.5-0.5) -- (2+0.5+0.5,0.5-0.5) -- (2+0.5+0.5,2) -- (2.5,2);
\foreach\i/\lab in {2/q ,1/t^{-1}}
\node at (-0.4,\i) {$\scriptstyle \lab$};
\foreach\i/\lab in {1/q^{2} ,2/q t^{-1}}
\node at (\i,3.5) {$\scriptstyle \lab$};
\foreach\i/\lab in {2/0 ,1/0}
\node at (0.3,\i) {$\scriptstyle \lab$};
\end{tikzpicture}\\[1 em]
=
&
\left(
\begin{tikzpicture}[scale=0.6, baseline=(current  bounding  box.center)]
\draw (1,0.5) -- node[pos=1,above] {$\scriptstyle 0$} (1,2.5);
\draw (2,0.5) -- node[pos=1,above] {$\scriptstyle 0$} (2,2.5);
\foreach\i/\lab in {2/,1/\delta_N}
\draw (0.5,\i) -- (2.5,\i);
\draw [rounded corners=5pt]  (2,0.5) -- (2,0.5-0.25) -- (2+0.5+0.25,0.5-0.25) -- (2+0.5+0.25,1) -- (2.5,1);
\draw [rounded corners=5pt]  (1,0.5) -- (1,0.5-0.5) -- (2+0.5+0.5,0.5-0.5) -- (2+0.5+0.5,2) -- (2.5,2);
\foreach\i/\lab in {2/0 ,1/0}
\node at (0.3,\i) {$\scriptstyle \lab$};
\end{tikzpicture}\right) z^{2}_{0}
+
\left(
\begin{tikzpicture}[scale=0.6, baseline=(current  bounding  box.center)]
\draw (1,0.5) -- node[pos=1,above] {$\scriptstyle 0$} (1,2.5);
\draw (2,0.5) -- node[pos=1,above] {$\scriptstyle 0$} (2,2.5);
\foreach\i/\lab in {2/,1/\delta_N}
\draw (0.5,\i) -- (2.5,\i);
\draw [rounded corners=5pt,red, ultra thick]  (2,0.5) -- (2,0.25) -- (2.75,0.25) -- (2.75,1) -- (2.5,1);
\draw[red,ultra thick] (2,0.5)--(2,1)--(2.5,1);
\draw [rounded corners=5pt]  (1,0.5) -- (1,0.5-0.5) -- (2+0.5+0.5,0.5-0.5) -- (2+0.5+0.5,2) -- (2.5,2);
\foreach\i/\lab in {2/0 ,1/0}
\node at (0.3,\i) {$\scriptstyle \lab$};
\end{tikzpicture}
+
\begin{tikzpicture}[scale=0.6, baseline=(current  bounding  box.center)]
\draw (1,0.5) -- node[pos=1,above] {$\scriptstyle 0$} (1,2.5);
\draw (2,0.5) -- node[pos=1,above] {$\scriptstyle 0$} (2,2.5);
\foreach\i/\lab in {2/,1/\delta_N}
\draw (0.5,\i) -- (2.5,\i);
\draw [rounded corners=5pt]  (2,0.5) -- (2,0.25) -- (2.75,0.25) -- (2.75,1) -- (2.5,1);
\draw[red,ultra thick] (1,0.5)--(1,1)--(2,1)--(2,2)--(2.5,2);
\draw [rounded corners=5pt,red, ultra thick]  (1,0.5) -- (1,0.5-0.5) -- (2+0.5+0.5,0.5-0.5) -- (2+0.5+0.5,2) -- (2.5,2);
\foreach\i/\lab in {2/0 ,1/0}
\node at (0.3,\i) {$\scriptstyle \lab$};
\end{tikzpicture}
+
\begin{tikzpicture}[scale=0.6, baseline=(current  bounding  box.center)]
\draw (1,0.5) -- node[pos=1,above] {$\scriptstyle 0$} (1,2.5);
\draw (2,0.5) -- node[pos=1,above] {$\scriptstyle 0$} (2,2.5);
\foreach\i/\lab in {2/,1/\delta_N}
\draw (0.5,\i) -- (2.5,\i);
\draw [rounded corners=5pt]  (2,0.5) -- (2,0.25) -- (2.75,0.25) -- (2.75,1) -- (2.5,1);
\draw[red,ultra thick] (1,0.5)--(1,2)--(2.5,2);
\draw [rounded corners=5pt,red, ultra thick]  (1,0.5) -- (1,0.5-0.5) -- (2+0.5+0.5,0.5-0.5) -- (2+0.5+0.5,2) -- (2.5,2);
\foreach\i/\lab in {2/0 ,1/0}
\node at (0.3,\i) {$\scriptstyle \lab$};
\end{tikzpicture}\right)z_{0}w_{1}\nonumber\\[1 em]
+&
\left(
\begin{tikzpicture}[scale=0.6, baseline=(current  bounding  box.center)]
\draw (1,0.5) -- node[pos=1,above] {$\scriptstyle 0$} (1,2.5);
\draw (2,0.5) -- node[pos=1,above] {$\scriptstyle 0$} (2,2.5);
\foreach\i/\lab in {2/,1/\delta_N}
\draw (0.5,\i) -- (2.5,\i);
\draw [rounded corners=5pt,gcol, ultra thick]  (2,0.5) -- (2,0.25) -- (2.75,0.25) -- (2.75,1) -- (2.5,1);
\draw[gcol,ultra thick] (2,0.5)--(2,1)--(2.5,1);
\draw [rounded corners=5pt]  (1,0.5) -- (1,0.5-0.5) -- (2+0.5+0.5,0.5-0.5) -- (2+0.5+0.5,2) -- (2.5,2);
\foreach\i/\lab in {2/0 ,1/0}
\node at (0.3,\i) {$\scriptstyle \lab$};
\end{tikzpicture}
+
\begin{tikzpicture}[scale=0.6, baseline=(current  bounding  box.center)]
\draw (1,0.5) -- node[pos=1,above] {$\scriptstyle 0$} (1,2.5);
\draw (2,0.5) -- node[pos=1,above] {$\scriptstyle 0$} (2,2.5);
\foreach\i/\lab in {2/,1/\delta_N}
\draw (0.5,\i) -- (2.5,\i);
\draw [rounded corners=5pt]  (2,0.5) -- (2,0.25) -- (2.75,0.25) -- (2.75,1) -- (2.5,1);
\draw[gcol,ultra thick] (1,0.5)--(1,1)--(2,1)--(2,2)--(2.5,2);
\draw [rounded corners=5pt,gcol, ultra thick]  (1,0.5) -- (1,0.5-0.5) -- (2+0.5+0.5,0.5-0.5) -- (2+0.5+0.5,2) -- (2.5,2);
\foreach\i/\lab in {2/0 ,1/0}
\node at (0.3,\i) {$\scriptstyle \lab$};
\end{tikzpicture}
+
\begin{tikzpicture}[scale=0.6, baseline=(current  bounding  box.center)]
\draw (1,0.5) -- node[pos=1,above] {$\scriptstyle 0$} (1,2.5);
\draw (2,0.5) -- node[pos=1,above] {$\scriptstyle 0$} (2,2.5);
\foreach\i/\lab in {2/,1/\delta_N}
\draw (0.5,\i) -- (2.5,\i);
\draw [rounded corners=5pt]  (2,0.5) -- (2,0.25) -- (2.75,0.25) -- (2.75,1) -- (2.5,1);
\draw[gcol,ultra thick] (1,0.5)--(1,2)--(2.5,2);
\draw [rounded corners=5pt,gcol, ultra thick]  (1,0.5) -- (1,0.5-0.5) -- (2+0.5+0.5,0.5-0.5) -- (2+0.5+0.5,2) -- (2.5,2);
\foreach\i/\lab in {2/0 ,1/0}
\node at (0.3,\i) {$\scriptstyle \lab$};
\end{tikzpicture}
\right)z_{0}w_{2}\nonumber\\[1 em]
+
&\left(
\begin{tikzpicture}[scale=0.6, baseline=(current  bounding  box.center)]
\draw (1,0.5) -- node[pos=1,above] {$\scriptstyle 0$} (1,2.5);
\draw (2,0.5) -- node[pos=1,above] {$\scriptstyle 0$} (2,2.5);
\foreach\i/\lab in {2/,1/\delta_N}
\draw (0.5,\i) -- (2.5,\i);
\draw [rounded corners=5pt,red, ultra thick]  (2,0.5) -- (2,0.25) -- (2.75,0.25) -- (2.75,1) -- (2.5,1);
\draw[red,ultra thick] (2,0.5)--(2,1)--(2.5,1);
\draw [rounded corners=5pt,red,ultra thick]  (1,0.5) -- (1,0.5-0.5) -- (2+0.5+0.5,0.5-0.5) -- (2+0.5+0.5,2) -- (2.5,2);
\draw[red,ultra thick] (2,0.5)--(2,1)--(2.5,1);
\draw[red,ultra thick] (1,0.5)--(1,1)--(2,1)--(2,2)--(2.5,2);
\foreach\i/\lab in {2/0 ,1/0}
\node at (0.3,\i) {$\scriptstyle \lab$};
\end{tikzpicture}
+
\begin{tikzpicture}[scale=0.6, baseline=(current  bounding  box.center)]
\draw (1,0.5) -- node[pos=1,above] {$\scriptstyle 0$} (1,2.5);
\draw (2,0.5) -- node[pos=1,above] {$\scriptstyle 0$} (2,2.5);
\foreach\i/\lab in {2/,1/\delta_N}
\draw (0.5,\i) -- (2.5,\i);
\draw [rounded corners=5pt,red, ultra thick]  (2,0.5) -- (2,0.25) -- (2.75,0.25) -- (2.75,1) -- (2.5,1);
\draw[red,ultra thick] (2,0.5)--(2,1)--(2.5,1);
\draw [rounded corners=5pt,red,ultra thick]  (1,0.5) -- (1,0.5-0.5) -- (2+0.5+0.5,0.5-0.5) -- (2+0.5+0.5,2) -- (2.5,2);
\draw[red,ultra thick] (2,0.5)--(2,1)--(2.5,1);
\draw[red,ultra thick] (1,0.5)--(1,2)--(2.5,2);
\foreach\i/\lab in {2/0 ,1/0}
\node at (0.3,\i) {$\scriptstyle \lab$};
\end{tikzpicture}
\right) w^{2}_{1}\
+
\left(
\begin{tikzpicture}[scale=0.6, baseline=(current  bounding  box.center)]
\draw (1,0.5) -- node[pos=1,above] {$\scriptstyle 0$} (1,2.5);
\draw (2,0.5) -- node[pos=1,above] {$\scriptstyle 0$} (2,2.5);
\foreach\i/\lab in {2/,1/\delta_N}
\draw (0.5,\i) -- (2.5,\i);
\draw [rounded corners=5pt,gcol, ultra thick]  (2,0.5) -- (2,0.25) -- (2.75,0.25) -- (2.75,1) -- (2.5,1);
\draw[gcol,ultra thick] (2,0.5)--(2,1)--(2.5,1);
\draw [rounded corners=5pt,gcol,ultra thick]  (1,0.5) -- (1,0.5-0.5) -- (2+0.5+0.5,0.5-0.5) -- (2+0.5+0.5,2) -- (2.5,2);
\draw[gcol,ultra thick] (2,0.5)--(2,1)--(2.5,1);
\draw[gcol,ultra thick] (1,0.5)--(1,1)--(2,1)--(2,2)--(2.5,2);
\foreach\i/\lab in {2/0 ,1/0}
\node at (0.3,\i) {$\scriptstyle \lab$};
\end{tikzpicture}
+
\begin{tikzpicture}[scale=0.6, baseline=(current  bounding  box.center)]
\draw (1,0.5) -- node[pos=1,above] {$\scriptstyle 0$} (1,2.5);
\draw (2,0.5) -- node[pos=1,above] {$\scriptstyle 0$} (2,2.5);
\foreach\i/\lab in {2/,1/\delta_N}
\draw (0.5,\i) -- (2.5,\i);
\draw [rounded corners=5pt,gcol, ultra thick]  (2,0.5) -- (2,0.25) -- (2.75,0.25) -- (2.75,1) -- (2.5,1);
\draw[gcol,ultra thick] (2,0.5)--(2,1)--(2.5,1);
\draw [rounded corners=5pt,gcol,ultra thick]  (1,0.5) -- (1,0.5-0.5) -- (2+0.5+0.5,0.5-0.5) -- (2+0.5+0.5,2) -- (2.5,2);
\draw[gcol,ultra thick] (2,0.5)--(2,1)--(2.5,1);
\draw[gcol,ultra thick] (1,0.5)--(1,2)--(2.5,2);
\foreach\i/\lab in {2/0 ,1/0}
\node at (0.3,\i) {$\scriptstyle \lab$};
\end{tikzpicture}
\right)w^{2}_{2}\nonumber\\[1 em]
+
&
\left(
\begin{tikzpicture}[scale=0.6, baseline=(current  bounding  box.center)]
\draw (1,0.5) -- node[pos=1,above] {$\scriptstyle 0$} (1,2.5);
\draw (2,0.5) -- node[pos=1,above] {$\scriptstyle 0$} (2,2.5);
\foreach\i/\lab in {2/,1/\delta_N}
\draw (0.5,\i) -- (2.5,\i);
\draw [rounded corners=5pt,gcol, ultra thick]  (2,0.5) -- (2,0.25) -- (2.75,0.25) -- (2.75,1) -- (2.5,1);
\draw[gcol,ultra thick] (2,0.5)--(2,1)--(2.5,1);
\draw [rounded corners=5pt,red,ultra thick]  (1,0.5) -- (1,0.5-0.5) -- (2+0.5+0.5,0.5-0.5) -- (2+0.5+0.5,2) -- (2.5,2);
\draw[red,ultra thick] (1,0.5)--(1,1)--(2,1)--(2,2)--(2.5,2);
\foreach\i/\lab in {2/0 ,1/0}
\node at (0.3,\i) {$\scriptstyle \lab$};
\end{tikzpicture}+
\begin{tikzpicture}[scale=0.6, baseline=(current  bounding  box.center)]
\draw (1,0.5) -- node[pos=1,above] {$\scriptstyle 0$} (1,2.5);
\draw (2,0.5) -- node[pos=1,above] {$\scriptstyle 0$} (2,2.5);
\foreach\i/\lab in {2/,1/\delta_N}
\draw (0.5,\i) -- (2.5,\i);
\draw [rounded corners=5pt,gcol, ultra thick]  (2,0.5) -- (2,0.25) -- (2.75,0.25) -- (2.75,1) -- (2.5,1);
\draw[gcol,ultra thick] (2,0.5)--(2,1)--(2.5,1);
\draw [rounded corners=5pt,red,ultra thick]  (1,0.5) -- (1,0.5-0.5) -- (2+0.5+0.5,0.5-0.5) -- (2+0.5+0.5,2) -- (2.5,2);
\draw[red,ultra thick] (1,0.5)--(1,2)--(2.5,2);
\foreach\i/\lab in {2/0 ,1/0}
\node at (0.3,\i) {$\scriptstyle \lab$};
\end{tikzpicture}
+
\begin{tikzpicture}[scale=0.6, baseline=(current  bounding  box.center)]
\draw (1,0.5) -- node[pos=1,above] {$\scriptstyle 0$} (1,2.5);
\draw (2,0.5) -- node[pos=1,above] {$\scriptstyle 0$} (2,2.5);
\foreach\i/\lab in {2/,1/\delta_N}
\draw (0.5,\i) -- (2.5,\i);
\draw [rounded corners=5pt,red, ultra thick]  (2,0.5) -- (2,0.25) -- (2.75,0.25) -- (2.75,1) -- (2.5,1);
\draw[red,ultra thick] (2,0.5)--(2,1)--(2.5,1);
\draw [rounded corners=5pt,gcol,ultra thick]  (1,0.5) -- (1,0.5-0.5) -- (2+0.5+0.5,0.5-0.5) -- (2+0.5+0.5,2) -- (2.5,2);
\draw[gcol,ultra thick] (1,0.5)--(1,1)--(2,1)--(2,2)--(2.5,2);
\foreach\i/\lab in {2/0 ,1/0}
\node at (0.3,\i) {$\scriptstyle \lab$};
\end{tikzpicture}
+
\begin{tikzpicture}[scale=0.6, baseline=(current  bounding  box.center)]
\draw (1,0.5) -- node[pos=1,above] {$\scriptstyle 0$} (1,2.5);
\draw (2,0.5) -- node[pos=1,above] {$\scriptstyle 0$} (2,2.5);
\foreach\i/\lab in {2/,1/\delta_N}
\draw (0.5,\i) -- (2.5,\i);
\draw [rounded corners=5pt,red, ultra thick]  (2,0.5) -- (2,0.25) -- (2.75,0.25) -- (2.75,1) -- (2.5,1);
\draw[red,ultra thick] (2,0.5)--(2,1)--(2.5,1);
\draw [rounded corners=5pt,gcol,ultra thick]  (1,0.5) -- (1,0.5-0.5) -- (2+0.5+0.5,0.5-0.5) -- (2+0.5+0.5,2) -- (2.5,2);
\draw[gcol,ultra thick] (1,0.5)--(1,2)--(2.5,2);
\foreach\i/\lab in {2/0 ,1/0}
\node at (0.3,\i) {$\scriptstyle \lab$};
\end{tikzpicture}\right)w_{1}w_{2}
\nonumber
\end{align}
By substituting the Boltzmann weights \eqref{tikz:coloredvertices} we compute the coefficients in front of each monomial:
\begin{align}\label{eq:P-example}
 P_{(2,1)/(1)}(w_1,w_2;z_{0})=
 &a_{(2,1),(1)}
 \Big{(}z^{2}_{0}+
 \frac{(1-t) (2+q+t + 2 q t)}{(1+q) (1+t) (t-q)}
 (z_{0}w_{1}+z_{0}w_{2})\\
 +&
 \frac{(1-t) \left(1-q t^2\right)}{(1+q) (1+t) (q-t)^2}
 (w^{2}_{1}+w^{2}_{2})+
 \frac{(1-t)^2 (2+q+t+2 q t)}{(1+q) (1+t) (q-t)^2}w_{1}w_{2}   \Big{)}
 \nonumber
\end{align}
where the coefficient $a_{(2,1),(1)}$ reads:
\begin{align}
    \label{eq:acoef_example}
a_{(2,1),(1)}= 
\frac{(1+q)(1+t)(q-t)^2}{(1-t) (1-q t^2)}
\end{align}
Therefore we computed $P_{(2,1)/(1)}(w_1,w_2;z_{0})$. By setting $z_0=0$ in \eqref{eq:P-example} we recover the standard skew Macdonald polynomial $P_{(2,1)/(1)}(w_1,w_2)$:
$$
P_{(2,1)/(1)}(w_1,w_2) = 
w_1^2 + 
\frac{(1-t)  (2+q+t+2 q t)}{1-q t^2}w_1 w_2 
+w_2^2
$$

\appendix
\section{A shuffle product identity}\label{app:EH}
The determinant formula \eqref{eq:H_def} for $H_k(t^{-1})$ is known as the Izergin determinant. This determinant satisfies various summation identities. For our purposes a useful identity is given by a subset formula from \cite{BWZJ}:
\begin{align}
    & \frac{\prod_{i,j=1}^k(x_i-y_j)}{\prod_{1\leq i<j\leq k}(x_i-x_j)(y_j-y_i)}
    \det_{1\leq i,j \leq k} \frac{(1-t)x_i}{(x_i-y_j)(y_j-t x_i)}\nonumber
    \\
    =&\sum_{r=0}^k 
    (-1)^{r} t^{r(r-1)/2}
    \sum_{\substack{S\subseteq [k]\\ |S|=r}}
    \prod_{i\in S}
    \prod_{j\in S^c}
    \frac{x_j-t x_i}{x_j-x_i} 
    \prod_{i\in S}
    \prod_{j=1}^k
    \frac{y_j-x_i}{y_j-t x_i} 
    \label{eq:det_subset-app}
\end{align}
When we set $y_i=q x_i$ the left-hand side above matches with the formula \eqref{eq:H_def} for $H_k(t^{-1})$ up to a factor. This leads to a summation formula for $H_k(t^{-1})$:
\begin{align}\label{eq:H_subset-app}
    H_k(t^{-1}) =   \frac{(1-q t^{-1})^k}{(1-t^{-1})^k}
    \sum_{r=0}^k 
    (-1)^{r} t^{-r(r+1)/2}
    \sum_{\substack{S\subseteq [k]\\ |S|=r}}
    \prod_{i\in S}
    \prod_{j\in S^c}
    \frac{x_i-t^{-1} x_j}{x_i-x_j} 
    \prod_{i\in S}
    \prod_{j=1}^k
    \frac{x_i-q x_j}{x_i-q t^{-1} x_j} 
\end{align}
The sum over subsets $S$ has the structure of the shuffle product \eqref{eq:sh_sub}. To see this we split the product over $j$ in the last factor in \eqref{eq:H_subset-app}
into two products: $j\in S$ and $j\in S^c$. The second product can be combined with the first factor in \eqref{eq:H_subset-app} giving:
\begin{align*}
 &   \sum_{\substack{S\subseteq [k]\\ |S|=r}}
    \prod_{i\in S}
    \prod_{j\in S^c}
    \frac{x_i-t^{-1} x_j}{x_i-x_j} 
    \prod_{i\in S}
    \prod_{j=1}^k
    \frac{x_i-q x_j}{x_i-q t^{-1} x_j} =
    \sum_{\substack{S\subseteq [k]\\ |S|=r}}
    \prod_{i\in S}
    \prod_{j\in S^c}
    \frac{x_i-t^{-1} x_j}{x_i-x_j} 
    \frac{x_i-q x_j}{x_i-q t^{-1} x_j} 
    \prod_{i,j\in S}
    \frac{x_i-q x_j}{x_i-q t^{-1} x_j} \\
    =
&\frac{(1-q)^{r}}{(1-q t^{-1})^{r}} t^{r(r-1)/2}    
    \sum_{\substack{S\subseteq [k]\\ |S|=r}}
    E_r(q;x_S)
    \prod_{i\in S}
    \prod_{j\in S^c}
    \zeta(x_j/x_i)
    =
\frac{(1-q)^{r}}{(1-q t^{-1})^{r}} t^{r(r-1)/2}    E_{k-r}(t q^{-1})*E_{r}(q)
\end{align*}
where in the second line we used \eqref{eq:zeta}, Definition \ref{def:E-el} of $E_r(q)$ and $E_{k-r}(t q^{-1})$ (which is equal to $1$) and the formula for the shuffle product \eqref{eq:sh_sub}.  As a consequence the element $H_k(t^{-1})$ admits the formula:
\begin{align}\label{eq:H_E-app}
    H_k(t^{-1}) =   \frac{(1-q t^{-1})^k}{(1-t^{-1})^k}
    \sum_{r=0}^k 
    (-t)^{-r}
    \frac{(1-q)^{r}}{(1-q t^{-1})^{r}}
    E_{k-r}(t q^{-1})*E_{r}(q)
\end{align}
We can derive analogously the identities for $H_k(q)$ and $H_k(t q^{-1})$ and summarize all three formulas by:
\begin{align}\label{eq:Hq_E-app}
    H_k(q_a) =   
    \sum_{r=0}^k 
    q_c^{k-r}
    \left(\frac{1-q_b}{1-q_b q_c}\right)^{k-r}
    \left(\frac{1-q_c}{1-q_b q_c}\right)^{r}
    E_{k-r}(q_b)*E_{r}(q_c)
\end{align}
where $(a,b,c)$ is a permutation of $(1,2,3)$.

\section{\texorpdfstring{The supersymmetric $F$-matrix}{}}\label{app:F}
In this section we discuss the connection between our $F$-matrix \eqref{eq:F-matrix-super} and the graded $F$-matrix of \cite{YZZ} and prove Lemma \ref{lem:F-prop-super}.

\subsection{\texorpdfstring{The supersymmetric $F$-matrix and gradation}{}}
Let $R^\star$ and $F^\star$ be the $R$ and $F$ matrices of \cite{YZZ} of the algebra $U_t(\widehat{sl}_{n+1|m})$. We rewrite $R^\star$ and $F^\star$ in the {\it multiplicative} convention (set $x=e^{-2u}$ and $t=e^{-2\eta}$ to recover the {\it additive} convention of \cite{YZZ}) and label vectors from $0$ to $n+m$ (as opposed to $1$ to $n+m+1$ as in \cite{YZZ}):
\begin{align}\label{eq:R-matrix-super_graded}
  R^\star(x) &= \sum_{i=0}^{n+m}
\left(\frac{t-x}{1-t x}\right)^{\delta_{i<m}}E^{(ii)}\otimes E^{(ii)}
+
t^{1/2}\frac{1-x}{1-t x}
\sum_{0\leq i\neq j\leq n+m}
E^{(ii)}\otimes E^{(jj)}+\\
&+\frac{1-t}{1-t x}
\sum_{0\leq j\neq i\leq n+m}
(-1)^{\delta_{i<m}\delta_{j<m}}
x^{\delta_{i>j}}
E^{(ij)}\otimes E^{(ji)}\nonumber
\end{align}
and:
\begin{align}\label{eq:F-super_graded}
    F^*_{N}(x_1\dots x_N) = 
    \sum_{\sigma \in \mathcal{S}_N} \sum_{(k_1\dots k_N)\in \mathcal{J}_\sigma} 
    \prod_{\substack{1\leq i< j\leq N\\ k_i=k_j<m}} (1+t)\frac{x_{\sigma(i)}-x_{\sigma(j)}}{x_{\sigma(i)}-t x_{\sigma(j)}}\cdot
    \prod_{i=1}^N E^{(k_i k_i)}_{\sigma(i)} R^*_\sigma
\end{align}
The transformation between $R(x)$ from \eqref{eq:R_sup} and $R^\star(x)$ involves: change of $\mathbb Z_2$ gradation of the vector space, spectral parameter dependent gauge transform, relabeling of the basis vectors\footnote{The fermionic labels in \eqref{eq:R-matrix-super_graded} are associated with the vectors $\ket{0},\ldots , \ket{m-1}$ and in our conventions the fermionic labels are associated with the vectors $\ket{n+1},\ldots , \ket{n+m}$.} and a transformation known as the stochastic twist (see e.g. \cite{KMMO}). In this section we focus on the change of $\mathbb Z_2$ gradation because this transformation affects the rational function appearing in \eqref{eq:F-super_graded} and explains the form of the $F$-matrix in \eqref{eq:F-matrix-super}.

The change of $\mathbb Z_2$ gradation can be expressed using two diagonal matrices:
\begin{align}\label{eq:Dlr}
D_l := \sum_{i,j=0}^{n+m} (-1)^{\delta_{j<i<m}}E^{(i i)}\otimes E^{(j j)},
\qquad
D_r := \sum_{i,j=0}^{n+m} (-1)^{\delta_{j\leq i<m }}E^{(i i)}\otimes E^{(j j)}
\end{align}
We define the $R$-matrix $\tilde R$:
\begin{align}\label{eq:R-tilde}
    \tilde R(x) := D_l R^\star(x) D_r =D_r R^\star(x) D_l
\end{align}
With the present definitions we have:
\begin{align}
    \label{eq:R1}
    &R^\star(1) = P^\star,
    \qquad
   \tilde R(1) = P,\\
   \label{eq:P-Pstar}
   &P = D_l P^\star D_r =D_r P^\star D_l
\end{align}
where $P^\star$ is the graded permutation matrix. The $R^\star$ matrix satisfies the graded Yang--Baxter equation (see \cite{YZZ}) while $\tilde R$  satisfies the standard Yang--Baxter equation. The $R^\star$ and $F^\star$ matrices satisfy:
\begin{align}\label{eq:F-prop-super_graded}
F^*_N(x) = P^*_\sigma F^*_N(x_\sigma) P_{\sigma^{-1}}^*R^*_\sigma
\end{align}
and $\tilde R$ satisfies:
\begin{align}\label{eq:FR_ungraded}
\tilde F_N(x) = P_\sigma \tilde F_N(x_\sigma)P_{\sigma^{-1}} \tilde R_\sigma
\end{align}
where the matrix $\tilde F$ can be chosen to be of the form \eqref{eq:F-matrix-super} (up to a relabeling of the basis vectors):
\begin{align}\label{eq:F-matrix-super-ungraded}
    \tilde F_{N}(x_1\dots x_N) = 
    \sum_{\sigma \in \mathcal{S}_N} \sum_{(k_1\dots k_N)\in \mathcal{J}_\sigma} 
    \prod_{\substack{1\leq i< j\leq N\\ k_i=k_j\\k_i<m}}
    \frac{x_{\sigma(i)}+x_{\sigma(j)}}{x_{\sigma(i)}-t x_{\sigma(j)}}\cdot
    \prod_{i=1}^N E^{(k_i k_i)}_{\sigma(i)} \tilde  R_\sigma
\end{align}
The essential difference between $F^\star$ and $\tilde F$ is that in $\tilde F$ the factors $x_{\sigma(i)}-x_{\sigma(j)}$ are absent. 
The matrices $\tilde F$ and $F^\star$ are related by:
\begin{align}\label{eq:Ftilde-Fstar}
    \tilde F_N(x) =  X_N(x) D_N F^\star_N(x) D_N
\end{align}
where:
\begin{align}
    X_N&=\sum_{\alpha\in\{0\ldots n+m\}^N} 
    \prod_{1\leq i<j\leq N}\left(\frac{1}{1+t}\frac{x_i+x_j}{x_i-x_j} \right)^{\delta_{\alpha_j= \alpha_i<m}}
    \bigotimes_{i=1}^N E^{(\alpha_i \alpha_i)}\\
    D_N&=\sum_{\alpha\in\{0\ldots n+m\}^N} 
    \prod_{1\leq i<j\leq N}
    (-1)^{\delta_{\alpha_j\leq \alpha_i<m}} \bigotimes_{i=1}^N E^{(\alpha_i \alpha_i)}
\end{align}
Checking \eqref{eq:Ftilde-Fstar} amounts to performing elementary algebraic manipulations. By using \eqref{eq:Ftilde-Fstar} and \eqref{eq:F-prop-super_graded} one can prove \eqref{eq:FR_ungraded}:
\begin{align*}
    \tilde F_N(x) =  X_N(x) D_N F^\star_N(x) D_N
    &=  X_N(x) D_N P^\star_\sigma F^\star_N(x_\sigma) P^\star_{\sigma^{-1}} R^\star_\sigma D_N \\
    &=   P_\sigma X_N(x_\sigma) D_N F^\star_N(x_\sigma) P^\star_{\sigma^{-1}} R^\star_\sigma D_N\\
    &=   P_\sigma X_N(x_\sigma) D_N F^\star_N(x_\sigma) D_N P_{\sigma^{-1}} R_\sigma =   P_\sigma   \tilde F_N(x_\sigma) P_{\sigma^{-1}} R_\sigma 
\end{align*}
where in the first line we used \eqref{eq:F-prop-super_graded}, in the second and third lines we used the following equations:
\begin{align}\label{eq:PR-star_esq}
    X_N(x) D_N P^\star_\sigma =P_\sigma X_N(x_\sigma) D_N ,
    \qquad
    P^\star_{\sigma^{-1}} R^\star_\sigma D_N = 
    D_N P_{\sigma^{-1}} R_\sigma
\end{align}
and the final equality is due to \eqref{eq:Ftilde-Fstar}. In the first equation in \eqref{eq:PR-star_esq} the antisymmetry of the factors $x_i-x_j$ in $X_N$ plays an important role.

\subsection{Proof of Lemma \ref{lem:F-prop-super}}\label{app:proof_F-lemma}
The $F$-matrix $F_N$, given in \eqref{eq:F-matrix-super}, contains sectors labelled by compositions of non-negative integers $(l_0,l_1\ldots l_{n+m})$,  such that $N=l_0+\cdots +l_{n+m}$. Moreover, it is also lower-triangular. Due to our labeling of the basis vectors these two aspects imply that the vectors $\ket{\lambda^+}$ and covectors $\bra{\lambda^-}$, with:
\begin{align}\label{eq:lambda_pm_appendix}
\lambda^- :=     
(0^{l_0}1^{l_1}\ldots (n+m)^{l_{n+m}}), 
\qquad
\lambda^+ := 
((n+m)^{l_{n+m}}\ldots 1^{l_1} 0^{l_0})
\end{align}
must be right end left eigenvectors of $F_N$ respectively. Therefore in order to prove \eqref{eq:F-diag-super} we need to compute the matrix elements:
\begin{align}
\bra{\lambda^\pm}F_N(x)\ket{\lambda^\pm}
\end{align}
From \eqref{eq:F-matrix-super} it follows that we can write all matrix elements of the $F$-matrix in terms of matrix elements of $R_\sigma$:
\begin{align}\label{eq:F-elements}
\bra{\lambda}F_N(x)\ket{\mu} =  \prod_{\substack{1\leq i<j\leq N\\ \lambda_{\sigma(i)}=\lambda_{\sigma(j)}>n}} \frac{x_{\sigma(i)}+x_{\sigma(j)}}{x_{\sigma(i)}-t x_{\sigma(j)}}
\cdot 
\bra{\lambda}R_\sigma(x) \ket{\mu}
\end{align}
where $\sigma$ is any permutation such that: 
\begin{align}\label{eq:lambda_ord}
\lambda_{\sigma(1)}\leq \cdots \leq \lambda_{\sigma(N)}
\end{align}
By choosing $\lambda=\mu=\lambda^-$ in \eqref{eq:F-elements} and $\sigma$ to be the identity permutation we get:
\begin{align}
    \label{eq:ev-lambda_minus}
    \bra{\lambda^-}F_N(x)\ket{\lambda^-}=
    \prod_{\substack{1\leq i<j\leq N\\ \lambda_i^{-}=\lambda_j^{-}>n}} \frac{x_i+x_j}{x_i-t x_j}
\end{align}
Next we set $\lambda =\mu=\lambda^+$ in \eqref{eq:F-elements}. The permutation $\sigma$ must be chosen such that:
$$
\sigma(\lambda^+)= \lambda^-
$$
We select $\sigma$ in a way that does not involve permutations of labels of $\lambda^+$ which are the same. This means that $\sigma$ needs to consist of factors of the form $\sigma_{(j,k)}$ which perform local cyclic permutations:
$$
\sigma_{(j,k)}(\ldots j^{l_j} k^{l_k}\ldots) = 
(\ldots k^{l_k} j^{l_j}\ldots),
\qquad j>k
$$
leaving the labels in $\ldots$ untouched. We use $R_\sigma =P_\sigma \check R_\sigma$ (see \eqref{eq:check_R-R_sigma}) and compute:
\begin{align}
    \label{F-lambdaplusX}
\bra{\lambda^+}F_N(x)\ket{\lambda^+}=
\prod_{\substack{1\leq i<j\leq N\\ \lambda_{\sigma(i)}^{+}=\lambda_{\sigma(j)}^{+}>n}} \frac{x_{\sigma(i)}+x_{\sigma(j)}}{x_{\sigma(i)}-t x_{\sigma(j)}}
\cdot 
\bra{\lambda^+}P_\sigma \check R_\sigma \ket{\lambda^+} = 
\prod_{\substack{1\leq i<j\leq N\\ \lambda_{i}^{+}=\lambda_{j}^{+}>n}} \frac{x_{i}+x_{j}}{x_{i}-t x_{j}}
\cdot 
\bra{\lambda^-}\check R_\sigma \ket{\lambda^+}
\end{align}
For a demonstration of how $\bra{\lambda^-}\check R_\sigma \ket{\lambda^+}$ is computed we consider the case $n+m=2$. Using the graphical notation for $\check R_i$ from \eqref{eq:R_i-graph} we have:
\begin{align}\label{eq:min_R_pl}
\bra{\lambda^-}\check R_\sigma \ket{\lambda^+}    =
\begin{tikzpicture}[scale=0.8, baseline=0.5cm]
\node[above] at (0,6-2.5) {$\scriptstyle 2$};
\node at (0.5,6-2.5) {$\ldots$};
\node[above] at (1,6-2.5) {$\scriptstyle 2$};
\node[above] at (2,6-2.5) {$\scriptstyle 1$};
\node at (2.5,6-2.5) {$\ldots$};
\node[above] at (3,6-2.5) {$\scriptstyle 1$};
\node[above] at (4,6-2.5) {$\scriptstyle 0$};
\node at (4.5,6-2.5) {$\ldots$};
\node[above] at (5,6-2.5) {$\scriptstyle 0$};
\draw[arrow=0.05] (0,6-2.5) -- (2,4-2.5) -- (4,2-2.5) -- (4,0-2.5);
\draw[arrow=0.05] (1,6-2.5) -- (3,4-2.5) -- (5,2-2.5) -- (5,0-2.5);
\draw[arrow=0.05] (2,6-2.5) -- (0,4-2.5) -- (0,2-2.5) -- (2,0-2.5);
\draw[arrow=0.05] (3,6-2.5) -- (1,4-2.5) -- (1,2-2.5) -- (3,0-2.5);
\draw[arrow=0.05] (4,6-2.5) -- (4,4-2.5) -- (2,2-2.5) -- (0,0-2.5);
\draw[arrow=0.05] (5,6-2.5) -- (5,4-2.5) -- (3,2-2.5) -- (1,0-2.5);
\node[below] at (0,0-2.5) {$\scriptstyle 0$};
\node at (0.5,0-2.5) {$\ldots$};
\node[below] at (1,0-2.5) {$\scriptstyle 0$};
\node[below] at (2,0-2.5) {$\scriptstyle 1$};
\node at (2.5,0-2.5) {$\ldots$};
\node[below] at (3,0-2.5) {$\scriptstyle 1$};
\node[below] at (4,0-2.5) {$\scriptstyle 2$};
\node at (4.5,0-2.5) {$\ldots$};
\node[below] at (5,0-2.5) {$\scriptstyle 2$};
\end{tikzpicture}
\end{align}
The only possible local configurations in the above diagram correspond to the fourth vertex in the first row and third vertex in the second row in \eqref{tikz:coloredvertices} such that the labels at the top follow the lines. By computing the Boltzmann weights we arrive at:
$$
\bra{\lambda^-}\check R_\sigma \ket{\lambda^+}=
\prod_{\substack{1\leq i<j \leq N\\ \lambda^+_i\neq \lambda^+_j}}
t^{\delta_{\lambda^+_j>0}}\frac{x_i-x_j}{x_i-t x_j}
$$
Combining this with \eqref{F-lambdaplusX} gives us the eigenvalue:
\begin{align}
    \label{F-lambdaplus}
\bra{\lambda^+}F_N(x)\ket{\lambda^+}=
\prod_{\substack{1\leq i<j\leq N\\ \lambda_i^{+}=\lambda_j^{+}>n}} \frac{x_i+x_j}{x_i-t x_j}
\cdot 
\prod_{\substack{1\leq i<j \leq N\\ \lambda^+_i\neq \lambda^+_j}}
t^{\delta_{\lambda^+_j>0}}\frac{x_i-x_j}{x_i-t x_j}
\end{align}
The expressions \eqref{eq:ev-lambda_minus} and \eqref{F-lambdaplus} are precisely the eigenvalues appearing in \eqref{eq:F-diag-super} and the Lemma follows.

\section*{Acknowledgments}
We would like to thank Jan de Gier, Andrei Negu\c{t}, Michael Wheeler and Paul Zinn-Justin for many interesting discussions. A. G. gratefully acknowledges financial support from the Australian Research Council.

\bibliographystyle{plain}
\bibliography{literature}{}
\end{document}